%% file: Semi-random_Stability-arXiv.tex
\documentclass[acmsmall]{acmart}
\usepackage{booktabs} 

\usepackage[most]{tcolorbox}

\usepackage[ruled]{algorithm2e} 
\renewcommand{\algorithmcfname}{ALGORITHM}
\SetAlFnt{\small}
\SetAlCapFnt{\small}
\SetAlCapNameFnt{\small}
\SetAlCapHSkip{0pt}
\IncMargin{-\parindent}

\setcitestyle{acmnumeric}

\input{macro.tex}

\title{The Impact of a Coalition: Assessing the Likelihood of Voter Influence in Large Elections} 
\author{
Lirong Xia
} 
\affiliation{
\institution{RPI}
\city{Troy}
\state{NY}
  \country{USA}
  }
\email{xialirong@gmail.com}
\begin{document}


\setcounter{page}{0}
\begin{abstract}
For centuries, it has been widely believed  that the influence of a small coalition of voters is negligible in a large election. Consequently, there is a large body of literature on characterizing the  likelihood for an election to be influenced when the votes follow certain distributions, especially  the likelihood of being manipulable by a single voter under the i.i.d.~uniform distribution, known as the Impartial Culture (IC).

In this paper, we extend previous studies in three aspects: (1) we propose a more   general   semi-random model, where a  distribution adversary  chooses a worst-case distribution and then a   contamination adversary modifies up to $\psi$ portion of the data, (2) we consider many  coalitional influence  problems, including coalitional manipulation, margin of victory, and various vote controls and bribery, and (3) we consider arbitrary and variable coalition size $B$. Our main theorem provides asymptotically tight bounds on the semi-random likelihood of the existence of a size-$B$ coalition that can successfully influence the election under a wide range of voting rules. Applications of the main theorem and its proof techniques resolve long-standing open questions about the likelihood of coalitional manipulability under IC, by showing that the likelihood is $\Theta\left(\min\left\{\frac{B}{\sqrt n}, 1\right\}\right)$ for many commonly-studied voting rules.

The main technical contribution is a characterization of the semi-random likelihood for a  Poisson multinomial variable (PMV) to be unstable, which we believe to be a general and useful technique with independent interest.

\end{abstract}
\begin{titlepage}

\maketitle
 
\end{titlepage}
\setcounter{page}{1} 
 


\section{Introduction}

For centuries, it has been widely believed  that the influence of a small coalition of voters, especially the special case of a single voter, is negligible in a large election. For example, Condorcet commented in 1785, that {\em ``In single-stage elections, where there are a great many voters, each voter's influence is very small''~\citep{Condorcet1785:Essai}} (see the translation by~\citet[pp.245]{McLean1994:Condorcet}). As another example, Hegel commented in  {\em The Philosophy of Right} in 1821, that  {\em ``the casting of a single vote is of no significance where there is a multitude of electors''} (see the translation and comments by~\citet{Buchanan1974:Hegel}).  

How small the influence is? The answer depends on the definition and measure of influence. Many types of influence were defined and investigated in the literature. For example, {\em coalitional manipulation ($\cm$ for short)} refers to the phenomenon in which a coalition of voters have incentive to misreport their preferences to make the winner  more favorable to all of them. The {\em margin of victory ($\mov$ for short)} is the size of the smallest  coalition of voters who have the power to change the winner by voting differently, regardless of their incentives.  

Clearly, a sensible and informative measure cannot be completely based on the worst-case analysis. Take $\cm$ with a single manipulator  for example:  for many voting rules, there  exists a situation where a single voter can and has incentive to change the election outcome by changing his/her vote, due to the celebrated Gibbard-Satterthwaite theorem~\citep{Gibbard73:Manipulation,Satterthwaite75:Strategy}. Similarly, unless the voting rule always chooses the same winner, there exists a situation where a single voter is highly influential, i.e., $\mov=1$.


Consequently, there is a large body of literature on understanding the influence of voters using {\em average-case} analysis, especially for $\cm$. Since the pioneering works by~\citet{Pazner1978:Cheatproofness} and~\citet{Peleg1979:Note} in the 1970's, much of the literature has focused on characterizing the asymptotic likelihood for randomly generated votes to be coalitional manipulable for a fixed number of alternatives,  as the number of voters $n\ra\infty$. 
Previous work has established an $O(\frac{1}{\sqrt n})$ upper bound for many commonly-studied voting rules by  a coalition of constantly many manipulators, under the i.i.d.~uniform distribution, known as the {\em Impartial Culture (IC)} in social choice. As for the lower bound, for a single manipulator w.r.t.~IC,  \citet{Slinko2002:Asymptotic} proved an $\Omega(\frac{1}{\sqrt n})$  lower bound  under the plurality rule, and the quantitative Gibbard-Satterthwaite theorems, e.g.,~\citep{Friedgut2011:A-quantitative,Mossel2015:A-quantitative}, established an $\Omega(\frac{1}{n^{67}})$ lower bound for a single manipulator under all voting rules that are constantly far away from dictatorships. \citet{Xia2023:Semi-Random}  proved an $\Omega(\frac{B}{\sqrt n})$ lower bound for an arbitrary coalition size $1\le B\le \sqrt n$ under any Condorcet consistent rule. Nevertheless, no matching bound was known for any non-Condorcet-consistent rules (except plurality) with $B=1$, or for any rule with a variable coalition size $B$. See Section~\ref{sec:related-work} for more related work and discussions.  The key research question is:
\begin{center}
{\bf How likely a coalition of voters can influence a large election?}
\end{center}
The significance of addressing this question has been widely recognized, as~\citet[p.187]{Pattanaik1978:Strategy} discussed for (coalitional) manipulation soon after the discovery of the Gibbard-Satterthwaite theorem:
{\em ``For, if the likelihood of such strategic voting is negligible, then one need not be unduly worried about the existence of the possibility as such.''} The answer to this question  also  plays a central role in many other studies, and a small influence can be positive news or negative news depending on the context. For example, a small influence is desirable under various robustness measures, such as decisiveness and privacy~\citep{Liu2020:How}.   On the other hand, a  large influence is desirable in justifying the {\em voting power} of small groups of voters~\citep{Banzhaf-III1968:One-Man}, and therefore would encourage voter turnout~\citep{Downs1957:An-Economic,Riker1968:A-Theory}.  

Surprisingly, the question remains largely open despite its significance and long history. Little was known for variable coalition size $B$, other means of influence, such as margin of victory ($\mov$), vote controls, and bribery~\citep{Faliszewski2016:Control}, and/or other probabilistic models for generating votes.  Previous work faces two challenges: first, IC has been widely criticized of being unrealistic (see, e.g.,  \citep[p.~30]{Nurmi1999:Voting}, \citep[p.~104]{Gehrlein2006:Condorcets}, and~\citep{Lehtinen2007:Unrealistic}), which means that the conclusions drawn under IC may only have limited implications in practice.  Second, technically, existing tools for likelihood analysis, especially various central limit theorems, are too coarse due to an $O(\frac{1}{\sqrt n})$ error~\citep{Xia2020:The-Smoothed}. Accurately bounding the likelihood is desirable, because the $O(\frac{1}{\sqrt n})$ upper bound in previous work may still be too large, especially when an influenceable election leads to a high social cost, such as the wide spread  of negative news in a well-connected society~\citep{Xia2023:Semi-Random}. 

\subsection{Our Contributions}
We make three contributions to address the question on the likelihood of coalitional influence: our  main conceptual contribution is a novel and general semi-random model for generating votes; our main modeling contribution is the conversion of commonly-studied coalitional influence problems to  {\em PMV-instability problems}; and our main technical contribution is an accurate answer to  general PMV-instability problems, which allows us to accurately characterize the likelihood of many coalitional influence problems.

\myparagraph{Conceptual contribution: The $(\Pi,\psi)$-semi-random model.} Our model combines and generalizes two models in the literature. The first is the semi-random model for social choice in~\citep{Xia2020:The-Smoothed}, which was inspired by and resembles the celebrated {\em smoothed analysis}~\cite{Spielman2004:Smoothed}; and the second is the celebrated {\em contamination model}  in robust statistics~\citep{Huber1964:Robust} and robust machine learning~\citep{Diakonikolas2021:Robust}. Our model has two parameters $(\Pi,\psi)$, where $\Pi$ is a set of distributions over all rankings over  $m$ alternatives and $\psi\in [0,1]$, and is therefore called the {\em $(\Pi,\psi)$-semi-random model}. Under our model, the collection of votes for $n$ agents, called a {\em profile}, is generated in the following three steps.
\begin{itemize}
\item Step 1: as in~\citep{Xia2020:The-Smoothed}, a {\em distribution adversary}  chooses a (worst-case) distribution for each agent. The distributions are denoted by $\vec \pi\in \Pi^n$. 
\item Step 2: the vote for each agent is independently (but not necessarily identically) generated from his/her distribution in $\vec\pi$. The votes are denoted by $P'$ and is called a   {\em tentative profile}.
\item Step 3: as in~\citep{Huber1964:Robust,Diakonikolas2021:Robust}, a  {\em contamination adversary} modifies $P'$ up to $\psi n$ votes to obtain a profile $P$. 
\end{itemize}


\myparagraph{Modeling contribution: The PMV-instability problem.} Given a $(\Pi,\psi)$-semi-random model and a coalitional influence problem $X$, a voting rule $r$, and a coalition size $B$, we define the {\em max-semi-random likelihood} of $X$, denoted by $\satmax{X}{\Pi,\psi}(r,n,B)$,  as:
\begin{equation}
\label{dfn:max-s-sat}
\satmax{X}{\Pi,\psi}(r,n,B) \triangleq  \underbrace{\sup\nolimits_{\vec\pi\in\Pi^n}}_{\begin{subarray}{c} \text{distribution}\\ \text{adversary}\end{subarray}}\ \underbrace{\Pr\nolimits_{P'\sim\vec\pi}}_{\begin{subarray}{c} \text{tentative}\\ \text{profile}\end{subarray}}\ (\underbrace{\exists P\overset\psi\approx P'}_{\begin{subarray}{c} \text{contamination}\\\text{adversary}\end{subarray}}\text{ s.t. }\underbrace{\sat{X}(r,P,B)=1}_{\begin{subarray}{c} 
\text{can be influenced}\\ \text{with budget } B\end{subarray}}),
\end{equation}
where $P \overset\psi\approx P'$ means that $P$  can be obtained from $P'$ by modifying no more than $\psi n$ votes, and $\sat{X}(r,P,B)=1$ if  $P$ can be influenced with budget $B$; otherwise $\sat{X}(r,P,B)=0$. For example, when  $X = \cm$, $\cm(r,P,B)=1$ if and only if there exist a coalition of no more than $B$ voters who have incentive to misreport their preferences.


Notice that the histogram of the tentative profile in our model is a {\em Poisson multivariate variable (PMV)}~\citep{Xia2021:How-Likely}. We model many coalitional influence problems under commonly-studied voting rules, including $\cm$ and $\mov$, as {\em PMV-instability problems} (Definition~\ref{dfn:PMV-instability-problem}), which aim to characterize the likelihood for a PMV to be {\em unstable}, in the sense that after modifying no more than $\psi n$ votes, the resulting histogram is in a given ``source'' polyhedron $\sourcepoly$, and can be further influenced to be in a ``target'' polyhedron $\targetpoly$ under budget $B$.

\myparagraph{Technical contribution: Characterizing the likelihood.} We address the   PMV-instability problem  under a large class of $(\Pi,\psi)$-semi-random models as follows. 
\begin{itemize}
\item {\bf \boldmath When  $\psi = 0$ (Theorem~\ref{thm:PMV-instability}),} the  solution to the  PMV-instability problem  is
$$0,  \exp(-\Theta(n)) ,\text{phase transition at $\Theta(\sqrt n)$, or phase transition at $\Theta(n)$}$$ 
The two phase transition cases refer to relatively sharp in crease of the likelihood as a function of $B$. More precisely, in the phase-transition-at-$\Theta(\sqrt n)$ case, the  likelihood reaches its maximum $\text{poly}^{-1}(n)$ before $B=\Theta(\sqrt n)$. In the phase-transition-at-$\Theta( n)$ case, the  likelihood increases from $\exp(-\Theta(n))$ to its maximum $\text{poly}^{-1}(n)$ around $B=\Theta(n)$.

\item {\bf \boldmath When  $\psi > 0$ (Theorem~\ref{thm:PMV-instability-psi>0}),} the  solution to the  PMV-instability problem is
$$0,  \exp(-\Theta(n)) ,\text{ or } \Theta(1)$$
\end{itemize}
The formal statements of Theorem~\ref{thm:PMV-instability} and Theorem~\ref{thm:PMV-instability-psi>0} also characterize  conditions and asymptotically tight bounds on the likelihood. The two theorems are useful technical tools for accurately bounding semi-random likelihood of coalitional influence, as shown in the following three applications. 
\begin{itemize}
\item {\bf First (Theorem~\ref{thm:PMV-instability-applications}),} we  prove   that for $X=\cm$ and $X=\mov$, many commonly-studied voting rules $r$, a large class of $\Pi$ including IC, any coalition size $B\ge 1$, and any sufficiently large $n$, 
$$\satmax{X}{\Pi,0}(r,n,B) =  \Theta\left(\min\left\{\frac{B}{\sqrt n},1\right\}\right)$$
Moreover, for every $\psi>0$, $\satmax{X}{\Pi,\psi}(r,n,B) =  \Theta\left(1\right)$. A straightforward application of  Theorem~\ref{thm:PMV-instability-applications} to $\cm$ under IC (Corollary~\ref{coro:IC}) not only closes the previous $\left(\Omega(\frac{1}{n^{-67}}), O(\frac{1}{\sqrt n})\right)$ gap for integer positional scoring rules and STV with $B=1$,  but also provides asymptotically tight bounds  under many rules for every $B\ge 1$: roughly speaking, each additional manipulator (up to $O(\sqrt n)$) increases the likelihood of success by $\Theta(\frac{1}{\sqrt n})$. This is good news when the social cost of dealing with the manipulation is low, but it is bad news if the social cost is high. In the latter case, voting rules with lower likelihood of manipulability are desirable.
\item {\bf Second (Theorem~\ref{thm:PMV-instability-GSR-upper}),} we prove an $O\left(\min\left\{\frac{B}{\sqrt n},1\right\}\right)$ upper bound on the likelihood of many types of coalitional influence, including $\cm$ and $\mov$, for all {\em generalized scoring rules (GSRs)}~\citep{Xia08:Generalized}, which include all voting rules mentioned in this paper. This supersedes all previous upper bounds we are aware of  and extends them to arbitrary $B\ge 1$ and a more general (semi-random) model. 
\item {\bf Third (Theorem~\ref{thm:CML}),} we propose  a new  coalitional influence problem called {\em coalitional manipulation for the loser}, denoted by $\cml$, which requires that a coalition of voters are incentivized to misreport their preferences in order to make the loser win. We prove that for any integer positional scoring rules with $m\ge 3$ (except veto), any $\Pi$ from a large class, and any sufficiently large $n$ and $B$, 
$$\satmax{\cml}{\Pi,0}(r_{\vec s} ,n,B) =  \Theta\left(\min\left\{\frac{B}{\sqrt n},1\right\}^{m-1}\right) $$
Moreover, for every $\psi>0$, $\satmax{\cml}{\Pi,\psi}(r_{\vec s},n,B) =  \Theta\left(1\right)$.  While $\cml$ may be of independent interest, the main purpose of this result is to illustrate that the likelihood of coalitional influence can be much smaller than $\Theta(\frac{B}{\sqrt n})$ or even $\Omega(\frac{1}{n^{-67}})$, can be non-linear in $B$, the degree of polynomial can depend on $m$, and each additional budget is marginally more powerful.  We are not aware of a similar phenomenon in the literature.
\end{itemize} 

\section{Preliminaries}
\label{sec:prelim}
For any  $q\in\mathbb N$,  let $[q]=\{1,\ldots,q\}$. Let $\ma=[m]$ denote a set of $m\ge 3$ {\em alternatives}. Let $\ml(\ma)$ denote the set of all linear orders over $\ma$. Let $n\in\mathbb N$ denote the number of agents (voters). Each agent uses a linear order $R\in\ml(\ma)$ to represent his or her preferences, called a {\em vote}, where $a\succ_R b$ or $\{a\}\succ_R\{b\}$ means that the agent prefers alternative $a$ to alternative $b$. The vector of $n$ agents' votes, denoted by $P$, is called a {\em (preference) profile}, sometimes called an $n$-profile. The set of $n$-profiles for all $n\in\mathbb N$ is denoted by $\ml(\ma)^* = \bigcup_{n =1}^{\infty} \ml(\ma)^n$. 

For any   profile $P$, let $\hist(P)\in {\mathbb R}_{\ge 0}^{m!}$ denote the anonymized profile of $P$, also called the {\em histogram} of $P$, which contains the number of occurrence of every  linear order in $\ml(\ma)$ according to $P$.  An {\em irresolute voting rule} $\cor:\ml(\ma)^*\ra (2^{\ma}\setminus \{\emptyset\})$ maps a profile to a non-empty set of winners in $\ma$. A {\em resolute} voting rule $r$ is a special irresolute voting rule that always chooses a single alternative as the (unique) winner. Often a resolute rule is obtained from an irresolute rule by applying a {\em tie-breaking mechanism}. For example, the lexicographic tie-breaking  chooses the co-winner with the smallest index as the unique winner.   
 

\vspace{2mm}\noindent{\bf Integer positional scoring rules.}  An {\em (integer) positional scoring rule}  $\cor_{\vec s}$  is  characterized by an integer scoring vector $\vec s=(s_1,\ldots,s_m)\in{\mathbb Z}^m$ with $s_1\ge s_2\ge \cdots\ge s_m$ and $s_1>s_m$. For any alternative $a$ and any linear order $R\in\ml(\ma)$, we let  $\vec s(R,a)=s_i$, where $i$ is the rank of $a$ in $R$.  Commonly-studied integer positional scoring rules include  {\em plurality}, which uses the scoring vector $(1,0,\ldots,0)$, {\em Borda}, which uses the scoring vector $(m-1,m-2,\ldots,0)$, and {\em veto}, which uses the scoring vector $(1,\ldots,1,0)$. 

We recall the definition of {\em generalized scoring rules (GSRs)}~\citep{Xia08:Generalized} based on separating hyperplanes~\citep{Xia09:Finite,Mossel13:Smooth} as follows. 

\begin{dfn} 
\label{dfn:GISR} A {\em generalized scoring rule (GSR)} $r$ is defined by (1) a set of $K\ge 1$ hyperplanes $\vH = (\vec h_1,\ldots,\vec h_K)\in ({\mathbb R}^{m!})^K$ and (2) a function $g:\{+,-,0\}^K\ra \ma$. For any  profile $P$, we let $r(P) = g(\sign_\vH(\hist(P)))$, where 
$\sign_\vH(\vec x) = (\sign(\vec h_1\cdot \vec x),\ldots, \sign( \vec h_K\cdot \vec x))$ represents the signs of $ \vec h_1\cdot \vec x,\ldots,  \vec h_K\cdot \vec x$.    When ${\vH}\in ({\mathbb Z}^{m!})^K$,  $r$ is called an {\em integer GSR (int-GSR)}.  
\end{dfn}

\begin{ex}[{\bf\boldmath Borda as a  GSR}{}]
\label{ex:borda-GSR}
Let $m=3$. Borda with lexicographic tie-breaking is  a GSR with $K=m$ and $\vec H = \{\vec h_1,\vec h_2,\vec h_3\}$  defined as follows. 
\begin{center}
\begin{tabular}{rrrrrrrrl}
& & $x_{123}$ & $x_{132}$ & $x_{213}$ & $x_{231}$ & $x_{312}$ & $x_{321}$ \\
$\vec h_1= $ & $($ & $1,$& $2,$& $-1,$& $-2,$& $1,$& $-1$& $)$\\
$\vec h_2= $ & $($ & $2,$& $1,$& $1,$& $-1,$& $-1,$& $-2$& $)$\\
$\vec h_3= $ & $($ & $1,$& $-1,$& $2,$& $1,$& $-2,$& $-1$& $)$\\
\end{tabular}
\end{center}
Let $\vec x = (x_{123},x_{132},x_{213},x_{231},x_{312},x_{321})$ denote the histogram of a profile, where $x_{123}$ represents the number of $[1\succ 2\succ 3]$ votes. It follows that $\vec h_1\cdot \vec x$ is the  Borda score of alternative $1$ minus the  Borda score of $2$ in the profile; $\vec h_2\cdot \vec x$ is the  Borda score of alternative $1$ minus the  Borda score of $3$ in the profile; and $\vec h_3\cdot \vec x$ is the  Borda score of alternative $2$ minus the  Borda score of $3$ in the profile. The $g$ function chooses the winner based on $\signH(\hist(P))$ and break ties lexicographically when necessary.\end{ex}




\myparagraph{\bf Coalitional manipulation and margin of victory.}  Let $r$ be a resolute rule, $P$ be a preference profile, and $B\ge 0$ be a  budget. A {\em coalitional influence problem} is defined by a binary function $X$ such that $X(r,P,B)=1$ if and only if $P$ can be influenced under $r$ with budget $B$. For example, {\em coalitional manipulation}  ($\cm$)  is defined by $\cm(r,P,B)$ such that $\cm(r,P,B)=1$ if and only there exists $P'\subseteq P$ with $|P'|\le B$ and $P^*$ with $|P^*|=|P'|$,  such that for all $R\in P'$, $r(P-P'+P^*)\succ_R r(P)$. The {\em margin of victory} ($\mov$) is defined by $\mov(r,P,B)$ such that $\mov(r,P,B)=1$ if and only if there exist a coalition of no more than $B$ voters who can  change the winner by voting differently (regardless of their preferences and incentives). 

\myparagraph{\bf Other rules and coalitional influence problems.} In the main text we focus on $X=\cm$ and $r=\borda$.  See Appendix~\ref{app:more-rule} for the definitions of some other commonly-studied voting rules (which are GSRs), i.e., ranked pairs, Schulze, maximin, Copeland, and STV, and  Appendix~\ref{app:more-CI}  for the definitions of  some other commonly-studied coalitional influence problems, i.e., constructive/destructive control by adding/deleting votes and bribery. Many results in the main text apply to these rules and coalitional influence problems, as stated in their full versions in the Appendix.


\section{\boldmath The $(\Pi,\psi)$-Semi-Random Model}
Before introducing our model, let us first briefly recall the definitions of two semi-random models in previous work. 
In the semi-random model in~\citep{Xia2020:The-Smoothed}, a set $\Pi$ of distributions over $\ml(\ma)$, a voting rule $r$, and a coalitional influence problem $X$ are given. Then, a {\em distribution adversary} chooses a distribution from $\Pi$ for each agent, whose vote is generated independently (but may not identically), to maximize the probability for $X$ to be $1$. Formally, under this model we are interested in characterizing
$${\sup\nolimits_{\vec\pi\in\Pi^n}}  {\Pr\nolimits_{P\sim\vec\pi}} (\sat{X}(r,P,B)=1)$$
In the contamination model~\citep{Huber1964:Robust,Diakonikolas2021:Robust},   $\psi\in[0,1]$ and a distribution $\pi$ over $\ml(\ma)$ are given, and a {\em tentative profile} $P'$ is generated i.i.d.~from $\pi$. After seeing $P'$, a {\em contamination adversary} modifies up to $\psi n$ votes in $P'$ to satisfy $X$. Formally, under this model we are interested in characterizing
$$  {\Pr\nolimits_{P'\sim\pi^n}} (\exists P\overset\psi\approx P' \text{ s.t. } \sat{X}(r,P,B)=1)$$
We propose a semi-random model that combines the two aforementioned models sequentially: first, a distribution adversary chooses $\vec\pi\in\Pi^n$; second, a tentative profile $P'$ is generated; and third, a contamination adversary modifies up to $\psi n$ votes in $P'$ to obtain $P$. Formally, we have the following definition.
\begin{dfn}[{\bf \boldmath The $(\Pi,\psi)$-semi-random model and likelihood of coalitional influence}{}]
\label{dfn:Pi-psi-model}
Given a voting rule $r$,  a budget $B\ge 0$, an influence problem $X$, $\Pi$, and $\psi\in [0,1]$, the {\em max-semi-random likelihood} of $X$ under $r$ with $n$ agents and budget $B$, denoted by $\satmax{X}{\Pi,\psi}(r,n,B)$,  is defined as:
$$\hspace{15mm}
\satmax{X}{\Pi,\psi}(r,n,B) \triangleq  \underbrace{\sup\nolimits_{\vec\pi\in\Pi^n}}_{\begin{subarray}{c} \text{distribution}\\ \text{adversary}\end{subarray}}\ \underbrace{\Pr\nolimits_{P'\sim\vec\pi}}_{\begin{subarray}{c} \text{tentative}\\ \text{profile}\end{subarray}}\ (\underbrace{\exists P\overset\psi\approx P'}_{\begin{subarray}{c} \text{contamination}\\\text{adversary}\end{subarray}}\text{ s.t. }\underbrace{\sat{X}(r,P,B)=1}_{\begin{subarray}{c} 
\text{can be influenced}\\ \text{with budget } B\end{subarray}}) \hspace{15mm}\eqref{dfn:max-s-sat}
$$
\end{dfn}

Notice that while technically, the goal of the distribution adversary and the contamination adversary is to maximize $X(r,P,B)$, they are ``virtual'' and are only used to model the worst-case nature in the data generation process. On the other hand, the influencer(s) are real person who aim at influencing the election after the votes are generated. Therefore, in this paper we model them separately, especially because we believe that the $(\Pi,\psi)$-semi-random model can be useful for studying the likelihood of other properties of interest that do not have an ``influencer'', such as likelihood of Condorcet's voting paradox and the ANR impossibility theorem~\citep{Xia2020:The-Smoothed,Xia2023:Semi-Random}. Allowing the influencers to conduct multiple means to influence the election is an interesting direction for future work. 

The $(\Pi,\psi)$-semi-random model generalizes popular models in previous work as shown in the following example.
\begin{ex}
\label{ex:model-special-cases}
The semi-random model in~\citep{Xia2020:The-Smoothed} is equivalent to the $(\Pi,\psi)$-semi-random model with $\psi = 0$. The contamination model~\citep{Huber1964:Robust,Diakonikolas2021:Robust} is equivalent to the $(\Pi,\psi)$-semi-random model with $|\Pi|=1$. Impartial Culture (IC) is equivalent to the $(\{\piuni\},0)$-semi-random model, where $\piuni$ is the uniform distribution over $\ml(\ma)$.
\end{ex}

 Let us take a look at another example of the $(\Pi,\psi)$-semi-random model.

\begin{ex}[{\bf\boldmath Semi-random $\cm$ under Borda}{}]
\label{ex:sCM-Borda}
Let $X=\cm$ and $r=\borda$ with lexicographic tie-breaking. Let $\ma = \{1,2,3\}$ and  $\Pi = \{\pi^1,\pi^2\}$, where $\pi^1$ and $\pi^2$ are distributions in Table~\ref{tab:sCM-Borda}.

\begin{minipage}[t][][b]{1\textwidth}
\centering
\begin{tabular}{|@{\ }c@{\ }|c|c|c|c|c|c|c| }
\hline & \small $1\succ 2\succ 3$& \small $1\succ 3\succ 2$& \small $2\succ 3\succ 1$& \small $3\succ 2\succ 1$& \small $2\succ 1\succ 3$& \small $3\succ 1\succ 2$ \\

\hline $\pi^1$& $1/4$& $1/4$&$1/4$& $1/12$& $1/12$&$1/12$  \\

\hline $\pi^2$& $1/12$& $1/12$&$1/12$ & $1/4$& $1/4$&$1/4$ \\

\hline
\end{tabular} 
\captionof{table}{\small $\Pi$ in Example~\ref{ex:sCM-Borda}.\label{tab:sCM-Borda}}
\end{minipage}

When $n=2$ and $B=1$,  $\satmax{\cm}{\Pi,0}(\borda,2,1) = \sup\nolimits_{\vec\pi\in\{\pi^1,\pi^2\}^n}\Pr\nolimits_{P\sim\vec\pi}\sat{\cm}(\borda,P,1)$. That is, the adversary has four choices of $\vec \pi$, i.e., $\{(\pi^1,\pi^1),(\pi^1,\pi^2),(\pi^2,\pi^1),(\pi^2,\pi^2)\}$. Each $\vec\pi$ leads to a distribution over the set of all  profiles of two agents, i.e., $\ml(\ma)^2$.  As we will see later in Example~\ref{ex:thm-Borda-CM}, for every sufficiently large $n$,  $\satmax{\cm}{\Pi,0}(\borda,n,1)= \Theta(\frac{1}{\sqrt n })$, and for every $\psi>0$, $\satmax{\cm}{\Pi,\psi}(\borda,n,1)= \Theta(1)$.  
\end{ex}
Following~\citep{Xia2020:The-Smoothed,Xia2021:How-Likely},  we make the following assumptions on $\Pi$ throughout the paper.

\begin{tcolorbox} {\bf \boldmath Assumptions on $\Pi$.}
(1) {\em Strict positiveness}, which means that there exists a constant $\epsilon>0$ such that the probabilities over all rankings in all $\pi\in\Pi$ are larger than $\epsilon$. (2) {\em Closedness}, which means that $\Pi$ is a closed set in the Euclidean space. 
\end{tcolorbox} 
For example, IC (where $\Pi=\{\piuni\}$ and $\psi =0$) and the model in Example~\ref{ex:sCM-Borda} satisfy the assumptions.

\myparagraph{\bf \boldmath Generality and limitations of the  model.} The $(\Pi,\psi)$-semi-random model generalizes several popular models as shown in Example~\ref{ex:model-special-cases}. It introduces strong (and adversarial) correlations among the votes compared to the semi-random model in~\cite{Xia2020:The-Smoothed}, therefore addressing a major limitation of the latter, which is the statistical independence of the votes. Compared to the contamination model~\citep{Huber1964:Robust,Diakonikolas2021:Robust}, the $(\Pi,\psi)$-semi-random model is more flexible and general in choosing the distribution for generating the tentative profile. The major limitation  is the strict positiveness  assumption of $\Pi$.  Nevertheless, as discussed in~\citep{Xia2020:The-Smoothed,Xia2021:How-Likely}, distributions in many commonly-studied and widely-used models are strictly positive, such as Mallows' model~\cite{Mallows57:Non-null} and random utility models~\citep{Thurstone27:Law}. Moreover, this limitation is partly addressed by the contamination adversary in our model. Therefore, while more general models are always desirable, we believe that the $(\Pi,\psi)$-semi-random model is a significant step forward, especially compared to the large body of literature that is based on IC, which is a special case of our model as well.

\section{Coalitional Influence Problems as PMV-Instability Problems}
\label{sec:model}
We start with an example of modeling $\cm$ with $\psi = 0$ under Borda as  systems of linear inequalities, which motivates the study of the more general {\em PMV-instability problem} that will be defined later in this section (Definition~\ref{dfn:PMV-instability-problem}). 
\begin{ex}
\label{ex:illustration}  
Let $\ma = \{1,2,3\}$ and let $r$ be  Borda  with lexicographic tie-breaking.  Let $\csus{n,B}^{1\ra2}$ denote the histograms of all $n$-profiles  that satisfy the following conditions: (1) the winner before manipulation is $1$, (2) a coalition of no more than $B$ manipulators are motivated to change the winner to $2$ by casting different votes. 
Let  $o_1$ (respectively, $o_2$) denote the number of voters who change their votes from $[2\succ 1\succ 3]$ (respectively, $[3\succ 2\succ 1]$) to $[2\succ 3\succ 1]$.  

Then, the histogram $\vec x$ of an $n$-profile is in $ \csus{n,B}^{1\ra2}$  if and only if there exists an integer vector $\vo = (o_1,o_2)$ such that $(\vec x,\vo\,)$ is a feasible solution to the following linear program, where the objective is omitted because  only the feasibility matters. 
\begin{align*}
\left.\begin{tabular}{r@{\ }r} $-x_{123}-2 x_{132} + x_{213}+2 x_{231} -  x_{312} + x_{321} \le $ &\phantom{$-$}$ 0$ \\
 $-2x_{123}- x_{132} - x_{213}+ x_{231} + x_{312}  + 2 x_{321} \le $ &$ 0$  \\
  $-\vec x \le $ &$ 0 $
\end{tabular}\right\} & \begin{tabular}{l} \text{1 wins before} \\ \text{manipulation}\end{tabular}\\
\left.\begin{tabular}{r@{\ }r}  
 $x_{123}+2 x_{132} - (x_{213}-o_1) - 2 (x_{231}+o_1+o_2) + x_{312} - (x_{321}-o_2) \le $ &$-1$ \\
 $-x_{123}+ x_{132} - 2(x_{213}-o_1) -  (x_{231}+o_1+o_2) + 2x_{312} +(x_{321}-o_2) \le $ &$ 0$\\
  $- o_1 \le 0$, $-o_2\le 0$,  $- (x_{213}-o_1) \le 0$, $- (x_{321}-o_2) \le 0$,
 $ o_1 + o_2 \le $ &$B$ 
\end{tabular}\right\}& \begin{tabular}{l} \text{2 wins after} \\ \text{manipulation}\end{tabular} 
\end{align*}
\end{ex}

 In Example~\ref{ex:illustration}, the effect of each manipulator can be modeled by its changes to the histogram, and the manipulators aim at manipulating  vectors in a {\em source} polyhedron, which represents $1$ being the winner, into a {\em target} polyhedron, which represents $2$ being the winner, under the budget constraint $B$.  This motivates us to define {\em instability settings} as follows.

 \begin{dfn}[{\bf \boldmath Instability settings}{}]
In an {\em instability setting} $\vosetting{}\triangleq\langle\sourcepoly, \targetpoly, \voset{}, \vec c\,\rangle$,   

$\bullet$ $\sourcepoly$ and $\targetpoly$ are polyhedra in $\mathbb R^{q}$ for some $q\in\mathbb N$, the subscript S and T represent ``source'' and  ``target'', respectively. For $Y\in \{\text{S},\text{T} \,\}$, let
 $\cpoly{Y} \triangleq \left\{\vec x\in {\mathbb R}^q: \ba_Y \times\invert{\vec x}\le \invert{\vbb_Y}\right\},$ 
where $\ba_{Y}$ is an integer matrix of $q$ columns; 

$\bullet$ $\voset{}\subseteq {\mathbb R}^{q}$ is a finite set of {\em vote operations}, and let $\vomatrix{}$ denote the $|\voset{}|\times q$ matrix whose rows are the vectors in $\voset{}$; 

$\bullet$ $\vec c \in \mathbb R_{\ge 0}^{|\voset{}|}$  is the cost vector for the vote operations in $\voset{}$.
\end{dfn} 
W.l.o.g.,  we assume  $\vec c>\vec 0$  and  the minimum cost of a single operation is $1$, i.e., $\min_{i\le q}[\vec c\,]_i = 1$. Given an instability setting $\vosetting{}$, $n\in\mathbb N$, and a budget $B \ge 0$, we  let $\csus{n,B}$ denote the set of non-negative size-$n$ integer vectors that represent {\em unstable} histograms w.r.t.~vote operations in $\voset{}$ and budget $B$. That is,
$$\csus{n,B} \triangleq \left\{\underbrace{\vec x\in  \sourcepoly\cap {\mathbb Z}_{\ge 0}^q}_{\vec x\text{ is in }\sourcepoly}: 
\underbrace{\vec x\cdot \vec 1 = n}_{n\text{-profile}}\text{ and }\exists \underbrace{\vo\in {\mathbb Z}_{\ge 0}^{|\voset{}|}}_{\text{vote operations}}\text{ s.t. }\underbrace{\vec c\cdot\vo \le B}_{\text{budget constraint}} \text{ and }\underbrace{\vec x+ \vo\times \vomatrix{}\in\targetpoly}_{\text{manipulated to be in }\targetpoly}\right\}$$

 
\begin{ex}
\label{ex:influence-model}
In the setting of Example~\ref{ex:illustration}, $q=m!=6$.  $\csus{n,B}^{1\ra 2}$ is the set of unstable histograms of the instability settings where $\sourcepoly$ (respective, $\targetpoly$) is the polyhedron that represents $1$ (respectively, $2$) being the winner, $\voset{} = \{(0,0,-1,1,0,0), (0,0,0,1,0,-1)\}$ (the indices to rankings are the same as in Example~\ref{ex:borda-GSR}), $\vomatrix{} = \left[\begin{tabular}{@{}r@{, }r@{, }r@{, }r@{, }r@{, }r@{}}0 & 0& -1& 1& 0& 0\\ 0 & 0& 0& 1& 0& -1\end{tabular}\right]$, $\vec c = (1,1)$.
\end{ex}

Notice that Examples~\ref{ex:illustration} and~\ref{ex:influence-model} only model the   situations where alternative $1$ is manipulated to alternative $2$. Similarly, we can define $\csus{n,B}^{1\ra 3}$, $\csus{n,B}^{2\ra 1}$, $\csus{n,B}^{2\ra 3}$, $\csus{n,B}^{3\ra 1}$, and $\csus{n,B}^{3\ra 2}$. Let $\calM$ denote the set of the six instability settings, and define 
\begin{equation}
\label{eq:union-ex}
\csus{n,B}^\calM \triangleq \csus{n,B}^{1\ra 2}\cup \csus{n,B}^{1\ra 3}\cup\csus{n,B}^{2\ra 1}\cup\csus{n,B}^{2\ra 3}\cup\csus{n,B}^{3\ra 1}\cup\csus{n,B}^{3\ra 2}
\end{equation}

As observed in~\citep{Xia2021:How-Likely}, the histogram of the tentative profile   is a {\em Poisson multivariate variable (PMV)}, defined as follows.
\begin{dfn}[{\bf Poisson multivariate variable (PMV)}{}]
Given $n,q\in\mathbb N$ and a vector $\vec \pi=(\pi_1,\ldots,\pi_n)$ of $n$ distributions over  $\{1,\ldots,q\}$, we define an {\em $(n,q)$-PMV}, denoted by $\vXp$, to be the histogram of $n$ independent random variables whose distributions are $\{\pi_1,\ldots,\pi_n\}$, respectively.
\end{dfn}
Then, it is not hard to verify that
\begin{equation}
\label{eq:CM-Borda}
\satmax{\cm}{\Pi,0}(\borda, n, B) =\sup\nolimits_{\vec\pi\in\Pi^n}\Pr\left(  \vXp  \in  \csus{n,B}^\calM\right)
\end{equation}
It turns out that many other coalitional influence problems  can be modeled as the union of multiple instability settings (formally called the {\em multi-instability setting} in Definition~\ref{dfn:multi-instability-setting} in Appendix) as proved in the following lemma. 
\begin{lem}[{\bf Coalitional Influence  as instability settings}{}]
\label{lem:CI-GSR}
For any  $X\in\{\cm,\mov\}$ and any GSR $r$, there exist multiple instability settings  $\calM = \{\vosetting^i:  i\le I\}$   such that for every $n$-profile $P$ and every $B\ge 0$, $\sat{X}(r,P,B) =1$ if and only if $\hist(P)\in \csus{n,B}^\calM$.
\end{lem}
The full statement  of the lemma (which covers other coalitional influence problems) and its proof can be found in Appendix~\ref{app:proof-prop-CI-GSR}.  When $\psi>0$, the contamination adversary can modify the profile to be any integer vector in the $\psi n$ neighborhood of $\vXp$ in $L_1$, formally defined as {\em $\psi$-fraction-neighborhood} as follows.
\begin{dfn}
For any $\psi\in [0,1]$ and any $\vec x\in {\mathbb R}^{q}$, define the {\em $\psi$-fraction-neighborhood of $\vec x$} as
$$\nb{\vec x}{\psi}\triangleq \{\vec x': |\vec x'|_1=|\vec x'|_1\text{ and } |\vec x'-\vec x|_1\le 2\psi |\vec x|_1\}\}$$ 
For any  set $S$, define 
$$\nb{S}{\psi}\triangleq \bigcup\nolimits_{\vec x\in S} \nb{\vec x}{\psi}$$
\end{dfn} 

In light of \eqref{eq:CM-Borda} and Lemma~\ref{lem:CI-GSR}, we propose the {\em PMV-instability  problem}  as follows. 
\begin{dfn}[{\bf\boldmath  PMV-instability problem}{}]
\label{dfn:PMV-instability-problem}
Given an instability setting $\vosetting{}=\langle\sourcepoly, \targetpoly, \voset{}, \vec c\,\rangle$,  a set $\Pi$ of distributions over $[q]$, $\psi\in[0,1]$, $n\in\mathbb N$, and $B\ge 0$, we are interested in characterizing
\begin{align*}
&\text{\bf  max-semi-random   instability: } \sup\nolimits_{\vec\pi\in\Pi^n}\Pr\left(\vXp\in  \nb{\csus{n,B}}{\psi}\right)
\end{align*}
\end{dfn}
That is, the max-semi-random instability  is the upper bound  on the probability such that the contamination adversary makes up to $\psi$ fraction of changes to the PMV, and the resulting histogram is unstable, when the underlying probabilities $\vec\pi$ are adversarially chosen from $\Pi^n$.  Notice that while $\nb{\csus{n,B}}{\psi}$ may contain non-integer vectors,  $\vXp$ is a non-negative integer vector.  In the main text we focus on presenting results for max-semi-random  instability. The min-semi-random instability can be defined similarly in  Appendix~\ref{sec:PMV-instability-dfn-full}.

\section{Solving The PMV-Instability Problem}
\label{sec:PMV-instability}
We address the PMV-instability problem by discussing the $\psi =0$ case in Section~\ref{sec:main-psi=0} and then the $\psi>0$ case in Section~\ref{sec:main-psi>0}.  We will primarily focus on presenting the $\psi =0$ case because it generalizes previous studies in semi-random social choice as discussed in Example~\ref{ex:model-special-cases}, and its proof serves as the basis for the $\psi>0$ case.


\subsection{\boldmath The $\psi = 0$ Case}
\label{sec:main-psi=0}
When $\psi =0$, the contamination adversary cannot change the tentative profile, and the $0$-fraction-neighborhood of any vector is itself. Therefore, the PMV-instability problem becomes  characterizing
$$ \sup\nolimits_{\vec\pi\in\Pi^n}\Pr\left( \vXp \in   \csus{n,B} \right)$$
The model becomes the semi-random model in~\citep{Xia2020:The-Smoothed}. Nevertheless, the challenge here is to characterize the probability as a function of  (variable) $B$, whereas previous work~\citep{Xia2020:The-Smoothed,Xia2021:How-Likely} characterizes a similar probability for a fixed polyhedron.

Let us start with some high-level intuitions  to motivate the notation, statement, and proof of Theorem~\ref{thm:PMV-instability}. 

\myparagraph{\bf Intuitions.} If $\csus{n,B}=\emptyset$, then the max-semi-random   instability is $0$ by definition, so we will refer to this case  as the {\bf \boldmath $0$ case}. Suppose the $0$ case does not hold, i.e., $\csus{n,B}\ne \emptyset$, then we will adopt an approximation of $\Pr(\vXp\in \csus{n,B})$ based on the following two  approximations. 
\begin{itemize}
\item First,  let us {\em pretend} that all integrality constraints in the linear system (like the one in Example~\ref{ex:illustration}) are relaxed. That is, let us pretend that  $\vXp$ can take non-integer values, and the vote operation vector $\vo$ can be fractional (but still need to be non-negative). For simplicity, assume that $\vec c = \vec 1$. Then, the possible changes to the histogram as the result of  a unit of budget is characterized by $\conv(\voset{})$, where $\conv(\cdot)$ represents the convex hull. It follows that $\csus{n,B}$ can be approximated by a polyhedron $\cpoly{B}$, which is  the intersection of $\sourcepoly$ and the Minkowski addition of $\targetpoly$ and $\bigcup_{0\le b\le B}-b\cdot \conv(\voset{})$.  See Figure~\ref{fig:illustration-PMV} (a) for an illustration of $\cpoly{B}$ in the shaded area, where $\voset{} = \{\vo_1,\vo_2\}$. Notice that $\cpoly{B}$ is in the $-\conv(\voset{})$ direction  of $\targetpoly$.

\item Second,  let us  {\em pretend}  that $\vXp$ is distributed as a $(q-1)$-dimensional Gaussian distribution $\calN_{\vec\pi}$ (whose mean is $\sum_{j=1}^n\pi_j\in n\cdot\conv(\Pi)$) in the hyperplane $\{\vec x:\vec x\cdot \vec 1 = n\}$.  This approximation is justified by various multi-variable central limit theorems, e.g.,~\citep{Bentkus2005:A-Lyapunov-type,Valiant2011:Estimating,Daskalakis2016:A-Size-Free,Diakonikolas2016:The-fourier,Raic2019:A-multivariate}. 
\end{itemize}
With these  approximations, we adopt the following approximation of $\Pr(\vXp\in \csus{n,B})$:
\begin{equation}\label{eq:approximation-intuition}
\Pr(\vXp\in \csus{n,B})\approx\Pr(\calN_{\vec\pi}\in \cpoly{B})
\end{equation}
\begin{figure}[htp]
\centering
\begin{tabular}{cc}
\includegraphics[width = 0.5\textwidth]{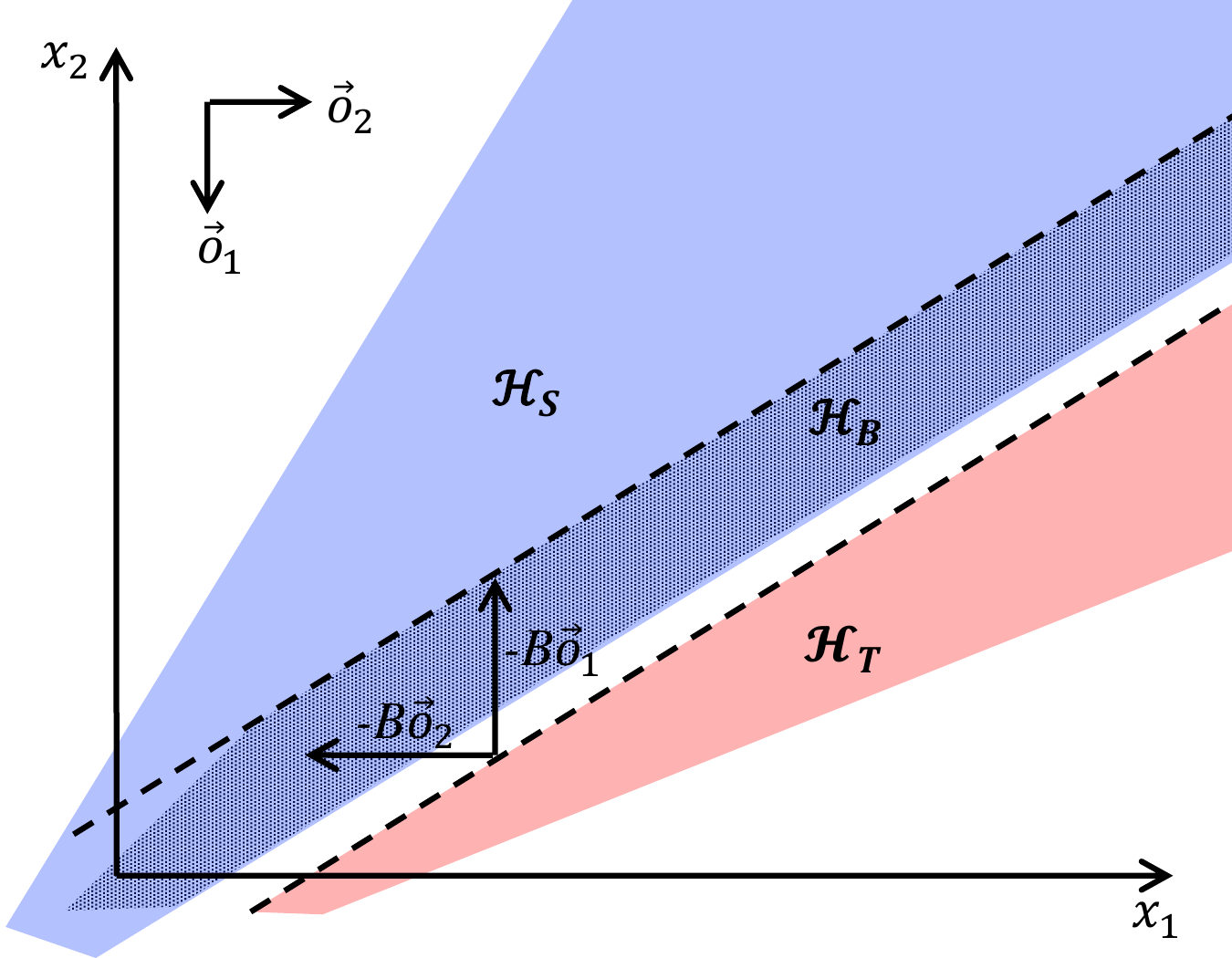} & \includegraphics[width = 0.5\textwidth]{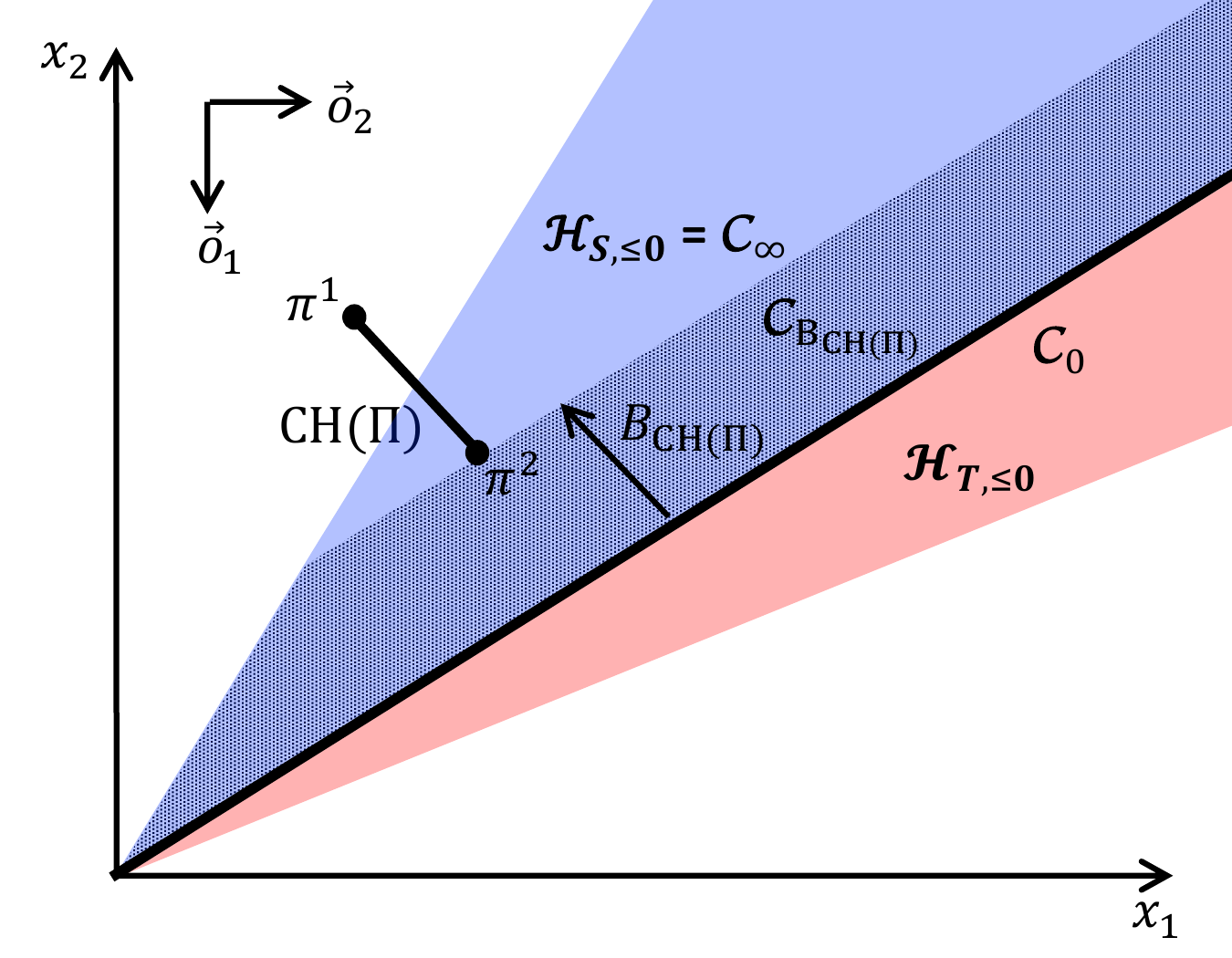}\\
(a) $\cpoly{B}$. & (b) $B_{\conv(\Pi)}$ and $\sus{B_{\conv(\Pi)}}$.
\end{tabular}
\caption{Illustration  of  notation.\label{fig:illustration-PMV}}
\end{figure}
Now, let us take a look at the max-semi-random PMV instability in light of the approximation in \eqref{eq:approximation-intuition}. Because the probability mass of $\calN_{\vec\pi}$ is mostly centered around an $O(\sqrt n)$ neighborhood of its mean, to maximize  $\Pr(\calN_{\vec\pi}\in \cpoly{B})$, the distribution adversary aims to choose $\vec \pi$ so that the ``volume'' of the intersection of an $O(\sqrt n)$ neighborhood of $\sum_{j=1}^n\pi_j$ and $\cpoly{B}$ is as large as possible. 

When $B = O(\sqrt n)$, it turns out that  $\cpoly{B}$ is close to $\sourcepolyz\cap \targetpolyz$, where for any polyhedron $\poly  \triangleq \{\vec x: \ba\times \invert{\vec x}\le \invert{\vbb} \}$, $\polyz  \triangleq \{\vec x: \ba\times \invert{\vec x}\le \invert{\vec 0} \}$ denotes its {\em characteristic cone}, also known as the {\em recess cone}.  Therefore, if $\conv(\Pi)\cap \sourcepolyz\cap \targetpolyz=\emptyset$, then the mean of $\vXp$ is $\Theta(n)$ away from $\sourcepolyz\cap \targetpolyz$, which implies that the likelihood is (exponentially) small due to straightforward applications of Hoeffding's inequality. We call this case the {\bf \boldmath exponential case}.

When $B = O(\sqrt n)$ and $\conv(\Pi)\cap \sourcepolyz\cap \targetpolyz\ne \emptyset$,  the distribution adversary can choose $\vec\pi\in\Pi^n$ so that the mean of $\vXp$ is either in $\sourcepolyz\cap \targetpolyz$ or is $O(1)$ away, which means that the likelihood is large. In this case, there is a  {\bf \boldmath phase transition at $B=\Theta(\sqrt n)$}, as it will be shown that the likelihood reaches its (asymptotic) max at $B = \Theta(\sqrt n)$.

When $B= \Theta(n)$, again, the distribution adversary aims at choosing $\vec\pi\in\Pi^n$  so that the mean of $\vXp$ is close to $\cpoly{B}$. Let $B_{\conv(\Pi)}$ denote the smallest budget so that a ``pseudo-conic'' approximation to $\cpoly{B_{\conv(\Pi)}}$, denoted by $\sus{B}$ and is formally defined below in \eqref{eq:CB}, touches $\conv(\Pi)$, which is the convex hull of $\Pi$. See Figure~\ref{fig:illustration-PMV} (b) for an illustration of $B_{\conv(\Pi)}$   and $\sus{B_{\conv(\Pi)}}$ (the shaded area).  It is then expected that when $B<B_{\conv(\Pi)}\cdot n$, the max-semin-random instability whould be  small, and when  $B>B_{\conv(\Pi)} \cdot n$, the max-semin-random   instability whould be large. In other words, the semi-random stability has a {\bf \boldmath phase transition at  $B=\Theta(n)$}.

Nevertheless, characterizing the conditions and likelihood for each case is highly challenging, as the approximation  only  provides a qualitative  intuition. Existing multi-variate central limit theorems are often too coarse due to an $\Omega(\frac{1}{\sqrt n})$ error, as discussed in~\citep{Xia2021:How-Likely}.

\vspace{2mm}\noindent{\bf Notation.} Now let us define notation to formalize the intuitions discussed above. For any budget $B\ge 0$, we let $\cpoly{B}$ denote the relaxation of $\csus{n,B}$ by removing the size constraint and the integrality constraints on $\vec x$ and on $\vo$. Recall that  $\vo$ still needs to be non-negative. That is,
\begin{equation}\label{dfn:HB}
\cpoly{B} \triangleq \left\{\vec x\in  \sourcepoly:  \exists \vo\in {\mathbb R}_{\ge 0}^{|\voset{}|} \text{ s.t. }\vec c\cdot\vo \le B \text{ and }\vec x+ \vo\times \vomatrix{}\in\targetpoly\right\}
\end{equation}
For example, Figure~\ref{fig:illustration-PMV} (a) illustrates  $\cpoly{B}$ in the shaded area. $\cpoly{B}$ is a polyhedron because it is the intersection of $\sourcepoly$ and the  Minkowski addition of $\targetpoly$ and the following polyhedron $\calQ_{B}$:
$$\calQ_{B} \triangleq \left\{- \vo\times \vomatrix{}:  \vo\in {\mathbb R_{\ge 0}^{|\voset{}|}} \text{ and }\vec c\cdot \vo\le B   \right\}$$
Specifically, with infinite budget ($B=\infty$), we have 
$$\cpoly{\infty} = \left\{\vec x\in  \sourcepoly:  \exists \vo\in {\mathbb R}_{\ge 0}^{|\voset{}|} \text{ s.t. }  \vec x+ \vo\times \vomatrix{}\in\targetpoly\right\} = \sourcepoly\cap (\targetpoly+ \calQ_{\infty})$$
For every $ B\ge 0$, we define $\sus{B}$ to be the polyhedron that consists of all (possibly non-integer) vectors in $\sourcepolyz$ that can be manipulated to be in $\targetpolyz$ by using (possibly non-integer) operations $\vo$ under budget constraint $B$. That is,   
\begin{equation}
\label{eq:CB}
\sus{B} \triangleq \left\{\vec x\in  \sourcepolyz:  \exists \vo\in {\mathbb R}_{\ge 0}^{|\voset{}|} \text{ s.t. } \vec c\cdot \vo\le B\text{ and }\vec x+ \vo\times \vomatrix{}\in\targetpolyz\right\} 
\end{equation}
It is not hard to verify that 
$\sus{B}   = \sourcepolyz\cap (\targetpolyz + \calQ_{B} )$  and $\sus{B}$ can be viewed as a ``pseudo-conic'' approximation to $\cpoly{B}$, as $\sus{B}$ is defined based on the characteristic cones of $\sourcepoly$ and $\targetpoly$, though $\sus{B}$ itself may not be a cone. Specifically,  $\sus{0}$  and $\sus{\infty}$  will play a central role in Theorems~\ref{thm:PMV-instability} and~\ref{thm:multi-instability}. It is not hard to verify that
$$\sus{0} \triangleq \sourcepolyz\cap \targetpolyz, \text{ and} $$
\begin{equation}
\label{eq:cone-infty}
\sus{\infty} \triangleq \left\{\vec x\in  \sourcepolyz:  \exists \vo\in {\mathbb R}_{\ge 0}^{|\voset{}|} \text{ s.t. } \vec x+ \vo\times \vomatrix{}\in\targetpolyz\right\} = \sourcepolyz\cap \left(\targetpolyz + \calQ_{\infty}\right)
\end{equation}
Figure~\ref{fig:illustration-PMV} (b) illustrates $\sus{0}$ (which is a line) and $\sus{\infty}$ (which is the same as $\sourcepolyz$).  Both $\sus{0}$ and $\sus{\infty}$ are polyhedral cones, because the intersection Minkowski addition  of two polyhedral cones is a polyhedral cone. Notice that when $B \not\in \{0,\infty\}$, $\sus{B }$ may not be a cone.

For any set $\Pi^*\subseteq {\mathbb R}^q$, let $B_{\Pi^*}\in\mathbb R$ be the minimum budget $B $ such that the intersection of $\Pi^*$ and $\sus{B }$ is non-empty. If no such $B $ exists (i.e., $\Pi^*\cap \sus{\infty}=\emptyset$), then we let $B_{\Pi^*} \triangleq \infty$. Formally,
\begin{equation}
\label{eq:B-min}
B_{\Pi^*} \triangleq \inf \{B\ge 0: \Pi^*\cap \sus{B } \ne \emptyset\}
\end{equation}
Figure~\ref{fig:illustration-PMV} (b) illustrates $B_{\conv(\Pi)}$ and $\sus{B_{\conv(\Pi)}}$ in the shaded area, where $\Pi= \{\pi^1,\pi^2\}$ and $\conv(\Pi)$ is   the line segment between $\pi^1$ and $\pi^2$. Next, we define  notation and conditions used in the statement of the  theorem.

\begin{dfn}
\label{dfn:conditions}
Given an instability setting, $\Pi$, $B$, and $n$, define
$$d_0 = \dim(\sus{0}), d_{\infty} =  \dim(\sus{\infty})\text{, and } d_{\Delta}  = d_{\infty}  - d_0,$$
where $\dim(\sus{0})$ is the {\em dimension} of $\sus{0}$, which is the dimension of the minimal affine space that contains $\sus{0}$. We also define the following three conditions:
\begin{center}
\begin{tabular}{|c|c|c|}
\hline $\condition{1}$ & $\condition{2}$ & $\condition{3}$  \\
\hline $\csus{n,B}=\emptyset$ & $\conv(\Pi)\cap\sus{\infty}=\emptyset$ & $\conv(\Pi)\cap\sus{0}=\emptyset$ \\
\hline
\end{tabular}
\end{center}
\end{dfn}
Recall that $\conv(\Pi)$ is the convex hull of $\Pi$. Because $\sus{0}\subseteq \sus{\infty}$, $d_\Delta\ge 0$. Also notice that $\condition{2}$ implies $\condition{3}$, or equivalently, $\neg\condition{3}$ implies $\neg\condition{2}$. 

\begin{thm}[\bf\boldmath  Semi-Random  PMV-Instability, $\psi =0$]
\label{thm:PMV-instability} Given any $q\in\mathbb N$, any closed and strictly positive $\Pi$ over $[q]$, and any instability settings $\vosetting{} = \langle\sourcepoly,\targetpoly,\voset{},\vec c\,\rangle$, any $C_2>0$ and $C_3>0$ with $C_2<B_{\conv(\Pi)}<C_3$,  any $n\in \mathbb N$, and any $B\ge 0$, 
$$ 
 \begin{array}{r@{}|l|@{}l@{}|l|}
\cline{2-4} & \text{\bf Name} &  \text{\bf\  Likelihood}&  \text{\bf Condition}\\
\cline{2-4} \multirow{6}{*}{$\sup_{\vec\pi\in\Pi^n}\Pr\left(\vXp \in \csus{n,B}\right)= \left\{\begin{array}{@{}r@{}}\\\\\\\\\\\end{array}\right.$}& \text{0 case}& \ 0 &\condition{1}\\
\cline{2-4}& \text{exp case}& \ \exp(-\Theta(n)) &\neg \condition{1}\wedge \condition{2}\\
\cline{2-4}&  \text{PT-$\Theta(\sqrt n)$}& \ \Theta\left(\dfrac{\min\{B+1,\sqrt n\}^{d_{\Delta}}}{(\sqrt n)^{q-d_0}} \right) &\neg \condition{1}\wedge \neg \condition{3}\\
\cline{2-4}& \text{PT-$\Theta(n)$} & \begin{array}{ l l } \exp(-\Theta(n)) &\text{if }B\le C_2n \\ \hline \Theta\left( (\frac{1}{\sqrt n})^{q-d_\infty}  \right) &\text{if }B\ge C_3n\end{array}&\begin{array}{@{}l}\text{otherwise, i.e.,}\\ \neg \condition{1}\wedge \neg \condition{2}\wedge \condition{3}\end{array}\\
\cline{2-4}
\end{array}
$$
\end{thm}
The $B+1$ value is introduced to handle the $B<1$ case. For every $B\ge 1$,   $\Theta(B+1)=\Theta(B)$.\\
\begin{minipage}[t][][b]{0.3\textwidth}
\noindent{\bf The four cases.} Following the intuition presented at the beginning of this  section, we  call the first case  in Theorem~\ref{thm:PMV-instability} the  {\em $0$ case}, the second case the  {\em  exponential case}, the third case the   {\em phase transition at $\Theta(\sqrt n)$ case} (PT-$\Theta(\sqrt n)$ for short), and the last case, which contains two subcases, the  {\em phase transition at $\Theta(n)$ case} (PT-$\Theta(n)$  for short). 
\end{minipage}
\hfill
\begin{minipage}[t][][b]{0.65\textwidth}
\centering
\includegraphics[width = 1\textwidth]{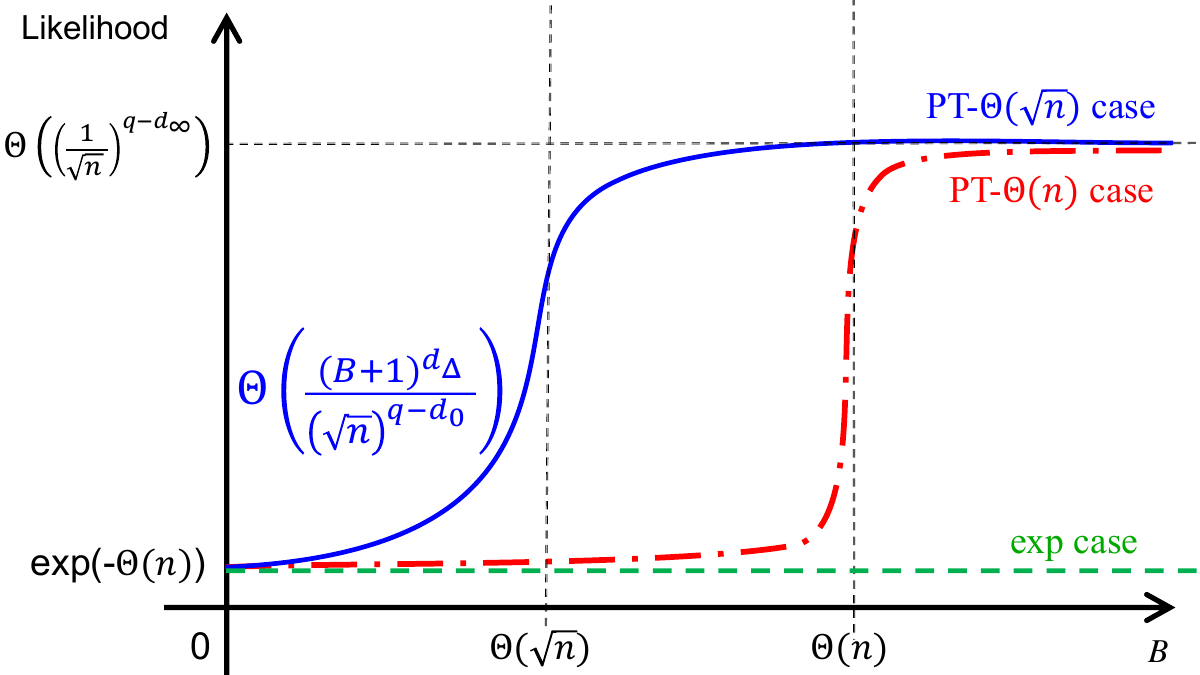}
\captionof{figure}{\small Illustration of Theorem~\ref{thm:PMV-instability}.  The x-axis is  in log scale. \label{fig:PMV-instability}}
\end{minipage}

\vspace{1mm}Figure~\ref{fig:PMV-instability} illustrates the max-semi-random instability as a function of $B$ for the exp case, PT-$\Theta(\sqrt n)$ case, and PT-$\Theta(n)$ case, respectively. Figure~\ref{fig:PMV-cases} (a)  illustrates  condition $\condition{2} = [\conv(\Pi)\cap\sus{\infty}=\emptyset]$ for the exp case. Figure~\ref{fig:PMV-cases} (b)  illustrates  condition $\neg\condition{3} =  [\conv(\Pi)\cap\sus{0}\ne \emptyset]$  for the PT-$\Theta(\sqrt n)$ case. Figure~\ref{fig:illustration-PMV} (b) illustrates condition $\neg\condition{2}\wedge \condition{3} =  [\conv(\Pi)\cap\sus{\infty} \ne \emptyset]\wedge [\conv(\Pi)\cap\sus{0} = \emptyset]$  for the PT-$\Theta(n)$ case.

\begin{figure}[htp]
\centering
\begin{tabular}{cc}
\includegraphics[width = 0.5\textwidth]{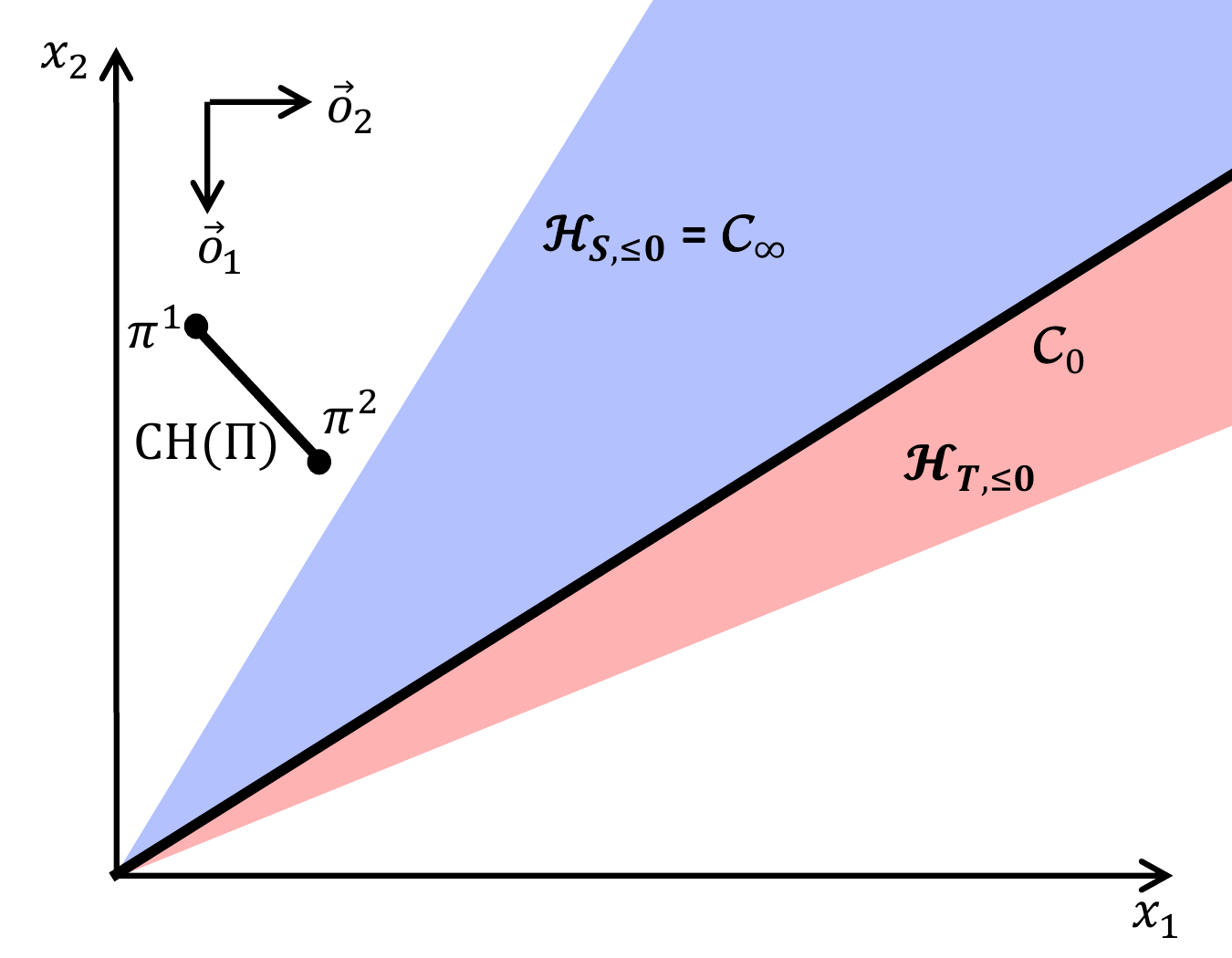} & \includegraphics[width = 0.5\textwidth]{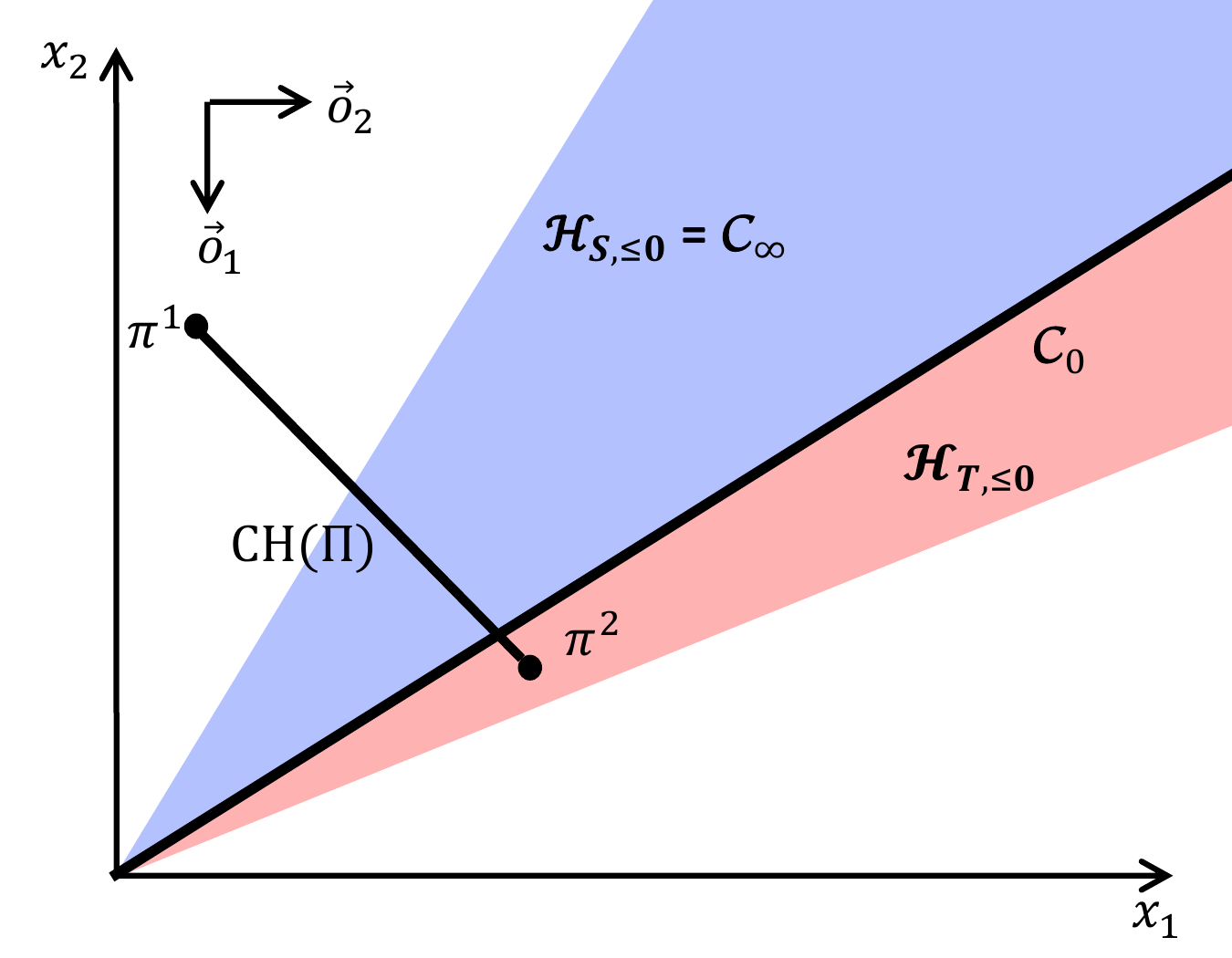}\\
(a) The exp case. & (b) The PT-$\Theta(\sqrt n)$ case.
\end{tabular}
\caption{Illustration of the exp case and phase-transition-at-$\Theta(\sqrt n)$ case.\label{fig:PMV-cases}}
\end{figure}
\noindent{\bf Uses and limitations of Theorem~\ref{thm:PMV-instability}.}   Theorem~\ref{thm:PMV-instability} provides a general and useful tool for accurately addressing PMV-instability problems, because it converts the complicated PMV-instability problems, which involve  reasoning about the likelihood of discrete events (about the PMV) that cannot be easily bounded by standard techniques, to deterministic geometric problems about $\conv(\Pi)$, $\sus{0}$, $\sus{\infty}$, $d_0$, and $d_\infty$.  It provides an almost complete characterization of the PMV-instability problem, which can be easily applied to resolve long-standing open questions, e.g., in Corollary~\ref{coro:IC}. Practically, $\dim_0$ and $\dim_\Delta$ can still be hard to characterize, but at least Theorem~\ref{thm:PMV-instability} provides a useful guideline about what to look for. See Section~\ref{sec:applications} for three examples of such applications. Theorem~\ref{thm:PMV-instability} has  two major limitations. First, it misses a knife-edge case of $B$, i.e., when $B$ is around  $B_{\conv(\Pi)}\cdot n$. Second, the constants in asymptotic bounds may be exponentially large in $m$. This may not be a critical issue when $m$ is small or is viewed as a constant as in most previous work.

  
\myparagraph{\bf Proof sketch of Theorem~\ref{thm:PMV-instability}.} At a high level, the proof follows after the intuitions presented  in the beginning of this subsection. The hardest part is the proof of the (asymptotically tight) polynomial bounds in the PT-$\Theta(\sqrt n)$ case. Take $\sup$ and $B = O(\sqrt n)$ for example. To prove the polynomial upper bound, our proof can be viewed as upper-bounding the ``volume''  of the intersection of an $O(\sqrt n)$ neighborhood of $\sum_{j=1}^n\pi_j$ and $\cpoly{B}$. We prove that, in $d_0$ dimensions, the volume is large, and each such dimension contributes a multiplicative $O(1)$ factor to the likelihood; in $d_\Delta$ dimensions, the volume is $O(B)$, and each such dimension contributes a multiplicative $O\left(\frac{B+1}{\sqrt n}\right)$ factor to the likelihood; and in the remaining $q-d_\infty$ dimensions, the volume is $O(1)$, and each such dimension contributes a multiplicative $O(\frac{1}{\sqrt n})$ factor to the likelihood. Putting all together, this proves the desired upper bound
$$O(1)^{d_0}\times O\left(\frac{B+1}{\sqrt n}\right)^{d_\Delta}\times  O\left(\frac{1}{\sqrt n}\right)^{q-d_\infty} =  O\left(\dfrac{ (B+1)^{d_{\Delta}}}{(\sqrt n)^{q-d_0}} \right)$$

  To prove the polynomial lower bound,  we first pretend that the PMV can take non-integer values, then  enumerate  (possibly non-integral) vectors that are far away from each other by exploring two directions. The first direction is the convex hull of $\sus{0}$, which is a   $(d_0-1)$-dimensional space that represents no budget ($B=0$), and each such dimension contributes an $\Omega(\sqrt n)$ multiplicative factor  to the total number of desirable vectors. The second direction is the convex hull of $\sus{\infty}$, which is a   $d_\Delta$-dimensional space   that represents infinite budget ($B=\infty$), and each such dimension contributes an $\Omega(B+1)$ multiplicative factor  to the total number of desirable vectors. Then, we  prove that for each such (possibly non-integral) vector, there exists a nearby integer vector, and apply the pointwise concentration bound~\citep[Lemma~1]{Xia2021:How-Likely} to prove the desired lower bound. The full version of the theorem and its full proof can be found in Appendix~\ref{app:proof-thm-PMV-instability}. $\hfill\Box$
  
\subsection{\boldmath  The $\psi > 0$ Case}
\label{sec:main-psi>0}

\begin{thm}[\bf\boldmath Semi-Random PMV-Instability, $\psi>0$]
\label{thm:PMV-instability-psi>0} Given any $q\in\mathbb N$, any closed and strictly positive $\Pi$ over $[q]$, and any instability setting $\vosetting{} = \langle\sourcepoly,\targetpoly,\voset{},\vec c\,\rangle$, any $C_1>0$, any $n\in \mathbb N$, and any $0\le B\le C_1 \sqrt n$,  
$$\sup_{\vec\pi\in\Pi^n}\Pr\left(\vXp\in \nb{\csus{n,B}}{\psi} \right)= 
\begin{cases}
0 & \text{if }\csus{n,B}=\emptyset\\
\exp(-\Theta(n)) & \text{otherwise, if }\conv(\Pi)\cap  \nb{\sus{0}}{\psi}=\emptyset\\
\Theta(1)&\text{otherwise}
\end{cases}$$  
 \end{thm} 
Theorem~\ref{thm:PMV-instability-psi>0} is simpler than Theorem~\ref{thm:PMV-instability}: suppose the $0$ case does not hold, then, either the likelihood is exponentially low even with $\Theta(\sqrt n)$ budget, or the likelihood is $\Theta(1)$ high even with $0$ budget. The high-level intuition behind the $\Theta(1)$ case is that the $\Theta(n)$ contamination adversary can be viewed as doing the influencers' job. Notice that this is only a technical observation. As discussed after Definition~\ref{dfn:Pi-psi-model}, we need to distinguish the contamination adversary from the influencers from the modeling perspective. 
 
\myparagraph{\bf Proof sketch of Theorem~\ref{thm:PMV-instability-psi>0}.} The $0$ case is straightforward. The exponential case is proved by  applying Hoeffding's inequality as in the proof of the exponential case in Theorem~\ref{thm:PMV-instability}. The most challenging case is the $\Theta(1)$ case, whose proof is similar to the proof of the  lower bound of the polynomial case in Theorem~\ref{thm:PMV-instability}. The main difference and challenge is to  specify $\Theta(n^{(q-1)/2})$ integer vectors in an $O(\sqrt n)$ neighborhood of the sum of distributions in some $\vec \pi^*\in \Pi^n$, such that each such vector can be modified by a contamination adversary  to a vector $\vec x\in \sus{0}$. The full version of Theorem~\ref{thm:PMV-instability-psi>0} and its proof can be found in Appendix~\ref{sec:psi>0}.
$\hfill\Box$

\section{Applications: Semi-Random Coalitional Influence}
\label{sec:applications} 

\myparagraph{The overall approach.}  
We propose the following two-step polyhedral approach in Procedure~\ref{proc:application} for characterizing   $(\Pi,\psi)$-semi-random $X$.
\renewcommand{\algorithmcfname}{Procedure}
\begin{algorithm} 
\caption{The polyhedral approach for  semi-random coalitional influence}\label{proc:application}
\begin{algorithmic}
\item {\bf Input:}  $X$, $r$, $\Pi$, $\psi$
\item {\bf Step (i):} Model $X$ under $r$ as a union of multiple instability settings $\calM$.
\item {\bf Step (ii):} Characterize the conditions and degree of polynomial for $\calM$ by applying Theorem~\ref{thm:PMV-instability} and Theorem~\ref{thm:PMV-instability-psi>0}, and then combine the results. 
\end{algorithmic}
\end{algorithm}

Step (i) is often easy. 
The difficulty level of Step (ii) is highly problem-dependent, and we see two potential difficulties: first, it might be hard to verify for which $n$,  $B$, and PMV-instability problem,  the $0$ case of Theorem~\ref{thm:PMV-instability} does not happen; and second, sometimes the outcomes of  the applications of Theorem~\ref{thm:PMV-instability} to the instability settings in $\calM$ are too complicated to be combined. 

In this section, we present three applications of Procedure~\ref{proc:application}  to illustrate its usefulness. 

\myparagraph{The first application} is an asymptotically tight characterization for many commonly-studied voting rules (defined in Appendix~\ref{app:more-rule}). Recall that $\piuni$ is the uniform distribution over $\ml(\ma)$.
\begin{thm}[{\bf Max-Semi-Random Coalitional Influence: Commonly Studied Rules}{}]
\label{thm:PMV-instability-applications}{
Let $r$ be an integer positional scoring rule, STV, ranked pairs, Schulze, maximin, or Copeland with lexicographic tie-breaking for any fixed $m\ge 3$. For  any closed and strictly positive $\Pi$ with $\piuni\in\conv(\Pi)$, any $X\in \{\cm,\mov\}$, there exists $ N>0$  such that for any $n>N$ and any $B\ge 1$,
$$\satmax{X}{\Pi,0}(r ,n,B) =  \Theta\left(\min\left\{\frac{B}{\sqrt n},1\right\}\right)  $$
Moreover, for any $\psi>0$, $\satmax{X}{\Pi,\psi}(r ,n,B) =\Theta(1)$.
} \end{thm}
\noindent{\bf Proof sketch.} We apply Procedure~\ref{proc:application} to $X=\cm$  under any  integer positional scoring rule $r_{\vec s}$ to illustrate the idea.  For Step (i) of Procedure~\ref{proc:application}, we adopt a representation that is similar to that in Examples~\ref{ex:illustration} and \ref{ex:influence-model} and \eqref{eq:union-ex} (formally defined  in the proof of Lemma~\ref{lem:CI-GSR} in Appendix~\ref{app:proof-prop-CI-GSR}). Due to the $O\left(\min\left\{\frac{B}{\sqrt n},1\right\}\right)$ upper bound in Theorem~\ref{thm:PMV-instability-GSR-upper}, which will be proved soon, it suffices to prove a matching  lower bound by identifying an instability setting $\vosetting$ that represents some unstable histograms $\csus{n,B}$, and then apply  Theorem~\ref{thm:PMV-instability} to $\vosetting$ to prove that for all $B\ge 1$, 
\begin{equation}
\label{eq:B-ineq-main-text}
 \sup\nolimits_{\vec\pi\in \Pi^n}\Pr\nolimits_{P\sim\vec \pi}(\hist(P)\in\csus{n,B}) = \Theta\left(\min\left\{\frac{B}{\sqrt n},1\right\}\right)
\end{equation}
Like in Examples~\ref{ex:illustration} and~\ref{ex:influence-model}, let $\sourcepoly$ denote the set of vectors where alternative $1$'s total score is at least as high as alternative $2$'s total score, which is strictly higher than the total score of any other alternative. Let $\targetpoly$ denote the set of vectors where $2$'s total score is the strictly highest. Let $\voset{\pm}^{1\ra 2}$ denote the vote operations where a ranking with $2\succ 1$ is replaced by another ranking. Let $\vosetting= \langle\sourcepoly, \targetpoly, \voset{\pm}^{1\ra 2},\vec 1\,\rangle$.

For Step (ii) of Procedure~\ref{proc:application}, we apply  Theorem~\ref{thm:PMV-instability} to $\vosetting$. We prove $\condition{1}=\text{false}$ for any sufficiently large $n$ by explicitly constructing a profile where a single manipulator can and has incentive to change the winner from $1$ to $2$. Then, we show that $\condition{3}=\text{false}$, because 
(1)  $\piuni\in \conv(\Pi)$ (assumption in Theorem~\ref{thm:PMV-instability-applications}), and (2) $\piuni\in \sus{0}$ (because $\sus{0}=\sourcepolyz\cap \targetpolyz$ represents vectors where $1$ and $2$ are co-winners, and $r_{\vec s}(\piuni) =\ma$). This means that the polynomial case of Theorem~\ref{thm:PMV-instability} holds. It is not hard to verify that $d_0 = \dim(\sus{0}) = m!-1$. We then prove that $d_\infty  = m!$ by showing that an interior point of $\sourcepolyz$ can be manipulated to an interior point of $\targetpolyz$. Inequality \eqref{eq:B-ineq-main-text} then follows after the polynomial case of $\sup$ of Theorem~\ref{thm:PMV-instability}, which proves the $\cm$ part under $r_{\vec s}$. 

The $\psi>0$ case follows after a straightforward application of Theorem~\ref{thm:PMV-instability-psi>0}. $\hfill\Box$

The full statement of Theorem~\ref{thm:PMV-instability-applications} (including other coalitional influence problems such as the ones defined in Appendix~\ref{app:more-CI})  and its full proof can be found in Appendix~\ref{app:applications-proofs}. The next example shows an application of Theorem~\ref{thm:PMV-instability-applications} to Borda. 
\begin{ex}
\label{ex:thm-Borda-CM}
In the setting of Example~\ref{ex:sCM-Borda}, notice that $\piuni = \frac12(\pi^1+\pi^2)$, which means that $\piuni\in\conv(\Pi)$. It follows from Theorem~\ref{thm:PMV-instability-applications} that for all  sufficiently large $n$, $\satmax{\cm}{\Pi,0}(\borda,n,1) =  \Theta\left( \frac{1}{\sqrt n} \right) $, and for every $\psi>0$, $\satmax{\cm}{\Pi,0}(\borda,n,1) =  \Theta\left(1\right) $.
\end{ex}

Theorem~\ref{thm:PMV-instability-applications} leads to the following  corollary about IC, which is the$(\{\piuni\},0)$-semi-random model as mentioned in Example~\ref{ex:model-special-cases}.
\begin{coro}[{\bf Likelihood of Coalitional Influence w.r.t.~IC}{}]
\label{coro:IC}
Let $r$ be an integer positional scoring rule, STV, ranked pairs, Schulze, maximin, or Copeland with lexicographic tie-breaking for fixed $m\ge 3$. For any $X\in \{\cm,\mov\}$,   there exists $N>0$  such that for any $n>N$ and any $B\ge 1$,

$\hfill \Pr\nolimits_{P\sim \text{IC}}(\sat{X}(r,P,B)) =  \Theta\left(\min\left\{\frac{B}{\sqrt n},1\right\}\right) \hfill $
\end{coro}
As discussed in the Introduction, the $X=\cm$ and $B=1$ case of Corollary~\ref{coro:IC} proves the $\Omega\left(\min\left\{\frac{1}{\sqrt n},1\right\}\right)$ matching lower bound for many commonly-studied voting rules, especially all integer positional scoring rules and STV.  

\myparagraph{The second application} is an upper bound for all generalized scoring rules  and all closed and strictly positive $\Pi$ with $\psi = 0$. 
\begin{thm}[{\bf Upper bound on Coalitional Influence  under GSRs}{}]
\label{thm:PMV-instability-GSR-upper}
Let $r$ denote any GSR with fixed $m\ge 3$. For any closed and strictly positive $\Pi$, any $X\in\{ \cm, \mov\}$, any $n$, and any $B\ge 1$,
$$\hfill \satmax{X}{\Pi,0}(r,n,B) =  O\left(\min\left\{\frac{B}{\sqrt n},1\right\}\right) \hfill $$
\end{thm}
\noindent{\bf Proof sketch.} For Step (i) of  Procedure~\ref{proc:application}, we adopt the same representation as in the proof of Lemma~\ref{lem:CI-GSR}. Step (ii) relies on the notation and results of a general dichotomy theorem on the semi-random likelihood of influence problems for GSRs (Theorem~\ref{thm:GSR} in Appendix~\ref{sec:semi-random-CI}). The proof can be found in Appendix~\ref{app:proof-Lemma-PMV-upper}. 
%
$\hfill\Box$ 

Theorem~\ref{thm:PMV-instability-GSR-upper} immediately extends all previous $O\left( \frac{1}{\sqrt n}  \right) $ upper bound on $\cm$ for a single manipulator ($B=1$) that will be discussed in Section~\ref{sec:related-work} to any coalition size $B\ge 1$, because all rules studied in these works are GSRs.  

\myparagraph{The third application} studies a new  notion of coalitional manipulation that aims at making the loser win under integer positional scoring rules. For any positional scoring rule, the loser is the alternative with the minimum total score.

\begin{dfn}[\bf Coalitional manipulation for the loser] Given any integer positional scoring rule $r_{\vec s}$ with lexicographic tie-breaking, for any profile $P$ and any $B\ge 0$, we define $\cml(r_{\vec s}, P, B) = 1$ if and only if a coalition of no more than $B$ voters have incentive to misreport their preferences to make the loser under $P$ win. 
\end{dfn}

Clearly, under veto, no coalition of voters have incentive to misreport their preferences to make the loser win, i.e., $\cml(\veto, P, B)=0$ for all $P$ and $B$. This is because any voter who  prefers the loser to the winner cannot reduce the difference between the total score of the loser and the sum of total score of all other alternatives, which means that no matter how they vote, the loser would not become the winner. The following theorem characterizes the max-semi-random $\cml$ under any other integer positional scoring rule for every sufficiently large $n$ and $B$.

\begin{thm}[\bf Max-Semi-Random Coalitional Manipulation for the Loser]
\label{thm:CML}{
Let $r_{\vec s}$  be an integer positional scoring rule with lexicographic tie-breaking for fixed $m\ge 3$ that is different from veto. For  any closed and strictly positive $\Pi$ with $\piuni\in\conv(\Pi)$, there exist $ N>0$ and $B^*>0$ such that for any $n>N$ and any $B\ge B^*$,
$$\satmax{\cml}{\Pi,0}(r_{\vec s} ,n,B) =  \Theta\left(\min\left\{\frac{B}{\sqrt n},1\right\}^{m-1}\right)$$

\noindent Moreover, for any $\psi>0$, $\satmax{\cml}{\Pi,\psi}(r_{\vec s} ,n,B) =\Theta(1)$.} \end{thm}
\noindent{\bf Proof sketch.}  For Step (i) of  Procedure~\ref{proc:application},  we first define a set of instability settings, denoted by $\calM_{\cml}$, that represents $\cml$ under $r_{\vec s}$.  For every pair of different alternatives $a,b$, define an instability setting $\vosetting_{\cml}^{a\ra b} = (\sourcepoly, \targetpoly, \voset{\pm}^{a\ra b}, \vec 1)$, where  $\sourcepoly$ denotes the set of vectors where $a$ is the winner and $b$ is the loser under $r_{\vec s}$, and  $\targetpoly$ denotes the set of vectors where $b$ is the winner  under $r_{\vec s}$. Define $\calM_{\cml} = \left\{\vosetting_{\cml}^{a\ra b}: a,b\in\ma, a\ne b\right\}$. 
It is  not hard to verify that   $\calM_{\cml}$ represents $\cml$ under $r_{\vec s}$.

Then, for Step (ii) of Procedure~\ref{proc:application}, we  apply Theorem~\ref{thm:PMV-instability} to every $\vosetting_{\cml}^{a\ra b}$. We first prove $\condition{1} = \text{false}$  for any sufficiently large $n$ and $B$, by constructing a successful manipulation by $B$ voters.  Then, let $\sus{0}^{a\ra b}= \sourcepolyz\cap \targetpolyz$. We have $\piuni\in \conv(\Pi)\cap \sus{0}$, which means that $\condition{3} = \text{false}$. Therefore, the polynomial case of Theorem~\ref{thm:PMV-instability} holds. It is then not hard to verify that   $d_0 = \dim(\sus{0}^{a\ra b}) = m!-(m-1)$ and $d_{\infty} = \dim(\sourcepolyz) = m!$, which means that $d_{\Delta} =  m-1 $. By Theorem~\ref{thm:PMV-instability}, we have
$$\hfill \sup\nolimits_{\vec\pi\in\Pi^n}\Pr\left(\vXp \in \csus{n,B}\right)= \Theta\left(\dfrac{\min\{B+1,\sqrt n\}^{m-1}}{(\sqrt n)^{m!-(m!-(m-1))}} \right) = \Theta\left(\min\left\{\frac{B}{\sqrt n},1\right\}^{m-1}\right)\hfill $$

Theorem~\ref{thm:CML} follows after applying this bound to all $a\ne b$. The $\psi>0$ case follows after a straightforward application of Theorem~\ref{thm:PMV-instability-psi>0}. The full proof can be found in Appendix~\ref{app:proof-thm-CML}. $\hfill\Box$


\section{Related Work and Discussions}
\label{sec:related-work}

Table~\ref{tab:cm} summarizes most relevant theoretical work on the likelihood of $\cm$. 
\begin{table}[htp]
\centering
\begin{tabular}{|p{.3\textwidth}|c| c@{\ }| c@{\ }|}
\hline \bf Rule & \boldmath $B$ & \bf Distribution & \bf  Likelihood\\ \hline
Plurality~\cite{Pazner1978:Cheatproofness} & $n^\alpha$ ($ \alpha<0.5$) & IC & $o(1)$ \\ \hline
Representable voting systems~\cite{Peleg1979:Note} & $o(\sqrt n)$ & positive i.i.d. &  $o(1)$ \\ \hline
Borda, plurality, and range voting~\cite{Nitzan1985:Vulnerability}& \multirow{ 9}{*}{$1$}  &  \multirow{ 9}{*}{IC}  & decreases in $n$ \\ \cline{1-1}\cline{4-4}
Plurality~\cite{Fristrup1989:A-Note}&   &   & \multirow{ 7}{*}{ $O\left(\frac{1}{\sqrt n}\right)$} \\ \cline{1-1} 
plurality with runoff~\cite{Slinko2002:Asymptotic}&  &  &   \\ \cline{1-1} 
Representable voting systems, UMG-based rules~\cite{Slinko2002:On-Asymptotic} &   & &   \\ \cline{1-1} 
Bucklin~\cite{Slinko2002:The-Majoritarian}& &  &   \\ \cline{1-3}
Positional scoring rules, top cycle, and Copeland~\cite{Baharad02:Asymptotic}& $1$ & small correlation &   \\ \hline
Plurality~\cite{Slinko2002:Asymptotic} & $1$ & IC & $\Omega(\frac{1}{\sqrt n})$ \\ \hline
Rules constantly far away from dictatorships~\cite{Friedgut2011:A-quantitative,Dobzinski08:Frequent,Xia08:Sufficient,Isaksson10:Geometry,Mossel2015:A-quantitative} & $1$ & IC & $\Omega(\frac{1}{n^{67}})$ \\ \hline
Condorcet rules~\cite{Xia2023:Semi-Random}& any $O(\sqrt n)$ & Semi-random & $\Omega(\frac{B}{\sqrt n})$ \\ \hline
Positional scoring rules~\cite{Slinko2004:How-large} &  $O(n^\alpha)$ ($\alpha<0.5$) & IC & $O((\frac{1}{n})^{0.5-\alpha})$ \\ \hline
Positional scoring rules~\cite{Pritchard09:Asymptotics} &  $c\sqrt n$ & IC & a function of $c$ \\ \hline
Positional scoring rules~\cite{Procaccia2007:Junta} & $o(\sqrt n)$ &  \multirow{6}{*}{ i.i.d.} & $o(1)$ \\ \cline{1-2}\cline{4-4}
GSRs~\cite{Xia08:Generalized}&  $\Theta(n^\alpha)$ ($\alpha<0.5$) & & \begin{tabular}{@{}c@{}}$O(n^{\alpha-0.5})$ if $\alpha<0.5$,\\ $1-\exp(-\Theta(n))$ if $\alpha>0.5$ \end{tabular}\\  \cline{1-2}\cline{4-4}
 GSRs~\cite{Mossel13:Smooth}&  around $c\sqrt n$ &  & phase transition\\  \cline{1-2}\cline{4-4}
  GDSRs~\cite{Xia15:Generalized}&  \begin{tabular}{@{}c@{}} $0$, $\Theta(\sqrt n)$, \\ $\Theta (n)$, or $\infty$ \end{tabular}& & $1-o(1)$ \\
  \hline
\end{tabular}
\caption{Literature on the likelihood of $\cm$.\label{tab:cm}}
\end{table}

\noindent{\bf\bf Likelihood of coalitional manipulability: upper bounds  w.r.t.~IC.} \citet{Pattanaik1975:Strategic} proposed to study the likelihood of strategic voting by single manipulator and conjectured that the likelihood is smaller in larger elections. \citet{Pazner1978:Cheatproofness} proved that the likelihood of single-voter  manipulability goes to $0$ as $n\ra\infty$ under plurality. They also noted  that the results can be extended to any coalition of $n^\alpha$  manipulators, where $0\le \alpha<1/2$.   \citet{Peleg1979:Note} proved that the likelihood of single-voter manipulability under any {\em representable voting system}, which includes positional scoring rules,  goes to $0$ under any positive   i.i.d.~distributions.  Peleg also considered coalitional manipulation problem, by noticing in footnote 2  that the result still holds for any coalition of $o(\sqrt n)$ voters.  \citet{Nitzan1985:Vulnerability} demonstrated that the likelihood decreases when $n$ is large  under Borda, plurality, and range voting.   \citet{Fristrup1989:A-Note} proved an $O\left(\frac{1}{\sqrt n}\right)$ rate of convergence for single-voter  manipulability under plurality. \citet{Kim1996:Statistical} proved that maixmin (a.k.a.,~Simpson's method) can be manipulated  by a coalition of unlimited number of voters almost surely as $n\ra\infty$.  \citet{Slinko2002:Asymptotic,Slinko2002:On-Asymptotic,Slinko2002:The-Majoritarian} proved an $O\left(\frac{1}{\sqrt n}\right)$ upper bound for a single manipulator under plurality with runoff, representable voting system, all voting rules based on the unweighted majority graphs, and Bucklin.
\citet{Baharad02:Asymptotic} proved an $O\left(\frac{1}{\sqrt n}\right)$ upper bound for a single manipulator under positional scoring rule, top cycle, and Copeland, when voters' preferences has small local correlations, which is more general than IC.   \citet{Slinko2004:How-large} investigated the likelihood of coalitional manipulation by up to $Cn^\alpha$ manipulator for any fixed $0\le \alpha < 1/2$, and proved an $O((\frac{1}{n})^{0.5-\alpha})$ upper bound for any positional scoring rule with strictly decreasing scores. \citet{Maus2007:Anonymous} characterized the least manipulable rule by a single manipulator among tops-only, anonymous, and surjective choice rules, to be the unanimity rules with status quo.  \citet{Pritchard09:Asymptotics} showed that, the likelihood for a coalition of $c\sqrt n$ manipulators to succeed under a positional scoring rule is a function of $c$, and provided an algorithm based on integer linear program to compute the minimum coalition size.


\myparagraph{\bf Likelihood of coalitional manipulability: lower bounds w.r.t.~IC.} 
\citet{Slinko2002:Asymptotic} proved an $\Omega(\frac{1}{\sqrt n})$ lower bound under plurality for a single manipulator. A quantitative Gibbard-Satterthwaite theorem was   proved for $m=3$ by~\citet{Friedgut2011:A-quantitative}, and  was subsequently developed in~\citep{Dobzinski08:Frequent,Xia08:Sufficient,Isaksson10:Geometry}, and the general case was proved by~\citet{Mossel2015:A-quantitative}, which implies that any voting rule that is constantly away from any dictatorships, the likelihood of single-voter manipulability  under IC is $\Omega\left(\frac{1}{n^{67}m^{166}}\right)$, i.e., $\Omega(\frac{1}{n^{67}})$ for any fixed $m$. \citet{Xia2023:Semi-Random}  proved an $\Omega(\frac{B}{\sqrt n})$ lower bound for $\cm$ for an arbitrary size $1\le B\le \sqrt n$  of the coalition under any Condorcet consistent rule w.r.t.~a large class of semi-random models that include IC. The tightness of these bounds (except for plurality and some Condorcet consistent rules for $B=1$) is an open question, which is addressed by our Corollary~\ref{coro:IC}.

\myparagraph{\bf Likelihood of coalitional manipulability: simulations.}  Beyond theoretical work, there is also a large literature on comparing the empirical coalitional manipulability of  commonly-studied voting rules, mostly by computer simulations~\citep{Chamberlin1985:Investigation,Kelly1993:Almost,Smith1999:Manipulability,Aleskerov1999:Degree,Pritchar07:Exact,Aleskerov2011:An-individual,Aleskerov2012:On-the-manipulability,Green-Armytage2016:Statistical}. These works confirm that the likelihood for a large election to be manipulable by a single manipulator is low.

\myparagraph{\bf\bf Likelihood of coalitional manipulability under other distributions.} As discussed above, the convergence-to-$0$ result for representable voting systems in \citep{Peleg1979:Note}  works for any i.i.d.~distribution, and the $O(1/\sqrt n)$ upper bound for positional scoring rule, top cycle, and Copeland by \citet{Baharad02:Asymptotic} works for distributions with small local correlations. \citet{Procaccia2007:Junta} proved that for weighted voters whose preferences are generated independently, positional scoring rules cannot be manipulated by $o(\sqrt n)$ manipulators almost surely as $n\ra\infty$. \citet{Xia08:Generalized} proved that any generalized scoring rule  and i.i.d.~distributions that satisfy certain conditions, a coalition of $\Theta(n^\alpha)$  manipulators is powerless when $\alpha<\frac 12$, as the likelihood for them to succeed is $O(n^{\alpha-0.5})$ ; and they are powerful when $\alpha>\frac 12$, as the likelihood is  lower bounded by $1-\exp(-\Omega(n))$ in such case.  \citet{Mossel13:Smooth} illustrated  a smooth transition from powerlessness to powerfulness for  a coalition of $c\sqrt n$ manipulators with variable $c$. \citet{Xia15:Generalized} proved that for a large class of influence problems including coalitional manipulation, under i.i.d.~distributions, with probability that goes to $1$  the number of voters needed is $0$, $\Theta(\sqrt n)$, $\Theta(n)$, or impossible. The paper does not characterize the likelihood for each case.  \citet{Durand2016:Can-a-condorcet} used Condorcification to decrease  coalitional manipulable profiles.

There is a large body of literature on the likelihood of coalitional manipulability under the {\em Impartial Anonymous Culture (IAC)}, which  assumes that each histogram happens equally likely  and resembles the flat Dirichlet distribution, based on theoretical analysis~\citep{Lepelley1987:The-proportion,Saari1990:Susceptibility,Lepelley1994:The-vulnerability,Huang2000:Analytical,Favardin2002:Borda,Favardin2006:Some,Slinko2006:How-the-size,Wilson2007:Probability,Lepelley2008:On-Ehrhart} and computer simulations~\citep{Lepelley2003:Voting,Pritchar07:Exact,Green-Armytage2016:Statistical}. Both IAC and IC are mainly of theoretical interest and  {\em ``are poor proxies of political electorates''}~\citep{Nurmi2012:On-the-Relevance}.

\vspace{2mm}\noindent{\bf\bf Other coalitional influence problems.}   $\mov$ measures the stability of elections and provides an upper bound on $\cm$ and some previous proofs of upper bounds on $\cm$ are done for $\mov$, such as~\citep{Slinko2002:Asymptotic,Xia08:Generalized}.  \citet{Pritchard2006:Average} proved that for any  positional scoring rule, the expected margin of victory under IC is $\Theta(\sqrt n)$ and characterized the voting rules with maximum expected $\mov$ under IC.   Results in \citep{Mossel13:Smooth,Xia15:Generalized} discussed above also apply to $\mov$.  \citet{Brill2022:Margin} studied the distribution of $\mov$ for some tournament rules, under a probability distribution on tournament graphs, where the direction of each edge is drawn independently and uniformly.

The likelihood of $\mov=1$ has been used to measure the {\em decisiveness} of voting, sometimes called {\em voting power}, which plays an important role in the paradox of voting~\cite{Downs1957:An-Economic} and in definitions of power indices in cooperative game theory. For two alternatives under the plurality rule, the voting power is equivalent to the likelihood of ties~\citep{Good1975:Estimating} (see, e.g.,~\cite{Xia2021:How-Likely} for a semi-random analysis on the likelihood of ties and references therein). In general, the two problems are different, for example as shown in Example~\ref{ex:ties-stable} and~\ref{ex:noties-unstable} in Appendix~\ref{app:tie-instability}.

Beyond $\cm$ and $\mov$, there is a large body of literature on the computational complexity of other types of coalition influence problems that based on vote operations, such as constructive/destructive control (making a designated alternative win/lose) by adding/deleting votes, and bribery (different changes in votes have different costs and a total budget is given). See~\citep{Faliszewski2016:Control} for a recent survey and see Appendix~\ref{app:more-CI} for their definitions. Little work has been done to analyze their likelihood of success, except~\citep{Xia15:Generalized}, which as discussed above, does not provide an accurate characterization.  The full versions of our results in the Appendix provide the first asymptotically tight bounds for these coalitional influence problems.

\vspace{2mm}\noindent{\bf Technical novelty.} While our Procedure~\ref{proc:application} is similar to the ``polyhedral approaches" adopted  in previous work  on semi-random social choice~\citep{Xia2020:The-Smoothed,Xia2021:How-Likely,Xia2021:Semi-Random}, the main technical tool of this paper (Theorem~\ref{thm:PMV-instability} for $\psi =0$ and Theorem~\ref{thm:PMV-instability-psi>0} for $\psi>0$) are  significant and non-trivial extensions of the main technical theorems in previous work, especially~\citep[Theorem~1]{Xia2021:How-Likely}, which  can be viewed as  a special case of our Theorem~\ref{thm:PMV-instability} with $\sourcepoly= \targetpoly$, $B=0$. The hardest part of the proof of Theorem~\ref{thm:PMV-instability} is  the  polynomial lower bound, whose $\sourcepoly= \targetpoly$ and $B=0$ case was proved in~\citep[Theorem~1]{Xia2021:How-Likely} by explicitly enumerating sufficiently many target integer vectors for the PMV. However, due to the generality of the PMV-instability problem, we do not see an easy way to perform a similar enumeration. To address this technical challenge, we take a different approach by first pretending that the PMV can take non-integer values, then  enumerating  (possibly non-integral) vectors that are far away from each other by exploring two directions: the direction that represents no budget (i.e., $B=0$) and the direction that represents infinite budget (i.e., $B=\infty$), and finally proving that for each such (possibly non-integral) vector, there exists a nearby integer target vector for the PMV.  The last step is achieved by carefully modeling the influence problem as an ILP on which the sensitivity analysis developed by~\citep{Cook86:Sensitivity} can be applied.  
We believe that our Theorem~\ref{thm:PMV-instability} and Theorem~\ref{thm:PMV-instability-psi>0} are useful and general tools for studying likelihood of coalitional influence, as exemplified by its applications to prove  Theorems~\ref{thm:PMV-instability-applications}--\ref{thm:CML}, Corollary~\ref{coro:IC}, and a general Theorem~\ref{thm:GSR} (Appendix~\ref{sec:semi-random-CI}).

\myparagraph{\bf Semi-random analysis and smoothed analysis in general.} Semi-random analysis~\citep{Blum1990:Some,Blum1995:Coloring} refers to the  analysis under a model where the process of generating  instances has adversarial components and random components. For example, in the smoothed complexity analysis~\citep{Spielman2009:Smoothed}, the input to an algorithm is obtained from an adversarially chosen ``ground truth'' plus a (small) random perturbation.  This can be viewed as the adversary directly choosing a distribution over data (from a set of distributions).  See~\citep{Feige2021:Introduction} for a recent survey of various semi-random models and complexity results under them. 

\section{Summary and Future Work}
We extend previous studies on average-case likelihood of coalitional manipulation in elections in three aspects: (1)  a more general semi-random model that introduces correlations among agents' preferences, (2) many other coalitional influence problems, and (3) arbitrary  coalition size, by taking the polyhedral approach (Procedure~\ref{proc:application}).  We believe that our model and results make a  non-trivial step forward, because they address  long-standing open questions  and expand  the scope of previous work to more general settings. 

Moving forward, we see three natural directions for future work. {\bf More general models:}  a natural question is how to relax the statistical independence and strict positiveness constraints on $\Pi$ to build a more general yet  tractable model. Additionally, at a high level, the contamination adversary may be viewed as another influencer. Therefore, another natural question is how to incorporate different kinds of adversaries and influencers to the model, e.g., by performing multimode control attacks~\citep{Faliszewski11:Multimode}.  {\bf Stronger theorems:} as discussed after Theorem~\ref{thm:PMV-instability}, characterizing the constants in the bounds is a natural open question.  Additionally, as discussed in the Introduction, Theorem~\ref{thm:PMV-instability-GSR-upper} poses a challenge to designing natural rules with lower likelihood of coalitional influence, or improving the $\Omega(\frac{1}{n^{-67}})$ lower bound (and extending it to $B\ge 2$). {\bf More applications:}  how to develop informative characterizations as done in Section~\ref{sec:applications} for other rules, other coalition influence problems, or other applications such as matching, resource allocation, fair division~\citep{Bai2022:Fair}, judgement aggregation~\citep{Liu2022:Semi-Random}, are important and challenging tasks.   

\newpage

\bibliographystyle{ACM-Reference-Format}
\bibliography{/Users/administrator/GGSDDU/references}

\newpage
\tableofcontents

\newpage
\appendix
%

\section{Additional Preliminaries}
 A {\em fractional} profile is a   profile $P$ coupled with a possibly non-integer and/or negative weight vector $\vec \omega_P=(\omega_R:R\in P)\in{\mathbb R}^{n}$ for the votes in $P$. 
Sometimes  the weight vector is omitted when it is clear from the context.

The {\em min-semi-random likelihood} is defined as:

\begin{equation}
\label{dfn:min-s-sat}
\satmin{X}{\Pi,\psi}(r,n,B) \triangleq  \underbrace{\inf\nolimits_{\vec\pi\in\Pi^n}}_{\begin{subarray}{l} \text{distribution}\\ \text{adversary}\end{subarray}}\underbrace{\Pr\nolimits_{P\sim\vec\pi}}_{\begin{subarray}{l} \text{preliminary}\\ \text{profile}\end{subarray}}(\underbrace{\exists P'\overset\psi\approx P}_{\begin{subarray}{l} \text{contamination}\\\text{adversary}\end{subarray}}\text{ s.t. }\underbrace{\sat{X}(r,P,B)=1}_{\text{influential}}),
\end{equation}

\subsection{Other Commonly-studied Voting Rules} 
\label{app:more-rule}
\vspace{2mm}\noindent{\bf Weighted Majority Graphs.}  For any (fractional) profile $P$ and any pair of alternatives $a,b$, let $ P[a\succ b]$ denote the total weight of votes in $P$ where $a$ is preferred to $b$. Let $\wmg(P)$ denote the {\em weighted majority graph} of $P$, whose vertices are $\ma$ and whose weight on edge $a\ra b$ is $w_P(a,b) = P[a\succ b] - P[b\succ a]$. Sometimes a distribution $\pi$ over $\ml(\ma)$ is viewed as a fractional profile, where for each $R\in\ml(\ma)$ the weight on $R$ is $\pi(R)$. In this case we let $\wmg(\pi)$ denote the weighted majority graph of the {fractional} profile represented by $\pi$. 

A voting rule is said to be {\em weighted-majority-graph-based (WMG-based)} if its winners only depend on the WMG of the input profile. In this paper we consider the following commonly-studied WMG-based rules.
\begin{itemize}
\item {\bf Copeland.} The Copeland rule is parameterized by a number $0\le \alpha\le 1$, and is therefore denoted by Copeland$_\alpha$, or $\copeland$ for short. For any fractional profile $P$, an alternative $a$ gets $1$ point for each other alternative it beats in their head-to-head competition, and gets $\alpha$ points for each tie. Copeland$_\alpha$ chooses all alternatives with the highest total score as the winners. 
\item  {\bf Maximin.} For each alternative $a$, its min-score  is defined to be $\min_{b\in\ma}w_P(a,b)$. Maximin, denoted by $\maximin$, chooses all alternatives with the max min-score as the winners.
\item {\bf Ranked pairs.} Given a profile $P$, an alternative $a$ is a winner under ranked pairs (denoted by $\rp$) if there exists a way to fix edges in $\wmg(P)$ one by one in a non-increasing order w.r.t.~their weights (and sometimes break ties), unless it creates a cycle with previously fixed edges, so that after all edges are considered, $a$ has no incoming edge. Ties between edges are broken lexicographically. For example, if  $1\ra 2$ and $2\ra 3$ have the same weight, then $1\ra 2$ is chosen first. If $1\ra 2$ and $1\ra 3$ have the same weight, then $1\ra 2$ is chosen first.
\item {\bf Schulze.} For any directed path in the WMG, its strength is defined to be the minimum weight on any single edge along the path. For any pair of alternatives $a,b$, let $s[a,b]$ be the highest weight among all paths from $a$ to $b$. Then, we write $a\succeq b$ if and only if $s[a,b]\ge s[b,a]$, and~\citep{Schulze11:New} proved that the strict version of this binary relation, denoted by $\succ$, is transitive. The Schulze rule, denoted by $\schulze$, chooses all alternatives $a$ such that for all other alternatives $b$, we have $a\succeq b$. 
\end{itemize}

\noindent{\bf STV with lexicographic tie-breaking mechanism.} The (single-winner) STV with lexicographic tie-breaking chooses winners in $m-1$ rounds. In each round, the {\em loser} of plurality under lexicographic tie-breaking is removed from the election.  We note that this rule is different from first computing STV winners under {\em parallel universe tie-breaking}~\citep{Conitzer09:Preference} and then breaking ties among the co-winners.

\vspace{2mm}\noindent{\bf Plurality with runoff.}  The plurality with runoff rule with lexicographic tie-breaking, denoted by $\plurunoff$, chooses the winner in two rounds. In the first round, the two alternatives with highest plurality scores are chosen (ties are broken lexicographically), and all other alternatives are removed. In the second round, the majority rule with lexicographic tie-breaking is applied to choose the winner. 


\begin{prop}
\label{prop:rvs-GSR}
Any representable voting system with lexicographic tie-breaking is a GSR. 
\end{prop}
\begin{proof}
The proof is similar to the proof that shows positional scoring rules with lexicographic tie-breaking are GSRs illustrated in Example~\ref{ex:borda-GSR}. Formally, we define the following score difference vector that is similar to the score difference vector defined for positional scoring rules~\citep{Xia2020:The-Smoothed}.

\begin{dfn}[\bf Score difference vector for representable voting system]\label{dfn:score-diff-rvs} 
For any scoring function $s:\ml(\ma)\times\ma\ra\mathbb R_{\ge 0}$ and any pair of different alternatives $a,b$, let $\score_{a,b}^{s}$ denote the $m!$-dimensional vector indexed by rankings in $\ml(\ma)$: for any $R\in\ml(\ma)$, the $R$-element of $\score_{a,b}^{s}$ is $s(R,a)-s(R,b)$.
\end{dfn}
Let $K={m\choose 2}$ and the hyperplanes are score difference vectors  $\{\score_{a,b}^{s}:a\in\ma, b\in\ma, a\ne b\}$. For any profile $P$, $\signH(\hist(P))$ contains information about the comparisons of total scores of all pairs of alternatives, from which $g$ chooses a winner and applies the tie-breaking mechanism when needed.
\end{proof}

\subsection{Other Commonly-studied Coalitional Influence Problems} 
\label{app:more-CI}
In the {\em constructive control by adding votes (CCAV)} (respectively, {\em destructive control by adding votes (DCAV)}) problem, we are given a distinguished alternative $d$, and we let $\ccav_{d}(r,P,B)=1$ (respectively, $\dcav_{d}(r,P,B)=1$), if  there exists a preference profile $P^*$ with $|P^*|\le B$ such that $r(P+P^*)=\{d\}$ (respectively, $r(P+P^*)\ne \{d\}$).  

In the {\em constructive control by deleting votes (CCDV)} (respectively, {\em destructive control by deleting votes (DCDV)}) problem, we are given a distinguished alternative $d$, and we let $\ccdv_{d}(r,P,B)=1$ (respectively, $\dcdv_{d}(r,P,B)=1$), if  there exists $P'\subseteq P$ with $|P'|\le B$ such that $r(P-P^*)=\{d\}$ (respectively, $r(P-P^*)\ne \{d\}$). 

For convenience, we let $\control$ denote the control problems introduced above, and let $\econtrol$ denote their effective variants, formally defined as follows.  
 \begin{dfn}[\bf Effective control problems]
 \label{dfn:control-problems}
Define
$$\econtrol = \{ \eccav,\eccdv,\edcav,\edcdv\}\text{ and }$$
$$\control =\{\ccav,\ccdv,\dcav,\dcdv\} $$
\end{dfn}

\section{Materials for Section~\ref{sec:model}}

\subsection{Full Version of the PMV-instability Problem}
\label{sec:PMV-instability-dfn-full}

\appDfn{\bf\boldmath  PMV-instability problem, full version}
{dfn:PMV-instability-problem}
{
Given an instability setting $\vosetting{}=\langle\sourcepoly, \targetpoly, \voset{}, \vec c\,\rangle$,  a set $\Pi$ of distributions over $[q]$, $n\in\mathbb N$, and $B\ge 0$, we are asked to bound
\begin{align*}
&\text{\bf  max-semi-random   instability: }\sup\nolimits_{\vec\pi\in\Pi^n}\Pr\left(\vXp\in  \nb{\csus{n,B}}{\psi}\right) \text{, and}\\
&\text{\bf  min-semi-random  instability: }\inf\nolimits_{\vec\pi\in\Pi^n}\Pr\left(\vXp\in  \nb{\csus{n,B}}{\psi}\right)
\end{align*}
}
\begin{dfn}[\bf multi-instability setting]
\label{dfn:multi-instability-setting} A {\em multi-instability setting}, denoted by $\calM = \{\vosetting^i:  i\le I\}$, is  a set of $I\in\mathbb N$ instability settings, where $\vosetting^i = \langle\sourcepoly^i,\targetpoly^i,\voset{}^i,\vec c^{\,i}\,\rangle$, whose unstable histograms  are denoted by $\csus{n,B}^{i}$. Let $\csus{n,B}^{\calM} = \bigcup_{i\le I}\csus{n,B}^{i}$.  
\end{dfn}

\begin{dfn}[\bf\boldmath The multi-instability problem]
\label{dfn:multi-instability-problem} 
Given a  multi-instability setting  $\calM = \{\vosetting{}^i: i\le I\}$,  a set $\Pi$ of distributions over $[q]$, $n\in\mathbb N$, and $B\ge 0$, we are asked to bound
\begin{align*}
&\text{\bf  max-semi-random  multi-instability: }\sup\nolimits_{\vec\pi\in\Pi^n}\Pr\left(\vXp \in\csus{n,B}^\calM\right) \text{, and}\\
&\text{\bf  min-semi-random multi-instability: }\inf\nolimits_{\vec\pi\in\Pi^n}\Pr\left(\vXp \in\csus{n,B}^\calM\right)
\end{align*} 
\end{dfn}
Like  PMV-instability problems, the max-(respectively, min-) semi-random  multi-instability represents the upper bound (respectively, lower bound) on the likelihood for the PMV to be unstable w.r.t.~any $\vosetting{}^i$ in $\calM$, when the underlying probabilities $\vec\pi$ is adversarially chosen from $\Pi^n$.

\subsection{Full Version of Lemma~\ref{lem:CI-GSR} and Its Proof}
\label{app:proof-prop-CI-GSR}
\appLem{\bf Coalitional Influence  as multi-instability, Full Version}
{lem:CI-GSR}{
For any coalitional influence problem $X\in\{\cm,\mov, \cb_{d,\vec c},\db_{d,\vec c}, \ecb_{d,\vec c},\edb_{d,\vec c}\}$ and any GSR $r$, there exist a set  $\calM = \{\vosetting^i:  i\le I\}$  of $I$ instability settings such that for every $n$-profile $P$ and every $B\ge 0$, $\sat{X}(r,P,B) =1$ if and only if $\hist(P)\in \csus{n,B}^\calM$.
}

\begin{proof} We first recall some formal notation about GSR. For any real number $x$,  let $\sign(x)\in\{+,-,0\}$ denote the sign of $x$. Given a set of $K$ hyperplanes in the $q$-dimensional Euclidean space, denoted by $\vH = (\vec h_1,\ldots,\vec h_K)$, for any $\vec x \in {\mathbb R}^q$, we let $\sign_\vH(\vec x) = (\sign(\vec x\cdot \vec h_1),\ldots, \sign(\vec x\cdot \vec h_K))$. In other words, for any $k\le K$, the $k$-th component of $\sign_\vH(\vec x)$ equals to $0$, if $\vec p$ lies in hyperplane $\vec h_k$; and it equals to $+$ (respectively, $-$) if  $\vec p$ lies in the positive (respectively, negative) side of  $\vec h_k$.
Each element in $\{+,-,0\}^K$ is called a {\em signature}.

\begin{dfn} [\bf Feasible signatures]
Given integer $\vH$ with $K = |\vH|$, let $\sk =  \{+,-,0\}^K$. A signature $\vec t\in\sk$ is {\em  feasible}, if there exists $\vec x\in \mathbb R^{m!}$  such that $\sign_\vH(\vec x) = \vec t$. 
Let $\fs \subseteq \sk$ denote the set of all feasible signatures.
\end{dfn}
The domain of any GISR $\cor$ can be naturally extended to $\mathbb R^{m!}$ and to $\fs$. Specifically, for any $\vec t\in\fs$ we let $\cor(\vec t) = g(\vec t)$. It suffices to define $g$ on the feasible signatures, i.e., $\fs$.

See~\citep[Section D.2]{Xia2021:Semi-Random} for the GSR representations of some commonly-studied voting rules, especially the rules defined in Appendix~\ref{app:more-rule}. 

Next, given $\vH$ and a feasible signature $\vec t$,  we recall from~\citep[Section D.4]{Xia2021:Semi-Random} the definition of $\ppoly{\vH,\vec t}$ that represents profiles whose signatures are $\vec t$. 
\begin{dfn}[\bf\boldmath $\ppoly{\vH,\vec t}$  ($\ppoly{\vec t}$ in short)]
\label{dfn:poly-H-t}
For any $\vH = (\vec h_1,\ldots,\vec h_K)\in (\mathbb R^{d})^K$ and any $\vec t\in \fs$, we let  $\pba{\vec t}=\left[\begin{array}{c}\pba{\vec t}_{+}\\ \pba{\vec t}_{-}\\ \pba{\vec t}_{0} \end{array}\right]$, where 
\begin{itemize}
\item $\pba{\vec t}_{+}$ consists of a row $-\vec h_i$ for each $i\le K$ with $t_i = +$.
\item $\pba{\vec t}_{-}$ consists of a row $\vec h_i$ for each $i\le K$ with $t_i = -$.
\item  $\pba{\vec t}_{0}$ consists of two rows $-\vec h_i$ and $\vec h_i$ for each $i\le K$ with $t_i = 0$.
\end{itemize}
Let $\pvbb{\vec t} = [\underbrace{-\vec 1}_{\text{for }\pba{\vec t}_{+}},\underbrace{-\vec 1}_{\text{for }\pba{\vec t}_{-}},\underbrace{\vec 0}_{\text{for }\pba{\vec t}_{0}}]$.  The corresponding polyhedron is denoted by $\ppoly{\vH,\vec t}$, or $\ppoly{\vec t}$  in short when $\vH$ is clear from the context. 
\end{dfn}

Then, we formally define some  vote operations that will be used in the proof.
\begin{dfn}  We define four vote operations as follows.
\begin{itemize}
\item {\bf Vote change}: $\vosc  = \{\hist(R_2)-\hist(R_1):R_1,R_2\in \ml(\ma) \}$.
\item {\bf Motivated vote change}: for any pair of different alternatives $a,b$, let 
$$\voset{\pm}^{a\ra b}  = \{\hist(R_b)-\hist(R_a):R_b,R_a\in \ml(\ma) \text{ and }b\succ_{R_a} a\}$$
\item {\bf Generalized vote change}: for any pair of different alternatives $a,b$, let 
$$\voset{\pm}^*  = \{\hist(R_1), \hist(R_2), \hist(R_2)-\hist(R_1):R_1,R_2\in \ml(\ma)  \}$$
\end{itemize}
\end{dfn}
We are now ready to define the instability settings whose union models the coalitional influence problems described in the statement of the proposition.
\begin{itemize}
\item {\bf \boldmath $X= \cm$.} For every pair of different alternatives $a,b$, and every pair of feasible signatures $\vec t_a,\vec t_b$ such that $r(\vec t_a) = \{a\}$ and $r(\vec t_b) = \{b\}$, $\calM$ contains
$$\langle\ppoly{\vec t_a}, \ppoly{\vec t_b}, \voset{\pm}^{a\ra b}, \vec 1\,\rangle$$

\item {\bf \boldmath $X= \mov$.} For every pair of different alternatives $a,b$, and every pair of feasible signatures $\vec t_a,\vec t_b$ such that $r(\vec t_a) = \{a\}$ and $r(\vec t_b) = \{b\}$, $\calM$ contains
$$\langle\ppoly{\vec t_a}, \ppoly{\vec t_b}, \voset{\pm}, \vec 1\,\rangle$$

\item {\bf \boldmath $X= \cb_{d,\vec c}$.} For every  every pair of feasible signatures $ \vec t,\vec t_a$ such that $r(\vec t_a) = \{a\}$ , $\calM$ contains
$$\langle\ppoly{\vec t}, \ppoly{\vec t_a}, \voset{\pm}^*, \vec c\,\rangle$$

\item {\bf \boldmath $X= \db_{d,\vec c}$.} For every  alternative $b\ne a$, and every pair of feasible signatures $\vec t,\vec t_b$ such that $r(\vec t_b) = \{b\}$, $\calM$ contains
$$\langle\ppoly{\vec t}, \ppoly{\vec t_b}, \voset{\pm}^*, \vec c\,\rangle$$

\item {\bf \boldmath $X= \ecb_{d,\vec c}$.} For every  every pair of feasible signatures $ \vec t,\vec t_a$ such that $r(\vec t)\ne \{a\}$ and $r(\vec t_a) = \{a\}$ , $\calM$ contains
$$\langle\ppoly{\vec t}, \ppoly{\vec t_a}, \voset{\pm}^*, \vec c\,\rangle$$

\item {\bf \boldmath $X= \edb_{d,\vec c}$.} For every  alternative $b\ne a$, and every pair of feasible signatures $\vec t_a,\vec t_b$ such that $r(\vec t_a) = \{a\}$ and $r(\vec t_b) = \{b\}$, $\calM$ contains
$$\langle\ppoly{\vec t_a}, \ppoly{\vec t_b}, \voset{\pm}^*, \vec c\,\rangle$$
 \end{itemize}

\end{proof}

\section{Materials for Section~\ref{sec:PMV-instability}}
\label{app:PMV-instability}

\subsection{Definitions and Conditions for Min-Semi-Random Likelihood}

For any set $\Pi^*\subseteq {\mathbb R}^q$, we define $B^-_{\Pi^*}\in\mathbb R$ to be the minimum budget $B $ such that $\Pi^*$ is completely contained in   $\sus{B }$. If no such $B $ exists, then we let $B^-_{\Pi^*} \triangleq \infty$. Formally,
\begin{equation}
\label{eq:B-min-minus}
B^-_{\Pi^*} \triangleq \inf \{B \ge 0: \Pi^*\subseteq \sus{B }\}
\end{equation}
For example, $B^-_{\conv(\Pi)}=\infty$ in Figure~\ref{fig:illustration-PMV} (b), because no matter how large $B$ is, $\sus{B}\subseteq \sourcepolyz$, and $\sourcepolyz$ does not contain all vectors in $\conv(\Pi)$.

Next, we define  notation and conditions used in the statement of the  theorem.

\appDfn{Full version}{dfn:conditions}{
Given an instability setting, $\Pi$, $B$, and $n$, define
$$d_0 = \dim(\sus{0}), d_{\infty} =  \dim(\sus{\infty})\text{, and } d_{\Delta}  = d_{\infty}  - d_0,$$
where $\dim(\sus{0})$ is the {\em dimension} of $\sus{0}$, which is the dimension of the minimal affine space that contains $\sus{0}$. We also define the following five conditions:
\begin{center}
\begin{tabular}{|c|c|c|c|c|}
\hline $\condition{1}$ & $\condition{2}$ & $\condition{3}$ & $\condition{4}$ & $\condition{5}$ \\
\hline $\csus{n,B}=\emptyset$ & $\conv(\Pi)\cap\sus{\infty}=\emptyset$ & $\conv(\Pi)\cap\sus{0}=\emptyset$ &  $\conv(\Pi)\subseteq \sus{\infty}$ & $\conv(\Pi)\subseteq\sus{0}$\\
\hline
\end{tabular}
\end{center}
}
Recall that $\conv(\Pi)$ is the convex hull of $\Pi$. Because $\sus{0}\subseteq \sus{\infty}$, $d_\Delta\ge 0$. Also notice that $\condition{2}$ implies $\condition{3}$, or equivalently, $\neg\condition{3}$ implies $\neg\condition{2}$. Similarly, $\condition{4}$ implies $\condition{5}$, or equivalently, $\neg\condition{5}$ implies $\neg\condition{4}$.

\subsection{Properties of $B_{\Pi^*}$ and $B^-_{\Pi^*}$}

\begin{claim}
\label{claim:B-min}
For any convex and compact set $\Pi^*$, if $B_{\Pi^*}\ne\infty$ then $\Pi^*\cap \sus{B_{\Pi^*}}\ne\emptyset$. 
\end{claim}
\begin{proof}
Let $\{B_j:j\in \mathbb N\}$ denote a  sequence that converges to $B_{\Pi^*}$, such that for all $j\in \mathbb N$, $\Pi^*\cap \sus{B_j}\ne\emptyset$. For any $j\in\mathbb N$, let $\vec y_j\in \Pi^*\cap \sus{B_j}$ denote an arbitrary vector. Because $\Pi^*$ is compact, a subsequence of $\{\vec y_j:j\in \mathbb N\}$, denoted by $\{\vec y_{j_i}:i\in \mathbb N\}$ converges to a vector $\vec y^*\in\Pi^*$. Notice that $\vec y_{j_i}$ is in $\Pi^*$, $\sourcepolyz$, and both are closed sets. Therefore, $\vec y^*\in  \Pi^*\cap \sourcepolyz$. 

Next, we prove that $\vec y^*\in \sus{B_{\Pi^*}}$. For every $i\in\mathbb N$, let $\vo_{j_i}\in {\mathbb R}_{\ge 0}^{|\voset{}|}$ denote the operation vector such that $ \vec c\cdot \vo_{j_i}\le B_{j_i}$ and $\vec y_{j_i} + \vo_{j_i}\times \vomatrix{}\in \targetpolyz$. Let $\vec x_{j_i} = \vec y_{j_i} + \vo_{j_i}\times \vomatrix{}$. Because $\vo_{j_i}$'s are bounded ($ \vec c\cdot \vo_{j_i}\le B_{1}$), there exists a subsequence $\{j_1':i\in\mathbb N\}$ of $\{j_1:i\in\mathbb N\}$ such that $\vo_{j_i'}$ converges to a vector $\vo^*$. It is not hard to verify that $\vec c\cdot\vo^*\le B_{\Pi^*}$ and $\{\vec x_{j_i'} =  \vec y_{j_i'}+\vo_{j_i'}\times \vomatrix{}:i\in\mathbb N\}$ converges to $\vec y^*+\vo^*\times \vomatrix{}$. Because for all $i\in\mathbb N$, $ \vec x_{j_i'}\in \targetpolyz$ and $\targetpolyz$ is closed, we have $\vec y^*+\vo^*\times \vomatrix{}\in \targetpolyz$  as well. This proves that $\vec y^*\in \sus{B_{\Pi^*}}$, which completes the proof of Claim~\ref{claim:B-min}. 
\end{proof}

\begin{claim}
\label{claim:B-min-minus}
For any bounded set $\Pi^*$, if  $\Pi^*\subseteq \sus{\infty}$ then $B^-_{\Pi^*}\ne\infty$.
\end{claim}
\begin{proof} It suffices to prove that there exists $B^*\ge 0$ so that $\Pi^*\subseteq \sus{B^*}$.
Because $\Pi^*$ is bounded, let $Q$ denote any cube that contains $\Pi^*$. Because $Q$ is a polytope and $\sus{\infty}$ is a polyhedral cone, $Q\cap \sus{\infty}$ is a polytope that contains $\Pi^*$. Let the V-representation of $Q\cap \sus{\infty}$ be $\conv(\{\vec x_1,\ldots,\vec x_k\})$ for some $k\in\mathbb N$. For every $j\le k$, because $\vec x+j\in \sus{\infty}$, there exists $B_j\ge 0$ such that $\vec x_j\in\sus{B_j}$. Let $B^*=\max\{B_1,\ldots,B_k\}$. It follows that $\Pi^*\subseteq Q\cap \sus{\infty} \subseteq \sus{B^*}$, which proves Claim~\ref{claim:B-min-minus}.
\end{proof}

\subsection{Full Version of Theorem~\ref{thm:PMV-instability} and Its Proof}
\label{app:proof-thm-PMV-instability}
\appThm{\bf\boldmath Semi-Random  PMV-Instability, $\psi =0$, full version}
{thm:PMV-instability}{
Given any $q\in\mathbb N$, any closed and strictly positive $\Pi$ over $[q]$, and any instability settings $\langle\sourcepoly,\targetpoly,\voset{},\vec c\,\rangle$, any $C_2>0$ and $C_3>0$ with $C_2<B_{\conv(\Pi)}<C_3$,  any $n\in \mathbb N$, and any $B\ge 0$,  
$$ 
 \begin{array}{r@{}|l|@{}l@{}|l|}
\cline{2-4} & \text{\bf Name} &  \text{\bf\  Likelihood}&  \text{\bf Condition}\\
\cline{2-4} \multirow{6}{*}{$\sup_{\vec\pi\in\Pi^n}\Pr\left(\vXp \in \csus{n,B}\right)= \left\{\begin{array}{@{}r@{}}\\\\\\\\\\\end{array}\right.$}& \text{0 case}& \ 0 &\condition{1}\\
\cline{2-4}& \text{exp case}& \ \exp(-\Theta(n)) &\neg \condition{1}\wedge \condition{2}\\
\cline{2-4}&  \text{PT-$\Theta(\sqrt n)$}& \ \Theta\left(\dfrac{\min\{B+1,\sqrt n\}^{d_{\Delta}}}{(\sqrt n)^{q-d_0}} \right) &\neg \condition{1}\wedge \neg \condition{3}\\
\cline{2-4}& \text{PT-$\Theta(n)$} & \begin{array}{ l l } \exp(-\Theta(n)) &\text{if }B\le C_2n \\ \hline \Theta\left( (\frac{1}{\sqrt n})^{q-d_\infty}  \right) &\text{if }B\ge C_3n\end{array}&\begin{array}{@{}l}\text{otherwise, i.e.,}\\ \neg \condition{1}\wedge \neg \condition{2}\wedge \condition{3}\end{array}\\
\cline{2-4}
\end{array}
$$
For any $C^-_2<B^-_{\conv(\Pi)}<C^-_3$,  any $n\in \mathbb N$, and any $B\ge 0$,  
$$ 
\begin{array}{r@{}|l|@{}l@{}|l|}
\cline{2-4} & \text{\bf Name} &  \text{\bf\  Likelihood}&  \text{\bf Condition}\\
\cline{2-4} \multirow{6}{*}{$\inf_{\vec\pi\in\Pi^n}\Pr\left(\vXp \in \csus{n,B}\right)= \left\{\begin{array}{@{}r@{}}\\\\\\\\\\\end{array}\right.$}& \text{0 case}& \ 0 &\condition{1}\\
\cline{2-4}& \text{exp case}& \ \exp(-\Theta(n)) &\neg \condition{1}\wedge \condition{4}\\
\cline{2-4}&  \text{PT-$\Theta(\sqrt n)$}& \ \Theta\left(\dfrac{\min\{B+1,\sqrt n\}^{d_{\Delta}}}{(\sqrt n)^{q-d_0}} \right) &\neg \condition{1}\wedge \neg \condition{5}\\
\cline{2-4}& \text{PT-$\Theta(n)$} & \begin{array}{ l l } \exp(-\Theta(n)) &\text{if }B\le C^-_2n \\ \hline \Theta\left( (\frac{1}{\sqrt n})^{q-d_\infty}  \right) &\text{if }B\ge C^-_3n\end{array}&\begin{array}{@{}l}\text{otherwise, i.e.,}\\ \neg \condition{1}\wedge \neg \condition{4}\wedge \condition{5}\end{array}\\
\cline{2-4}
\end{array}
$$
}

\begin{proof}  We first prove the  $\sup$ part  of the theorem, then leverage the techniques  to prove the \myhyperlink{inf}{$\inf$ part}. 

\vspace{3mm}
\noindent  {{\bf\boldmath Proof for the $\sup$ part.}} 
For convenience, hyperlinks (in red) to the proofs of four cases are provided as follows. 

\renewcommand{\arraystretch}{1.5}
$$\sup\limits_{\vec\pi\in\Pi^n}\Pr\left(\vXp \in \csus{n,B}\right) = \left\{\text{
\begin{tabular}{|l| c|c|c|@{}c@{}|}
\hline 
 \myhyperlink{PMV-proof-sup-0}{$0$ case}&     $0$ &  \multicolumn{3}{c|}{$\condition{1}$}\\
\hline \myhyperlink{PMV-proof-sup-exp}{  exponential case} &  $\exp(-\Theta(n))$  &  \multicolumn{2}{c|}{$\condition{2}$}  & \multirow{4}{*}{$\neg\condition{1}$}\\
\hline \multirow{2}{*}{ \myhyperlink{PMV-proof-sup-n}{PT-$\Theta(n)$-$\sup$}} &    $\exp(-\Theta(n))$  &  \myhyperlink{PMV-proof-sup-n-smallB}{$B\le C_2n$} & $\neg\condition{2}\wedge \condition{3}$ &  \\
\cline{2-4} &     $\Theta\left((\frac{1}{\sqrt n})^{q-d_\infty}\right) $ & \myhyperlink{PMV-proof-sup-n-largeB}{$B\ge C_3n$} & $\neg\condition{2}\wedge \condition{3}$ &  \\
\hline \myhyperlink{PMV-proof-sup-sqrtn}{PT-$\Theta(\sqrt n)$-$\sup$ }  &  $\Theta\left( \dfrac{\min\{B,\sqrt n\}^{d_{\Delta}}} {(  \sqrt n )^{q-d_0}}  \right)$  & \multicolumn{2}{c|}{ $\neg\condition{3}$} &\\
\hline
\end{tabular}}\right.
$$
\renewcommand{\arraystretch}{1}

\vspace{3mm}
\noindent \hypertarget{PMV-proof-sup-0}{{\bf\boldmath Proof for the $0$ case of $\sup$}} is straightforward, because $ \csus{n,B}=\emptyset$ states that the PMV-instability problem does not have a size-$n$ non-negative integer solution. In the rest of the proof for $\sup$, it suffices prove the exponential case and the polynomial case for all $n$ that are larger than a constant $N$.   This is because  for any $n$  such that the $0$ case does not hold (which means that $\csus{n,B}\ne\emptyset$), and for every $\vec\pi\in\Pi^n$,
$$\Pr\left(\vXp \in \csus{n,B}\right)\in [\epsilon^n,1]$$
Therefore, for every $n$ below a constant $N$, we have 
$$\sup_{\vec\pi\in\Pi^n}\Pr\left(\vXp \in \csus{n,B}\right)\in [\epsilon^N,1],$$
\vspace{3mm}
\noindent \hypertarget{PMV-proof-sup-exp}{{\bf\boldmath Proof for the exponential  case  of $\sup$.}}  
Because $\conv(\Pi)\cap\sus{\infty}=\emptyset$, $\conv(\Pi)$ is convex and compact, and $\sus{\infty}$ is convex,  due to the strict hyperplane separation theorem, for every $\vec\pi\in\Pi^n$, $\sum_{i=1}^{n}\pi_i$ is $\Omega(n)$ away from any vector in $\sus{\infty}$, which means that $\sum_{i=1}^{n}\pi_i$ is $\Omega(n)$ away from any vector in $\cpoly{\infty}$, because $\sus{\infty}$ is the characteristic cone of $\cpoly{\infty}$ as proved in the following claim. 

\begin{claim}
\label{claim:cc-h-inf}
The characteristic cone of $\cpoly{\infty}$ is $\sus{\infty}$.
\end{claim}
\begin{proof} We prove  two general observations about characteristic cones. For each $i\in\{1,2\}$, let $\ppoly{i}$ denote a polyhedron whose V-representation is $\calV_i+ \ppolyz{i}$, where $\calV_i$ is a convex polytope and $\ppolyz{i}$ is the characteristic cone of $\ppoly{i}$. 
\begin{itemize}
\item []{\bf Observation 1.} The characteristic cone of $\ppoly{1}+\ppoly{2}$ is $\ppolyz{1}+\ppolyz{2}$. This is because
$$\ppoly{1}+\ppoly{2} = (\calV_1+\calV_2) + (\ppolyz{1}+\ppolyz{2})$$
Here $\ppolyz{1}+\ppolyz{2}$ is indeed a finitely generated cone, because suppose for $i\in\{1,2\}$, $\ppolyz{i}$ is the convex cone generated from $\calB_i$, then it is not hard to verify that $\ppolyz{1}+\ppolyz{2}$ is a cone generated by $\calB_1\cup \calB_2$.
\item []{\bf Observation 2.} If $\ppoly{1}\cap \ppoly{2}\ne \emptyset$, then its characteristic cone  is $\ppolyz{1}\cap \ppolyz{2}$. This   is proved by the H-representations of  $\ppoly{1}$ and $\ppoly{2}$. Suppose for each $i\in\{1,2\}$, $\ppoly{i} = \{\vec x: \ba_i\times\invert{\vec x}\le \invert{\vbb_i}\}$, which means that $\ppolyz{i} = \left\{\vec x: \ba_i\times\invert{\vec x}\le \invert{\vec 0}\right\}$. Then, we have 
$$\ppoly{1}\cap \ppoly{2} = \left\{\vec x: \multimatrix{\ba_1\\ \ba_2}\times \invert{\vec x}\le \invert{\vbb_1,\vbb_2}\right \},$$
whose characteristic cone is 
$\left\{\vec x: \multimatrix{\ba_1\\ \ba_2}\times \invert{\vec x}\le \invert{\vec 0}\right \}=\ppolyz{1}\cap \ppolyz{2}$.
\end{itemize}
Recall that $\cpoly{\infty} = \sourcepoly\cap \left(\targetpoly + \calQ_{\infty}\right)$.  By Observation 1, the characteristic cone of $\targetpoly + \calQ_{\infty}$ is $\targetpolyz  + \calQ_{\infty}$. Then, by Observation 2, the characteristic cone of $\cpoly{\infty}$ is $\sourcepolyz\cap \left(\targetpolyz+ \calQ_{\infty}\right)$, which is $\sus{\infty}$ (due to (\ref{eq:cone-infty})).
\end{proof}
Then, the upper bound in the exponential  case  of $\sup$ follows after a straightforward application of Hoeffding's inequality and the union bound (applied to all $q$ dimension). More precisely, recall that $\vec\pi$ is strictly positive, then Hoeffding's inequality implies that for every   $i\in[q]$, the probability for the $i$th dimension of $\vXp$ to be more than $\Omega(n)$ away from the $i$th dimension of $\sum_{j=1}\pi_j$ is exponentially small. Therefore, according to the union bound, the probability for the $L_\infty$ distance between $\vXp$ and $\sum_{j=1}\pi_j$ to be $\Omega(n)$ is exponentially small, which implies that the probability for $\vXp$ to be in $\csus{n,B}$ is exponentially small.  The lower bound  in the exponential  case  is straightforward, because for any $\vec\pi\in\Pi^n$ (recall that all distributions in $\vec\pi$ are strictly positive) and any $\vec x\in\csus{n,B}$, $\Pr(\vXp=\vec x) = \exp(-\Theta(n))$.

\vspace{3mm}
\noindent \hypertarget{PMV-proof-sup-sqrtn}{{\bf \boldmath Proof for the phase transition at $\Theta(\sqrt n)$ case of $\sup$  (PT-$\Theta(\sqrt n)$-$\sup$ for short).}} The proof proceeds in the following three steps: 
\begin{itemize}
\item Prove the \myhyperlink{PMV-proof-sup-sqrtn-ub-smallB}{polynomial upper bound for $B\le \sqrt n$}, i.e., $O\left(\dfrac{(B+1)^{d_{\Delta}}}{(\sqrt n)^{q-d_0}}  \right)$.
\item Prove the \myhyperlink{PMV-proof-sup-sqrtn-ub-largeB}{polynomial upper bound for $B> \sqrt n$}, i.e.,
$O\left( (\frac{1}{\sqrt n})^{q-d_\infty}  \right) $.
\item Prove the asymptotically matching \myhyperlink{PMV-proof-sup-sqrtn-lb}{polynomial lower bound}, i.e., $\Omega\left(\dfrac{\min\{B+1,\sqrt n\}^{d_{\Delta}}}{(\sqrt n)^{q-d_0}}  \right)$
\end{itemize}
We first introduce some notation and assumptions that will be used in the proofs. Let $\ba_0 = \left[\begin{array}{c}\ba_{\text{S}}\\\ba_{\text{T}}\end{array}\right]$ and let $\ba_\infty$ denote an integer matrix  that characterizes $\sus{\infty}$. That is,
$$\sus{\infty} = \left\{\vec x\in {\mathbb R}^q:\ba_\infty\times \invert{\vec x} \le  \invert{\vec 0}\right\}$$
The existence of such $\ba_\infty$ is due to~\citep[Proposition~3.12]{Conforti2014:Integer}, which states that any polyhedron that has a rational H-representation has a rational V-representation, and vice versa. More precisely,  because $\targetpolyz$ has a rational H-representation, it has a rational V-representation, denoted by $\cone(\{\vec x_1,\ldots,\vec x_k\})$, where $\{\vec x_1,\ldots,\vec x_k\}\subseteq {\mathbb Q}^q$. Then, we have 
$$\targetpolyz+\calQ_{\infty} = \cone(\{\vec x_1,\ldots,\vec x_k\}\cup \voset{}),$$
which means that $\targetpolyz+\calQ_{\infty}$ can be represented by a set of linear inequalities with rational coefficients, due to~\citep[Proposition~3.12]{Conforti2014:Integer}. Consequently, $\sourcepolyz\cap(\targetpolyz+\calQ_{\infty})$ can be represented by a set of linear inequalities with rational coefficients, by combining the linear inequalities for $\sourcepolyz$ and the linear inequalities for $\targetpolyz+\calQ_{\infty}$.

Let $\ba_0^=$ and $\ba_\infty^=$ denote the implicit equalities of $\ba_0$ and $\ba_\infty$, respectively. We have $\rank(\ba_0^=) = q-d_0$ and $\rank(\ba_\infty^=) = q-d_\infty$ (\cite[Theorem 3.17]{Conforti2014:Integer}).  Next, we show that, without loss of generality,  in the rest of the proof for   PT-$\Theta(\sqrt n)$-$\sup$, we can assume that $\vec 1$ cannot be represented as a linear combination of rows in $\ba_0^=$  
or a linear combination of  rows in $\ba_\infty^=$. Formally,
\begin{asmp}\label{asmp:indwith1}
$\rank\left(\left[\begin{array}{c}\ba_0^=\\\vec 1\end{array}\right]\right) = q-d_0+1$.
\end{asmp}
\begin{asmp}\label{asmp:indwith2}
$\rank\left(\left[\begin{array}{c}\ba_\infty^=\\\vec 1\end{array}\right]\right) = q-d_\infty+1$.
\end{asmp}
To see that we can assume Assumption~\ref{asmp:indwith1}, suppose for the sake of contradiction that Assumption~\ref{asmp:indwith1} does not hold, which means that $\vec 1$ is a linear combination of rows in $\ba_0^=$. Then, for every $\vec x\in \sus{0}=\sourcepolyz\cap \targetpolyz$, we have $\vec x\cdot 1 = 0$. Therefore,  $\conv(\Pi)\cap \sus{0}=\emptyset$, which contradicts $\neg\condition{3}$. Similarly, if Assumption~\ref{asmp:indwith2} does not hold, then we have $\conv(\Pi)\cap \sus{\infty}=\emptyset$, which again contradicts $\neg\condition{3}$, because $\sus{0}\subseteq \sus{\infty}$.  
 
\vspace{3mm}
\noindent \hypertarget{PMV-proof-sup-sqrtn-ub-smallB}{{\bf \boldmath Proof for the polynomial upper bound of PT-$\Theta(\sqrt n)$-$\sup$, $B\le \sqrt n$.}} 

\vspace{2mm}\noindent{\bf\bf \bf\boldmath Overview of proof.} The proof proceeds in three steps. 
In \myhyperlink{ub-step1}{Step 1}, we use  $\ba_{0}^=$ and $\ba_\infty^=$ to define a partition of $[q]$ into three sets $I_{0+}, I_{0-}$, and $I_1$, which contain $q-d_0+1$, $d_\infty-d_0$, and $d_\infty-1$ numbers, respectively. For convenience, we rename the coordinates so that 
$$\underbrace{1,\ldots, q-d_\infty+1}_{I_{0+}},\underbrace{q-d_\infty+2, \ldots,q-d_0+1}_{I_{0-}},\underbrace{q-d_0+2, \ldots, q}_{I_1}$$
Let $I_{1+} = I_{0-}\cup I_1$ and let $I_0 = I_{0+}\cup I_{0-}$. \myhyperlink{ub-step2}{Step 2} proves two properties of the partition. Let $\cpolyn{B} =  \{\vec x\in \cpoly{B}:\vec x\cdot\vec 1 = n\}$ denote the subset of $\cpoly{B}$ that consists of all size-$n$ vectors.  First, in \myhyperlink{ub-step2.1}{Step 2.1}, we prove  that given the $I_{1+}$ coordinates of vectors in $\cpolyn{B}$, each of its remaining coordinates (in $I_{0+}$) can take no more than $O(1)$ integer values. Second, in \myhyperlink{ub-step2.2}{Step 2.2}, we prove that given the $I_{1}$ coordinates of vectors in $\cpolyn{B}$, each of its remaining coordinates (in $I_{0}$) can take no more than $O(B)$ integer values. 

In light of \myhyperlink{ub-step2.1}{Step 2.1} and \myhyperlink{ub-step2.2}{Step 2.2}, we can enumerate integer vectors $\vec y$ in $\cpolyn{B}$ as follows: first, we fix the $I_1$ coordinates of $\vec y$; second, each of the $d_\Delta=d_\infty-d_0$  coordinates in $I_{0-}$ takes  no more than $O(B)$ integer values; and finally, each of the $q-d_\infty+1$  coordinates in $I_{0+}$ takes  no more than $O(1)$ integer values. Then in \myhyperlink{ub-step3}{Step 3}, we leverage this enumeration method with the Bayesian network representation and the point-wise anti-concentration bound in~\citep{Xia2021:How-Likely} to prove the upper bound. 
 
\vspace{2mm}
\noindent{\bf\boldmath \hypertarget{ub-step1}{Step 1} of poly upper bound: Define the partition $[q]=I_{0+}\cup I_{0-}\cup I_{1}$.} 
Let $\calP_0$ and $\calP_\infty$ denote the affine hulls of $\sus{0}$ and $\sus{\infty}$, respectively (which are the same as the linear spaces generated by $\sus{0}$ and $\sus{\infty}$, because both contains $\vec 0$). It follows from~\citep[Theorem 3.17]{Conforti2014:Integer} that
\begin{equation}
\label{eq:calP}\calP_0 = \left\{\vec x\in {\mathbb R}^q:\ba_0^=\times \invert{\vec x} = \invert{\vec 0}\right\}\text{ and } \calP_\infty = \left\{\vec x\in {\mathbb R}^q:\ba_\infty^=\times \invert{\vec x} = \invert{\vec 0}\right\}
\end{equation}
That is, $\calP_0$ (respectively, $\calP_\infty$) consists of vectors that satisfies of the implicit equalities of $\ba_0$ (respectively, $\ba_\infty$).

Let $\ba_* = \left[\begin{array}{c}\ba_\infty^=\\\ba_0^=\end{array}\right]$. We  will define a partition of $[q]$ as $I_{0+}\cup I_{0-}\cup I_1$, where $|I_{0+}| = q- d_\infty+1$, $|I_{0-}| = d_\Delta$, and $|I_1| =  d_0-1$, and the partition satisfies the following two conditions.  
\begin{itemize}
\item {\bf \hypertarget{cond1}{Condition 1}.} The $I_{0+}$ columns of $\multimatrix{\ba_\infty^=\\\vec 1}$ are linearly independent.
\item {\bf \hypertarget{cond2}{Condition 2}.} The $I_{0+}\cup I_{0-}$ columns of $\multimatrix{\ba_0^=\\\vec 1}$ are linearly independent.
\end{itemize}
We first define two sets $I_{0+}'$ and $I_{0-}'$ as follows.

\vspace{2mm}\noindent{\bf \boldmath Define $I_{0+}'$.} Recall that $\rank(\ba_\infty^=) = q-d_{\infty}$. Therefore,  $\ba_\infty^=$ contains a set of $q-d_{\infty}$ linearly independent column vectors, whose indices are denoted by $I_{0+}'\subseteq [q]$. W.l.o.g.~let $I_{0+}' = \{1,\ldots, q-d_{\infty}\}$---if this is not the case, then we shift the $I_{0+}'$ columns in $\ba_\infty^=$ to be the first $q-d_{\infty}$ columns and rename  the coordinates.

\vspace{2mm}\noindent{\bf \boldmath Define $I_{0-}'$.} 
Notice that the  $I_{0+}'$ columns of $\ba_*$ are linearly independent (because their $\ba_\infty^=$ parts are already linearly independent). Because $\sus{0}\subseteq \sus{\infty}$, we have $\calP_0\subseteq \calP_\infty$, which means that
$$\left\{\vec x\in {\mathbb R}^q:\ba_* \times \invert{\vec x} = \invert{\vec 0}\right\} = \calP_\infty \cap \calP_0= \calP_0$$
This means that $\rank(\ba_*)=\rank(\ba_0^=) = q-d_0$. Consequently, there exist a set of $q-d_0-(q-d_\infty)= d_{\infty}-d_0=d_\Delta$ columns of $\ba_*$, whose indices are denoted by $I_{0-}'\subseteq ([q]\setminus I_{0+}')$, such the $I_{0+}' \cup I_{0-}'$ columns of $\ba_*$ are linearly independent. W.l.o.g., let $I_{0-}' = \{q-d_{\infty}+1,\ldots, q-d_0\}$. 

Notice that the $I_{0+}'$ columns of $\left[\begin{array}{c}\ba_\infty^=\\\vec 1 \end{array}\right]$ are  linearly independent (because their $\ba_\infty^=$ parts are already linearly independent). Let $J$ denote the indices to columns of $\left[\begin{array}{c}\ba_\infty^=\\\vec 1 \end{array}\right]$ that are linearly independent with the $I_{0+}'$ columns. That is, 
\begin{equation}
\label{eq:dfn-J}
J = \left\{i_+\in ([q]\setminus I_{0+}'): \text{the }I_{0+}'\cup\{i_+\}\text{ columns of } \left[\begin{array}{c}\ba_\infty^=\\\vec 1 \end{array}\right] \text{ are linearly independent}\right\}
\end{equation}
By definition, we have $J\ne \emptyset$, because according to Assumption~\ref{asmp:indwith2},
$$\rank\left(\left[\begin{array}{c}\ba_\infty^=\\\vec 1 \end{array}\right]\right) =\rank\left( \ba_\infty^=\right) +1 =  q-d_{\infty}+1>q-d_{\infty} = |I_{0+}'|$$ 
Next, we define two specific columns: $i_+\in J$ and $i_-$  in the following two cases ( $i_+=i_-$ in case 2), prove that the $I_{0+}'\cup I_{0-}'\cup\{i_-\}$ columns of $\left[\begin{array}{c}\ba_{*}\\\vec 1 \end{array}\right]$  are  linearly independent (in Claim~\ref{claim:partition-property}), and then use them  to define $I_{0+}, I_{0-}$ and $I_1$.
\begin{itemize}
\item {\bf \boldmath Case 1: $J\cap  I_{0-}'\ne \emptyset$.}  Let $i_+$ denote an arbitrary number in $J\cap  I_{0-}'$ and let $i_-\in ([q]\setminus (I_{0+}'\cup I_{0-}'))$ denote an arbitrary number such that the $I_{0+}'\cup I_{0-}'\cup\{i_-\}$ columns of $\left[\begin{array}{c}\ba_{*}\\\vec 1 \end{array}\right]$  are  linearly independent. The existence of such $i_-$ is guaranteed by the following two observations. First, according to the definitions of $I_{0+}'$ and $I_{0-}'$, the   $I_{0+}'\cup I_{0-}'$ columns of $\left[\begin{array}{c}\ba_{*}\\\vec 1 \end{array}\right]$ are linearly independent.  Second, according to Assumption~\ref{asmp:indwith1}, 
$$\rank\left(\left[\begin{array}{c}\ba_*\\\vec 1 \end{array}\right]\right) \ge \rank\left(\left[\begin{array}{c}\ba_0^=\\\vec 1 \end{array}\right]\right) = \rank\left( \ba_0^=\right) +1 =  q-d_0+1>q-d_0 = |I_{0+}'\cup I_{0-}'|$$ 
\item {\bf  \boldmath  Case 2: $J\cap  I_{0-}' = \emptyset$.} Choose any $i_+\in J\subseteq I_1$ and let $i_- = i_+$. We now prove that  the $I_{0+}'\cup I_{0-}' \cup \{i_-\}$ columns of $\left[\begin{array}{c}\ba_{*}\\\vec 1 \end{array}\right]$ are linearly independent. Suppose for the sake of contradiction this is not true, which means that column $i_-$ of $\left[\begin{array}{c}\ba_{*}\\\vec 1 \end{array}\right]$ can be written as an affine combination of the $I_{0+}'\cup I_{0-}'$ columns of $\left[\begin{array}{c}\ba_{*}\\\vec 1 \end{array}\right]=\left[\begin{array}{c}\ba_{\infty}\\\ba_{0}\\\vec 1 \end{array}\right]$. This mean that column $i_-$ of $\left[\begin{array}{c}\ba_\infty^=\\\vec 1 \end{array}\right]$ can be written as the same affine combination of the $I_{0+}'\cup I_{0-}'$ columns  of $\left[\begin{array}{c}\ba_\infty^=\\\vec 1 \end{array}\right]$. Recall that $J\cap  I_{0-}' = \emptyset$, which  means that in matrix $\left[\begin{array}{c}\ba_\infty^=\\\vec 1 \end{array}\right]$, each column   in $I_{0-}'$  is an affine combination of the $I_{0+}'$ columns  of $\left[\begin{array}{c}\ba_\infty^=\\\vec 1 \end{array}\right]$. Therefore, in  $\left[\begin{array}{c}\ba_\infty^=\\\vec 1 \end{array}\right]$, column $i_-$ is linearly {\em dependent} with the $I_{0+}'$ columns. This contradicts  the definition of $i_-$, which is the same as $i_+\in J$.
\end{itemize}
Notice that in both cases, the following claim holds.
\begin{claim}
\label{claim:partition-property}
The $I_{0+}'\cup I_{0-}'\cup\{i_-\}$ columns of $\left[\begin{array}{c}\ba_{*}\\\vec 1 \end{array}\right]$  are  linearly independent.
\end{claim}

\vspace{2mm}\noindent{\bf \boldmath Define $I_{0+}$, $I_{0-}$, and $I_{1}$.} Given $I_{0+}'$, $I_{0-}'$, $i_+$, and $i_-$ defined above, we are now ready to define $I_{0+}$, $I_{0-}$, and $I_{1}$ as follows. Let 
$$I_{0+} = I_{0+}'\cup \{i_+\}, I_{0-} = I_{0-}'\cup \{i_-\}\setminus \{i_+\}, \text{ and }I_1 = [q]\setminus (I_{0+} \cup I_{0-} )$$
By definition, we have $|I_{0+}| = q-d_\infty+1$, $|I_{0+}| =  d_0-d_\infty$, and $|I_{1}| = d_0-1$.  For convenience, we rename the coordinates so that 
$$I_{0+} = \{1,\ldots, q-d_{\infty}+1\},I_{0-} = \{q-d_{\infty}+2,\ldots,q-d_0+1\} \text{ and }I_1 = \{q-d_0+2,\ldots,q\}$$

\vspace{3mm}
\noindent{\bf \boldmath Verify \myhyperlink{cond1}{Condition 1}.}    Because $i_+\in J$, the $I_{0+} = I_{0+}'\cup \{i_+\}$ columns of $\multimatrix{\ba_\infty^=\\\vec 1}$ are linearly independent (due to the definition of $J$ in (\ref{eq:dfn-J})), which means that \myhyperlink{cond1}{Condition 1} is satisfied.  

\vspace{3mm}
\noindent{\bf \boldmath Verify \myhyperlink{cond2}{Condition 2}.}   Recall that when defining $i_-$, we  proved that the $I_{0+}  \cup I_{0-} = I_{0+}'  \cup I_{0-}'\cup \{i \} $ columns of $\multimatrix{\ba_*\\\vec 1}$ are linearly independent. Because $\rank\left(\multimatrix{\ba_\infty^= \\\ba_0^=}\right)=\rank(\ba_*)=\rank(\ba_0^=) = q-d_0$, each row in $\ba_\infty^=$ can be written as an affine combination of rows in $\ba_0^=$. Therefore, if  some linear combination of the $I_{0+}  \cup I_{0-} = I_{0+}'  \cup I_{0-}'\cup \{i \} $ columns of $\multimatrix{\ba_0^=\\\vec 1}$ equals to $\vec 0$, then the same linear combination of the $I_{0+}  \cup I_{0-} = I_{0+}'  \cup I_{0-}'\cup \{i \} $ columns of $\multimatrix{\ba_*\\\vec 1}$ equals to $\vec 0$ as well, which is a contradiction to Claim~\ref{claim:partition-property}. This verifies \myhyperlink{cond2}{Condition 2}.

In the remainder of the proof, we let
$$I_{0} = I_{0+}  \cup I_{0-} \text{ and }I_{1+} = I_{1}  \cup I_{0-}$$
This leads to two partitions of $[q]$, i.e., $[q] = I_0\cup I_1 = I_{0+}\cup I_{1+}$, which will be used in the next step.

\vspace{2mm}
\noindent{\bf\boldmath \hypertarget{ub-step2}{Step 2}  of poly upper bound: Bound the width of coordinates in $I_{0+}$ and $I_{0}$.} 
For any polyhedron $\poly\subseteq \mathbb R^{q}$, any $I\subseteq[q]$, any $\vec y_I\in {\mathbb R}^{I}$, and any $i\in ([q]\setminus I)$, let $\range{i}{\poly, {\vec y_I}}$ denote the difference between the maximum value of the $i$-th component of vectors in $\poly$ whose $I$-components are $\vec y_I$ and the minimum value of the $i$-th component of vectors in $\poly$ whose $I$-components are $\vec y_I$. Formally,
$$\range{i}{\poly, {\vec y_I}} = \max\nolimits_{\vec x\in \poly: [\vec x]_{I} = \vec y_I}[\vec x]_{i}-\min\nolimits_{\vec x\in \poly: [\vec x]_{I} = \vec y_I}[\vec x]_{i}$$
Recall that $\cpolyn{B} = \{\vec x\in\cpoly{B}: \vec x\cdot \vec 1 = n\}$. 
In \myhyperlink{ub-step2.1}{Step 2.1} and \myhyperlink{ub-step2.2}{Step 2.2}, we bound $\range{i}{\cpolyn{B}, {\vec y_I}}$ for $I = I_{1+}=I_{0-}\cup I_1$ and $I = I_1$, respectively.

\vspace{2mm}
\noindent{\bf\boldmath \hypertarget{ub-step2.1}{Step 2.1}  of poly upper bound: Bound the width of coordinates in $I_{0+}$.} In this step, we prove that 
there exists a constant $C^*$ such that for any $B$, any $n$, any $\vec y_{I_{1+}}\in {\mathbb R}^{I_{1+}}$, and any $i\in  I_{0+}$,  
\begin{equation*}
\range{i}{\cpolyn{B}, {\vec y_{I_{1+}}}}\le C^*
\end{equation*}
Notice that $\cpolyn{B}\subseteq \cpolyn{\infty} = \{\vec x\in\cpoly{\infty}:\vec x\cdot\vec 1 = n\}$. Therefore, it suffices to prove the following stronger inequality.
\begin{equation}
\label{eq:range-I1+}
\range{i}{\cpolyn{\infty}, {\vec y_{I_{1+}}}}\le C^*
\end{equation}
According to the V-representation of $\cpoly{\infty}$, for any $\vec y = (\vec y_{I_{0+},\vec y_{I{1+}}})\in \cpolyn{\infty}\subseteq \cpoly{\infty}$, we can write $\vec y = \vec v + \vec x$, where $\vec v = (\vec v_{I_{0+}},\vec v_{I_{1+}})$ is in a convex polytope  and $\vec x = (\vec x_{I_{0+}},\vec x_{I_{1+}})$ is in the characteristic cone of $\cpoly{\infty}$. Let $n' = \vec x\cdot 1$. Next, we use Gauss-Jordan elimination to define a matrix $\bd_\infty $ based on $\left[\begin{array}{c}\ba_\infty^=\\ \vec 1\end{array}\right]$, such that
\begin{equation}
\label{eq:x-rep}
\vec x_{I_{0+}} = (\vec x_{I_{1+}}, n')\times \bd_\infty 
\end{equation}

By Claim~\ref{claim:cc-h-inf}, we have $ \vec x\in \sus{\infty} \subseteq \calP_\infty$.  Recall from \myhyperlink{cond1}{Condition 1} that the first $q-d_{\infty}+1$ columns (i.e., the $I_{0+}$ columns) of $\left[\begin{array}{c}\ba_\infty^=\\ \vec 1\end{array}\right]$ are linearly independent, and recall from Assumption~\ref{asmp:indwith2} that $\rank\left(\multimatrix{\ba_\infty^=\\ \vec 1}\right) = q-d_{\infty}+1$.  Therefore, Gauss-Jordan elimination on 
$\left[\begin{array}{c}\ba_\infty^=\\ \vec 1\end{array}\right]\times \invert{\vec x} = \invert{\vec 0, n'}$ leads to a $d_{\infty}\times (q-d_{\infty}+1)$ matrix $\bd_\infty$ such that 
$$\left[\begin{array}{c}\ba_\infty^=\\ \vec 1\end{array}\right]\times \invert{\vec x} = \invert{\vec 0, n'}\text{ if and only if }\vec x_{I_0+} =  (\vec x_{I_{1+}}, n')\times \bd_\infty,$$
which proves (\ref{eq:x-rep}).


  Let $C_{\max}$ denote the maximum $L_\infty$ norm of  vectors in $\mV$, which means that $|\vec v|_\infty\le C_{\max}$. Then,
$$\left|n- n'\right| = \left|\vec y\cdot\vec 1-\vec x\cdot\vec1\right| = |\vec v\cdot \vec 1|   \le qC_{\max}$$ 

 Let $\hat C$ denote the maximum absolute value of entries in $\bd_\infty$ and let $C^*=C_{\max} + 2qC_{\max}\hat C$. We prove that $\vec y_{I_{0+}}$ is in a $C^*$ neighborhood of $(\vec y_{I_{1+}},n)\times \bd_\infty$ in $L_\infty$ as follows.
\begin{align*}
& | \vec y_{I_{0+}}- (\vec y_{I_{1+}},n)\times \bd_\infty|_\infty =  |{\vec v_{I_{0+}}+\vec x_{I_{0+}}}- (\vec v_{I_{1+}}+\vec x_{I_{1+}},n)\times \bd_\infty|_\infty \\
=&  | \vec v_{I_{0+}}+(\vec x_{I_{1+}}, n')\times \bd_\infty - (\vec v_{I_{1+}}+\vec x_{I_{1+}},n)\times \bd_\infty|_\infty   &\text{by (\ref{eq:x-rep})}\\
=& | \vec v_{I_{0+}} -  (\vec v_{I_{1+}},n-n')\times \bd_\infty|_\infty\\
\le& C_{\max} + 2qC_{\max}\hat C = C^*
\end{align*}
This proves (\ref{eq:range-I1+}) and completes \myhyperlink{ub-step2.1}{Step 2.1}. 

\vspace{2mm}
\noindent{\bf\boldmath \hypertarget{ub-step2.2}{Step 2.2} of poly upper bound:  Bound the width of coordinates in $I_{0}$.} In this step, we prove that 
there exists a constant $C^*$ such that for any $B\ge 0$, any $n$, any $\vec y_{I_{1}}\in {\mathbb R}^{I_{1}}$, and any $i\in  I_{0}$,  
\begin{equation}
\label{eq:range-I1}
\range{i}{\cpolyn{B}, {\vec y_{I_1}}}\le C^*(B+1)
\end{equation}
 We first prove that for any $\vec x\in \cpoly{B}$, there exists $\vec x'\in \sus{0}$ that is $O(B+1)$ away from $\vec x$ in $L_\infty$. 
\begin{claim}
\label{claim:close-to-cone0} There exists $C$ such that for any $\vec x\in \cpoly{B}$, there exists $\vec x'\in \sus{0}$ such that $|\vec x - \vec x'|_\infty\le C(B+1)$.
\end{claim}
\begin{proof}
The proof is done by analyzing the following two linear programs, denoted by LP$_{\poly}^B$ and LP$_{\cone }^B$ whose variables are $\vec x$ and $\vo$.

\renewcommand{\arraystretch}{1.5}
\begin{center}
\begin{tabular}{|c | c|}
\hline 
LP$_{\poly}^B$& LP$_{\cone}^B$ \\
\hline
\begin{tabular}{rl}
$\max$ &  $0$\\
s.t. & $\ba_{\text{S}}\times \invert{\vec x}\le \invert{\vbb_{\text S}}$ \\
& $\ba_{\text{T}}\times \invert{\vec x +\vo\times \vomatrix{}}\le \invert{\vbb_{\text T}}$ \\

& $-\vo\le \vec 0$\\
& $\vec c\cdot\vo \le B$
\end{tabular}
&  
\begin{tabular}{rl}
$\max$ &  $0$\\
s.t. & $\ba_{\text{S}}\times \invert{\vec x}\le \invert{\vec 0}$ \\
& $\ba_{\text{T}}\times \invert{\vec x +\vo\times \vomatrix{}}\le \invert{\vec 0}$ \\
& $-\vo\le \vec 0$\\
& $\vec c\cdot\vo \le 0$
\end{tabular}\\
\hline 
\end{tabular}
\end{center}
\renewcommand{\arraystretch}{1}
Because $\vec x\in \cpoly{B}$, there exists $\vo\ge \vec 0$ such that $\vec x+\vo\times \vomatrix{}\in \targetpoly$ and $\vec c\cdot \vec x\le B$. Therefore,  $(\vec x,\vec w)$ is a feasible solution to LP$_{\poly}^B$. Notice that LP$_{\cone}^B$ is feasible (for example, $\vec 0$ is a feasible solution) and LP$_{\poly}^B$ and LP$_{\cone}^B$ only differ on the right hand side of the inequalities. Therefore, due to~\citep[Theorem 5 (i)]{Cook86:Sensitivity}, there exists a feasible solution $(\vec x',\vec w')$  to LP$_{\cone}^B$ that is no more than $q\Delta \max\{|\vbb_{\text{S}}-\vec 0|_\infty,|\vbb_{\text{T}}-\vec 0|_\infty, B\} = O(B+1)$ away from $(\vec x,\vec w)$ in $L_\infty$, where $\Delta$ is the maximum absolute value of determinants of square sub-matrices of the left hand side of LP$_{\poly}^B$ and LP$_{\cone}^B$, i.e.,  $\multimatrix{\begin{array}{cc}\ba_{\text{S}} &  0 \\ \ba_{\text{T}} & \ba_{\text{T}}\times \invert{\vomatrix{}}\\ 0 & -{\mathbb I}\\ 0 & \vec c \end{array}}$. Recall that $\vec c\ge \vec 0$. This means that the $-\vo\le\vec 0$ constraint and the $\vec c\cdot\vo\le 0$ constraint in LP$_{\cone}^B$ imply that $\vo = \vec 0$, which means that $\vec x'\in \sus{0}$. This completes the proof of Claim~\ref{claim:close-to-cone0}.
\end{proof}

Like in \myhyperlink{ub-step2.1}{Step 2.1}, we define $\bd_0$ to be the matrix obtained from applying Gauss-Jordan elimination on $\multimatrix{\ba_0^=\\ \vec 1}$. That is, for every $\vec x' = (\vec x_{I_0}', \vec x_{I_1}')\in \sus{0}$, let $n'=\vec x'\cdot \vec 1$,  we have 
\begin{equation}
\label{eq:x-rep-I0}\vec x_{I_0}' =  (\vec x_{I_1}',n')\times \bd_0
\end{equation}
Next, we use Claim~\ref{claim:close-to-cone0} to prove that for any $\vec y = (\vec y_{I_0}, \vec y_{I_1})\in \cpolyn{B}$, $|\vec y - ((\vec y_{I_1},n)\times \bd_0, \vec y_{I_1})|_\infty = O(B+1)$, which would prove (\ref{eq:range-I1}).

For any $\vec y = (\vec y_{I_0}, \vec y_{I_1})\in \cpolyn{B}$, let $\vec x' = (\vec x_{I_0}', \vec x_{I_1}')\in \sus{0}$ denote the vector in $\sus{0}$ that is no more than $C(B+1)$ away from $\vec y$ guaranteed by Claim~\ref{claim:close-to-cone0}. This means that $\vec x_{I_1}'$ and $n'$ are $O(B+1)$ away from $\vec y_{I_1}$ and $n$, respectively, which implies that $\vec x_{I_0}' = (\vec x_{I_1}',n')\times \bd_0 $ is $O(B+1)$ away from $(\vec y_{I_1},n)\times \bd_0$. Also because $\vec x_{I_0}'$ is $O(B+1)$ away from $\vec y_{I_0}$, we have that $\vec y_{I_0}$ is $O(B+1)$ away from $(\vec y_{I_1},n)\times \bd_0$. Formally, let $d_{\max}^0$ denote the maximum absolute value of entries in $\bd_0$, we have the following bound.
\begin{align*}
& |\vec y - ((\vec y_{I_1},n)\times \bd_0, \vec y_{I_1})|_\infty = |\vec y_{I_0} - (\vec y_{I_1},n)\times \bd_0|_\infty\\
\le &|\vec y_{I_0} - \vec x_{I_0}'|_\infty+|  \vec x_{I_0}' - (\vec x_{I_1}',n')\times \bd_0|_\infty+|  (\vec x_{I_1}',n')\times \bd_0- (\vec y_{I_1},n)\times \bd_0|_\infty\\
\le & C(B+1)+0+ |  (\vec x_{I_1}'- \vec y_{I_1},n'-n)\times \bd_0 |_\infty& \text{by (\ref{eq:x-rep-I0})}\\
\le & C(B+1)(2qd_{\max}^0 +1)
\end{align*}
The last inequality holds because $|\vec x_{I_1}'- \vec y_{I_1}|_\infty \le C(B+1)$ and $|n'-n|\le qC(B+1)$. This completes the proof of \myhyperlink{ub-step2.2}{Step 2.2}.

\vspace{2mm}
\noindent{\bf\boldmath \hypertarget{ub-step3}{Step 3} of poly upper bound: Upper-bound the probability.}   Recall that $\vXp=\hist(P)$, where $P$ consists of $n$ independent random variables $Y_1,\ldots,Y_n$ distributed as $\vec \pi$.  Like~\citep{Xia2021:How-Likely}, we represent each $Y_j$  as two random variables $Z_j$ and $W_j$ and a simple Bayesian network based on the partition $[q] = I_0\cup I_1$. 

\begin{dfn}[\bf Alternative representation of $\bm{Y_1,\ldots,Y_{n}}$~\citep{Xia2020:The-Smoothed}]\label{dfn:altfory} For each $j\le n$, we define a Bayesian network with two random variables $Z_j \in \{0,1\}$ and $W_j\in [q]$, where $Z_j$ is the parent of $W_j$. The conditional probabilities are defined as follows.
\begin{itemize}
\item For each $\ell\in \{0,1\}$, let $\Pr(Z_j = \ell)  \triangleq \Pr(Y_j \in I_\ell)$. 
\item For each $\ell\in \{0,1\}$ and each $t\le q$, let $\Pr(W_j = t|Z_j=\ell)  \triangleq \Pr(Y_j = t|Y_j\in I_\ell)$.
\end{itemize} 
\end{dfn}
In particular, if $t\not\in I_\ell$ then $\Pr(W_j = t|Z_j=\ell)=0$. It is not hard to verify  that for any $j\le n$, $W_j$ has the same distribution as $Y_j$. For any $\vec z\in \{0,1\}^n$, we let $\ind_0(\vec z\,)\subseteq [n]$ denote the indices of components of $\vec z$ that equal to $0$. Given $\vec z$, we define the following random variables. 
\begin{itemize}
\item Let $\vec W_{\ind_0(\vec z\,)}  \triangleq \{W_j:j\in\ind_0(\vec z\,)\}$. That is, $\vec W_{\ind_0(\vec z\,)}$ consists of random variables $\{W_j: z_j = 0\}$.
\item Let $\hist(\vec W_{\ind_0(\vec z\,)})$ denote the vector of the $q-d_0+1 = |I_0|$ random variables that correspond to the histogram of $\vec W_{\ind_0(\vec z\,)}$ restricted to $I_0$. Technically, the domain of  every random variable in $\vec W_{\ind_0(\vec z\,)}$ is $[q]$, but since  they only receive positive probabilities on $I_0$, they are treated as random variables over $I_0$ when $\hist(\vec W_{\ind_0(\vec z\,)})$ is defined. 
\item Similarly, let $\vec W_{\ind_1(\vec z\,)} \triangleq\{W_j:j\in\ind_1(\vec z\,)\}$ and let $\hist(\vec W_{\ind_1(\vec z\,)})$ denote  the vector of $|I_1| = d_0-1$ random variables that correspond to the histogram of $\vec W_{\ind_1(\vec z\,)}$.
\end{itemize}

Let $\cpolynint{B}  \triangleq \cpolyn{B}\cap {\mathbb Z}^q$. For any $\vec y_{1}\in {\mathbb Z}_{\ge 0}^{d_0-1}$, we let $\cpolynint{B}|_{\vec y_{1}}$ denote the $I_0$ components of $\vec y\in \cpolynint{B}$ whose $I_1$ components are $\vec y_1$. 
Formally,  
$$\cpolynint{B}|_{\vec y_{1}} \triangleq\left\{\vec y_0\in {\mathbb Z}_{\ge 0}^{q-d_0+1}: (\vec y_0,\vec y_1)\in \cpolynint{B}\right\}$$ 
  
We recall the following calculations in~\citep{Xia2021:How-Likely}  for any $\vec\pi\in \Pi^n$, which is done by first separating the $|\ind_0(\vec z\,)| \ge 0.9 \epsilon n$ case (which happens with $1-\exp(-\Omega(n))$ probability) from the $|\ind_0(\vec z\,)|<  0.9\epsilon n$ (which happens with exponentially small probability), then applying the law of total probability conditioned on $\vec Z$, and finally using the conditional independence in the Bayesian network (i.e., $\vec W$'s are independent given $\vec Z$) to simplify the formula.
\begin{align}
&\Pr\nolimits_{P\sim \vec\pi}(\hist(P)\in \cpolynint{B})  \le  \sum_{\vec z\in \{0,1\}^n: |\ind_0(\vec z\,)| \ge 0.9 \epsilon n} \Pr(\vec Z = \vec z\,) \sum_{\vec y_1\in {\mathbb Z}_{\ge 0}^{d_0-1}}\Pr\left(\hist(\vec W_{\ind_1(\vec z\,)}) =\vec y_1 \;\middle\vert\; \vec Z = \vec z\right)\notag\\
&\hspace{40mm}\times \Pr\left(\hist(\vec W_{\ind_0(\vec z\,)}) \in \cpolynint{B}|_{\vec y_1} \;\middle\vert\; \vec Z = \vec z\right) + \Pr(|\ind_0(\vec z\,)| < 0.9\epsilon n)\label{eq:histz-T}
\end{align}

To upper-bound (\ref{eq:histz-T}), we will show that  for any $\vec z$ with $|\ind_0(\vec z\,)| \ge 0.9 \epsilon n$ and any $\vec y_1 \in {\mathbb Z}_{\ge 0}^{d_0-1}$,
\begin{equation}
\label{eq:histo-T}
\Pr\left(\hist(\vec W_{\ind_0(\vec z\,)}) \in \cpolynint{B}|_{\vec y_1} \;\middle\vert\; \vec Z = \vec z\right)=O((B+1)^{d_{\Delta}})\times O\left(n^{\frac{d_0-q}{2}}\right) 
\end{equation}
Conditioned on $\vec Z=\vec z$, $\hist(\vec W_{\ind_0(\vec z\,)})$ can be viewed as a PMV of $|\ind_0(\vec z\,)|$ strictly positive independent variables over $I_0=[q-d_0+1]$. Therefore, according to the point-wise anti-concentration bound~\cite[Lemma 3 in the Appendix]{Xia2020:The-Smoothed},  for any $\vec z$ with $|\ind_0(\vec z\,)| \ge 0.9 \epsilon n$ and any $\vec y_0\in \cpolynint{B}|_{\vec y_1}$, 
$$\Pr\left(\hist(\vec W_{\ind_0(\vec z\,)}) = \vec y_0\;\middle\vert\; \vec Z = \vec z\right) = O\left(|\ind_0(\vec z\,)|^{\frac{d_0-q}{2}}\right)= O\left(n^{\frac{d_0-q}{2}}\right)$$
Then, to prove (\ref{eq:histo-T}), it suffices to prove $ |\cpolynint{B}|_{\vec y_1}| = O((B+1)^{d_{\Delta}})$. This is done by enumerating vectors in $\cpolynint{B}|_{\vec y_1}$  as follows: According to \myhyperlink{ub-step2.2}{Step 2.2} of the poly upper bound, each $I_{0-}$ component of vectors in $\cpolynint{B}|_{\vec y_1}$ has no more than $\lceil C^*(B+1)+1\rceil $ choices, and given the $I_{0-}$ components, each $I_{0+}$ component has no more than $\lceil C^*+1\rceil $ choices, where $C^*$ is the maximum value of the constants in \myhyperlink{ub-step2.1}{Step 2.1} and \myhyperlink{ub-step2.2}{3.2}. Therefore, 
$$ |\cpolynint{B}|_{\vec y_1}| \le (C^*+1)^{q-d_{\infty}+1} (C^*(B+1)+1)^{d_{\infty}-d_0}  = O((B+1)^{d_{\infty}-d_0})= O((B+1)^{d_{\Delta}})$$
This proves  (\ref{eq:histo-T}).  Then, combining (\ref{eq:histz-T}) and  (\ref{eq:histo-T}), and following a  similar argument as in~\citep{Xia2021:How-Likely}, we have 
\begin{align*}
&\Pr\nolimits_{P\sim \vec\pi}(\hist(P)\in \csus{n,B})\le \Pr\nolimits_{P\sim \vec\pi}(\hist(P)\in \cpolynint{B})\\
\le & \sum_{\vec z\in \{0,1\}^n: |\ind_0(\vec z\,)| \ge 0.9 \epsilon n} \Pr(\vec Z = \vec z\,) \sum_{\vec y_1\in {\mathbb Z}_{\ge 0}^{d_0-1}}\Pr\left(\hist(\vec W_{\ind_1(\vec z\,)}) =\vec y_1 \;\middle\vert\; \vec Z = \vec z\right) \\
&\hspace{40mm} \times O((B+1)^{d_{\Delta}})\times O\left(n^{\frac{d_0-q}{2}}\right) + \exp(-\Theta(n))\\
= & \Pr(|\ind_0(\vec Z\,)| \ge 0.9 \epsilon n)\times O((B+1)^{d_{\Delta}})\times O\left(n^{\frac{d_0-q}{2}}\right) + \exp(-\Theta(n))\\
=&O\left((B+1)^{d_{\Delta}}\cdot \left(\frac{1}{\sqrt n}\right)^{q-d_0}\right)
\end{align*}
This proves the polynomial upper bound for $B\le \sqrt n$.

\vspace{3mm}
\noindent \hypertarget{PMV-proof-sup-sqrtn-ub-largeB}{{\bf \boldmath Proof for the polynomial upper bound of $\sup$ for $B > \sqrt n$.}} Notice that for any $B$ and $n$, $\cpoly{B,n}\subseteq \cpoly{\infty}$. Therefore, it suffices to prove the following stronger claim, which  holds for any $B$, any $n$, and any $\vec \pi\in \Pi^n$.
\begin{claim}
\label{claim:poly-upper-strong}
For any  $n$ and $\vec \pi\in \Pi^n$, $\Pr\left(\vXp \in \cpoly{\infty}\right) = O\left((\frac{1}{\sqrt n})^{q-d_\infty}\right)$.
\end{claim}
\begin{proof}
The proof is similar to \myhyperlink{ub-step3}{Step 3 of the $B\le \sqrt n$ case} above. We will define a Bayesian network for the partition  $[q]=I_{0+}\cup I_{1+}$, and then apply (\ref{eq:range-I1+}) to enumerate vectors in $\cpolynint{\infty}|_{\vec y_{1+}}$.

\begin{dfn}[\bf Alternative representation of $\bm{Y_1,\ldots,Y_{n}}$]
\label{dfn:alt-for-y+} For each $j\le n$, we define a Bayesian network with two random variables $Z_j^+ \in \{0,1\}$ and $W_j^+\in [q]$, where $Z_j^+$ is the parent of $W_j^+$. The conditional probabilities are defined as follows.
\begin{itemize}
\item For each $\ell\in \{0,1\}$, let $\Pr(Z_j^+ = \ell) \triangleq \Pr(Y_j \in I_{\ell+})$. 
\item For each $\ell\in \{0,1\}$ and each $t\le q$, let $\Pr(W_j^+ = t|Z_j^+=\ell)  \triangleq \Pr(Y_j = t|Y_j\in I_{\ell+})$.
\end{itemize}
\end{dfn}
In particular, if $t\not\in I_{\ell+}$ then $\Pr(W_j^+ = t|Z_j^+=\ell)=0$.
It is not hard to verify  that for any $j\le n$, $W_j^+$ follows the same distribution as $Y_j$. For any $\vec z\in \{0,1\}^n$, we let $\ind_{0+}(\vec z\,)\subseteq [n]$  denote the indices of components of $\vec z$ that equal to $0$. 
Given $\vec z$, we define the following random variables. 
\begin{itemize}
\item Let $\vec W^+_{\ind_{0+}(\vec z\,)} \triangleq\{W_j^+:j\in\ind_{0+}(\vec z\,)\}$. That is, $\vec W^+_{\ind_{0+}(\vec z\,)}$ consists of random variables $\{W_j^+: z_j = 0\}$.
\item Let $\hist(\vec W^+_{\ind_{0+}(\vec z\,)})$ denote the vector of the $|I_{0+}|= q-d_\infty+1$ random variables that correspond to the histogram of $\vec W^+_{\ind_{0+}(\vec z\,)}$ restricted to $I_{0+}$.  
\item Similarly, let $\vec W^+_{\ind_{1+}(\vec z\,)} \triangleq\{W_j^+:j\in\ind_{1+}(\vec z\,)\}$ and let $\hist(\vec W^+_{\ind_{1+}(\vec z\,)})$ denote  the vector of $|I_{1+}| = d_\infty-1$ random variables that correspond to the histogram of $\vec W^+_{\ind_{1+}(\vec z\,)}$.
\end{itemize}

For any $\vec y_{1+}\in {\mathbb Z}_{\ge 0}^{I_{1+}}$, we let $\cpolynint{\infty}|_{\vec y_{1+}}$ denote the $I_{0+}$ components of $\vec y\in \cpolynint{\infty}$ whose $I_{1+}$ components are $\vec y_{1+}$. 
Formally,  
$$\cpolynint{\infty}|_{\vec y_{1+}} \triangleq\left\{\vec y_{0+}\in {\mathbb Z}_{\ge 0}^{q-d_0+1}: (\vec y_{0+},\vec y_{1+})\in \cpolynint{\infty}\right\}$$ 
Like \myhyperlink{ub-step3}{Step 3 of the $B\le \sqrt n$ case},  for any $\vec\pi\in \Pi^n$,
\begin{align}
&\Pr\nolimits_{P\sim \vec\pi}(\hist(P)\in \cpolynint{\infty})\notag \\
 \le &  \sum_{\vec z\in \{0,1\}^n: |\ind_{0+}(\vec z\,)| \ge 0.9 \epsilon n} \Pr(\vec Z^+ = \vec z\,) \sum_{\vec y_{1+}\in {\mathbb Z}_{\ge 0}^{d_\infty-1}}\Pr\left(\hist(\vec W^+_{\ind_{1+}(\vec z\,)}) =\vec y_{1+} \;\middle\vert\; \vec Z^+ = \vec z\right)\notag\\
&\hspace{20mm}\times \Pr\left(\hist(\vec W^+_{\ind_{0+}(\vec z\,)}) \in \cpolynint{\infty}|_{\vec y_{1+}} \;\middle\vert\; \vec Z^+ = \vec z\right) + \Pr(|\ind_{0+}(\vec z\,)| < 0.9\epsilon n)\label{eq:histz-T+}
\end{align}
Following  (\ref{eq:range-I1+}) and the point-wise anti-concentration bound~\cite[Lemma 3 in the Appendix]{Xia2020:The-Smoothed},   for any $\vec z$ with $|\ind_{0+}(\vec z\,)| \ge 0.9 \epsilon n$ and any $\vec y_{1+} \in {\mathbb Z}_{\ge 0}^{d_\infty-1}$, we have
\begin{equation}
\label{eq:histo-T+}
\Pr\left(\hist(\vec W^+_{\ind_{0+}(\vec z\,)}) \in \cpolynint{\infty}|_{\vec y_{1+}} \;\middle\vert\; \vec Z^+ = \vec z\right)=O(1)\times O\left(n^{\frac{d_\infty-q}{2}}\right) 
\end{equation}
 Then, combining (\ref{eq:histz-T+}) and  (\ref{eq:histo-T+}) and recalling that $\hist(P)$ is a size-$n$ non-negative integer vector,  we have 
\begin{align*}
&\Pr\nolimits_{P\sim \vec\pi}(\hist(P)\in \csus{n,B})\le \Pr\nolimits_{P\sim \vec\pi}(\hist(P)\in \cpolynint{\infty}) \\
\le & \sum_{\vec z\in \{0,1\}^n: |\ind_{0+}(\vec z\,)| \ge 0.9 \epsilon n} \Pr(\vec Z^+ = \vec z\,) \sum_{\vec y_{1+}\in {\mathbb Z}_{\ge 0}^{d_\infty-1}}\Pr\left(\hist(\vec W^+_{\ind_{1+}(\vec z\,)}) =\vec y_{1+} \;\middle\vert\; \vec Z^+ = \vec z\right) \\
&\hspace{40mm} \times  O\left(n^{\frac{d_\infty-q}{2}}\right) + \exp(-\Theta(n))\\
= & \Pr(|\ind_{0+}(\vec Z^+\,)| \ge 0.9 \epsilon n) \times O\left(n^{\frac{d_\infty-q}{2}}\right) + \exp(-\Theta(n))
= O\left( \left(\frac{1}{\sqrt n}\right)^{q-d_\infty}\right)
\end{align*}
This proves Claim~\ref{claim:poly-upper-strong}.
\end{proof}
The polynomial upper bound for $B>\sqrt n$ follows after Claim~\ref{claim:poly-upper-strong}. 

\vspace{3mm}
\noindent \hypertarget{PMV-proof-sup-sqrtn-lb}{{\bf \boldmath Proof for the $\Omega\left(\dfrac{\min\{B+1,\sqrt n\}^{d_{\Delta}}}{(\sqrt n)^{q-d_0}} \right) $ lower bound  of PT-$\Theta(\sqrt n)$-$\sup$.}}    It suffices to prove the lower bound for $B\le \sqrt n$, because when $B\ge \sqrt n$, we have 
$$\sup_{\vec\pi\in\Pi^n}\Pr(\vXp\in\csus{n,B})\ge\sup_{\vec\pi\in\Pi^n}\Pr(\vXp\in\csus{n,\sqrt n}) = \Omega\left( (\frac{1}{\sqrt n})^{q-d_\infty} \right),$$
which is the desired lower bound. 

The proof  for $B\le \sqrt n$ proceeds in three steps. In \myhyperlink{lb-step1}{Step 1}, for any strictly positive $\pi^*\in\sus{0}$, we identify  an $O(\sqrt n)$ neighborhood of $n\cdot\pi^*$ that contains  $\Omega \left((B+1)^{d_{\Delta}}\cdot \left( \sqrt n\right)^{d_0-1}\right)$ integer vectors in $\csus{n,B}$. In   \myhyperlink{lb-step2}{Step 2}, for any $\pi^*\in\conv(\Pi)$, we define $\vec \pi^\circ=(\pi^\circ_1,\ldots,\pi^\circ_n)\in \Pi^n$ such that $\sum_{j=1}^n\pi_j^\circ$ is $O(\sqrt n)$ away from $n\cdot\pi^*$.  The lower bound is then proved in \myhyperlink{lb-step3}{Step 3}. Among the three steps, \myhyperlink{lb-step1}{Step 1} is the most involved part, and \myhyperlink{lb-step2}{Steps 2} and \myhyperlink{lb-step3}{3} follow after similar proofs in~\citep{Xia2021:How-Likely}.

\vspace{2mm}
\noindent{\bf\boldmath \hypertarget{lb-step1}{Step 1} of poly lower bound:  Identify $\Omega \left((B+1)^{d_{\Delta}}\cdot \left( \sqrt n\right)^{d_0-1}\right)$  vectors in $\csus{n,B}$.}   

\vspace{2mm}\noindent{\bf\bf Overview  of Step 1.} The proof  proceeds in  four sub-steps.  In \myhyperlink{lb-step1.1}{Step 1.1}, we prove in Claim~\ref{claim:int-close} that there exists a constant $C'$ such that for every (possibly non-integer) non-negative vector in $\cpolyn{B}$, there is an ``accompany'' integer vector in $\csus{n,B}$ that is at most $C'$ away in $L_\infty$. Then, we  identify a set $\calR_{B,n}$ of possibly non-integer vectors in $\cpolyn{B}$ that are at least $2C'$ away from each other, and apply Claim~\ref{claim:int-close} to obtain a set $\calR^{\mathbb Z}_{B,n}$ of integer vectors in $\csus{n,B}$ of the same size (as $\calR_{B,n}$). To define $\calR_{B,n}$, we explore two directions in an $O(\sqrt n)$ neighborhood of $n\cdot\pi^*$: the $\calP_0$ direction and the $\calP_\infty$ direction (which we recall from (\ref{eq:calP}) are affine hulls of $\sus{0}$ and $\sus{\infty}$, respectively). More precisely, we will  first enumerate $\Omega((\sqrt n)^{d_0-1})$ vectors in a  $\calP_0$ neighborhood of $n\cdot\pi^*$ (defined as $\calR_n^0$ in \myhyperlink{lb-step1.2}{Step 1.2}), and then enumerate  $\Omega((B+1)^{d_\Delta})$ vectors in a neighborhood of $n\cdot\pi^*$ (defined as $\calR_B^\infty$ in \myhyperlink{lb-step1.3}{Step 1.3}) that is a complement of $\calP_0$ in $\calP_\infty$. Finally, in \myhyperlink{lb-step1.4}{Step 1.4}, we formally define $\calR_{B,n}$ and $\calR^{\mathbb Z}_{B,n}$ and prove that  $|\calR_{B,n}|  = |\calR^{\mathbb Z}_{B,n}| = \Omega((B+1)^{d_\Delta}\cdot (\sqrt n)^{d_0-1})$ and $\calR_{B,n}\subseteq \cpolyn{B}\cap {\mathbb R_{\ge 0}^q}$  in Claim~\ref{claim:far-away} and Claim~\ref{claim:calR-in-HB}, respectively.

\vspace{2mm}
\noindent{\bf\boldmath \hypertarget{lb-step1.1}{Step 1.1}: Define $C'$.} We prove the following claim.

\begin{claim}
\label{claim:int-close}
There exists a constant  $C'$  that does not depend on $n$ or $B$, such that for any  $\vec x\in \cpolyn{B}\cap{\mathbb R}_{\ge 0}^q$, there exists a  (non-negative integer) vector $\vec x'\in \csus{n,B}$  such that $|\vec x - \vec x'|_\infty< C'$. 
 \end{claim}
\begin{proof} 

Consider the following linear program LP$_{B,n}$  whose variables are $\vec x$ and $\vo$:
\begin{equation*} 
\label{eq:LPBn}
LP_{B,n} = \left\{\text{\begin{tabular}{rl}
$\max$ &\ \ \ $ 0$\\
s.t. & $\ba_{\text{S}} \times \invert{\vec x} \le \invert{\vbb_\text{S}}$\\
 & $\ba_{\text{T}} \times \invert{\vec x + \vo\times \vomatrix{}} \le \invert{\vbb_\text{T}}$\\
 &  $\vec x\cdot\vec 1 = n$\\
  &  $\vec c\cdot \vo \le B$\\
   &  $\vec x  \ge \vec 0$, $\vo  \ge \vec 0$
\end{tabular}}
\right.
\end{equation*}
It is not hard to verify that, because $\vec x\in \cpolyn{B}\cap {\mathbb R_{\ge 0}^q}$, LP$_{B,n}$ has a feasible solution $(\vec x,\vo)$ for some $\vo\ge\vec 0$ (which may not be an integer vector).  Recall that   $\csus{n,B}\ne\emptyset$, which means that LP$_{B,n}$ has a feasible integer solution. Therefore, by~\citep[Theorem 1]{Cook86:Sensitivity},  LP$_{B,n}$ has an integer solution $(\vec x',\vo')$ whose $L_\infty$ distance to $(\vec x,\vo)$ is no more than $(q+|\voset{}|)\Delta$, where $\Delta$ is the maximum absolute determinant of square submatrices of the matrix that defines LP$_{B,n}$, i.e., $\left[\begin{array}{rc}
\ba_{\text{S}} & \vec 0\\
\ba_{\text{T}} & \ba_{\text{T}}\times \invert{\vomatrix{}}\\
\vec 1&\vec 0\\
- \vec 1&\vec 0\\
\vec 0&\vec c \\
- \mathbb I&\vec 0\\
\vec 0& -\mathbb I \ \ \ 
\end{array}\right]$. Notice that this matrix does not depend on $B$ or $n$. 
The claim follows after letting $C' \triangleq (q+|\voset{}|)\Delta$ and noticing that $(\vec x',\vo\,')\in\csus{n,B}$.
\end{proof}

\vspace{2mm}
\noindent{\bf\boldmath \hypertarget{lb-step1.2}{Step 1.2}: Define $\calR_n^0$.} Let $\calP_0^*$ denote the size-$0$ vectors of $\calP_0$. That is,
$$\calP_0^* = \calP_0\cap\left\{\vec x:\vec x\cdot \vec 1 = 0\right\} = \left\{\vec x\in{\mathbb R}^q: \left[\begin{array}{c}\ba_0^=\\  \vec 1\end{array}\right]\times \invert{\vec x} = \invert{\vec 0}\right\}$$
Recall from Assumption~\ref{asmp:indwith1} that $\rank\left(\left[\begin{array}{c}\ba_0^=\\  \vec 1\end{array}\right]\right) = \rank(\ba_0^=) +1 = q-d_0 +1$, which means that $\dim(\calP_0^*) = q -\rank\left(\multimatrix{\ba_0^=\\  \vec 1} \right) = d_0-1$. Therefore, $\calP_0^*$ contains a basis of  $d_0-1$ linearly independent vectors, denoted by  $\bB^0=\{\vec p_1^{\,0},\ldots,{\vec p_{d_0-1}}^{\,0}\}$.  Let ${\mathbb L^0}$ denote the lattice generated by $\bB^0$  excluding $\vec 0$. That is, 
$$ {\mathbb L^0} = \left\{ \sum\nolimits_{i=1}^{d_0-1} \gamma_i\cdot\vec p_i^{\,0}: \forall {i\le d_0-1}, \gamma_i\in {\mathbb Z}\text{ and } \exists  i\le d_0-1 \text{ s.t. } \gamma_i \ne 0\right\}$$
Let $C^0$ denote the minimum $L_\infty$ norm of all vectors in ${\mathbb L^0}$, i.e.,
$$C^0 =  \inf \left\{\left|\vec x \right|_\infty: \vec x\in{\mathbb L^0} \right\}$$
Next, we prove  $C^0>0$. Suppose for the sake of contradiction that $C^0=0$. Then,  there exist  a sequence of vectors 
$$\left\{\vec x^j = \sum\nolimits_{i=1}^{d_0-1} \gamma_i^j\cdot\vec p_i ^{\,0}: j\in\mathbb N\right\}$$
such that for all $j\in\mathbb N$, $|\vec x^j|_\infty\le \frac 1j$. Let $\vec\gamma^j = (\gamma_1^j,\ldots,\gamma_{d_0-1}^j)$. Because $|\vec\gamma^j|_\infty \ge 1$, we have  $ \frac{\vec x^j}{|\vec\gamma^j|_\infty}\le \frac 1j$. Notice that $\{\vec \gamma\in {\mathbb R}^{d_0-1}:|\vec \gamma|_\infty=1\}$ is closed and compact,  there exists a subsequence of $\left\{\frac{\vec\gamma^j}{|\vec\gamma^j|_\infty}\right\}$ that converges to a vector $\vec\gamma^*$ with $|\vec\gamma^*|_\infty =1$. Let  $\vec x^* = \sum\nolimits_{i=1}^{d_0-1} \gamma_i^*\cdot\vec p_i^{\,0}$. It follows that $\left\{\frac{\vec x^j}{|\vec\gamma^j|_\infty} : j\in\mathbb N \right\}$ converges to $\vec x^*$, which means that $|\vec x^*|_\infty = \lim_{j\ra\infty }\left|\frac{\vec x^j}{|\vec\gamma^j|_\infty}\right|_\infty = 0$.  This contradicts the  linear independence of vectors in $\bB^0$.  

Then, we use $C^0$ and $\bB^0$ to define a subset of $\frac{2C'}{C^0}\cdot {\mathbb L^0}\cup \{\vec 0\}$ as follows.
\begin{equation}
\label{eq:calR-n-0}
\calR_n^0 = \left\{\sum\nolimits_{i=1}^{d_0-1} \frac{2C'}{C^0}\gamma_i\cdot\vec p_i^{\,0}: \forall i\le d_0-1, \gamma_i\in \{ 0, 1,\ldots,\lfloor  \sqrt n \rfloor \}  \right\}
\end{equation}
By definition, $\calR_n^0\subseteq \calP_0^*\subseteq  \calP_0$, and for every vector $\vec x\in \calR_n^0$, we have $\vec x\cdot \vec 1 = 0$ and $|\vec x|_\infty = O(\sqrt n)$. The value $\frac{2C'}{C^0}$ is chosen so that  the $L_\infty$ distance between any pair of different vectors in $\calR_n^0$ is at least $2C'$.

\vspace{2mm}
\noindent{\bf\boldmath \hypertarget{lb-step1.3}{Step 1.3}: Define $\calR_B^\infty$.}  Recall that $\calP_0\subseteq \calP_\infty$, $d_0 = \dim(\calP_0)$ and $d_{\infty} = \dim(\sus{\infty}) =\dim(\calP_\infty)$. In the following procedure, we define a set of $d_\Delta$ linear independent vectors $\bB^\infty\subseteq \sus{\infty}$, which are the basis of a complement of $\calP_0$ in $\calP_\infty$.

\vspace{2mm}\noindent{\bf\bf \bf\boldmath Procedure for defining $\bB^\infty$.}Let $\{\vec p_1,\ldots,\vec p_{d_\infty}\}$ denote an arbitrary set of $d_\infty$ linearly independent vectors in 
$\sus{\infty}$, whose existence is guaranteed by the fact that $\dim(\sus{\infty}) = d_{\infty}$.
Start with $\bB^\infty =\emptyset$. For every $1\le j\le d_\infty$, we add $\vec p_j$ to $\bB^\infty$ if and only if it is linearly independent with $\calP_0\cup \bB^\infty$. At the end of procedure we have $|\bB^\infty| = d_\infty - d_0= d_\Delta$. W.l.o.g., let $\bB^\infty=\{\vec p_1^{\,\infty},\ldots,{\vec p_{d_\Delta}}^{\,\infty}\}$.

Let $\overline{\calP_0} = \Span(\bB^\infty)$. It follows that  $\dim(\overline{\calP_0}) = d_{\Delta}$ and
\begin{equation}
\label{eq:span-diff}
\calP_0\cap \overline{\calP_0}=\{\vec 0\}\text{ and }  \calP_0 +  \overline{\calP_0} = \calP_\infty
\end{equation}
 For each $i\le d_{\Delta}$, because $\vec p_i^{\,\infty}\in \sus{\infty}$, we have $\vec p_i^{\,\infty}\in\sourcepolyz$ and we can write $\vec p_i^{\,\infty} = \vec y_i^{\,\infty} - \vo_i^{\,\infty}\times \vomatrix{}$, where $\vec y_i^{\,\infty}\in \targetpolyz$ and $\vo_i^{\,\infty}\ge \vec 0$. For every $i\le d_\Delta$, we must have $\vec w_i^{\,\infty}\ne \vec 0$, because otherwise $\vec p_i^{\,\infty}\in \sus{0}\subseteq \calP_0$, which contradicts (\ref{eq:span-diff}), because $\vec p_i^{\,\infty}\in \overline{\calP_0}$.  Recall that  $\vec c >\vec 0$. W.l.o.g.~we can assume $\vec c\cdot\vec w_i^{\,\infty} =1$, otherwise we divide $\vec p_i^{\,\infty}$ by $\vec c\cdot \vec w_i^{\,\infty} $. 

Let $\mathbb L^\infty$ denote the lattice generated by $\bB^\infty$ excluding $\vec 0$. That is, $${\mathbb L^\infty} = \left\{ \sum\nolimits_{i=1}^{d_{\Delta}} \eta_i\cdot\vec p_i^{\,\infty}: \forall {i\le d_{\Delta}}, \eta_i\in {\mathbb Z}\text{ and } \exists  i\le d_{\Delta} \text{ s.t. }  \eta_i \ne 0\right\}$$
Let $C^\infty$ denote the minimum distance between ${\mathbb L^\infty}$ and $\calP_0$. That is, 
\begin{equation}
\label{eq:dis-P0-Pinfty}C^\infty = \inf \left\{\left|\vec x^{\,\infty} - \vec y\,\right|_\infty: \vec x^{\,\infty}\in{\mathbb L^\infty}, \vec y\in \calP_0\right\}
\end{equation}
Following an argument that is similar to the proof of  $C^0>0$, we have $C^\infty>0$. Formally, for the sake of contradiction suppose $C^\infty=0$. Then, there exist two sequences of vectors 
$$\left\{\vec x^j = \sum\nolimits_{i=1}^{d_{\Delta}} \eta_i^j\cdot\vec p_i^{\,\infty} : j\in\mathbb N\right\}\text{ and }\left\{\vec y^j\in \calP_0:j\in\mathbb N\right\},$$
such that for all $j\in\mathbb N$,  $|\vec x^j -\vec y^j|_\infty\le \frac 1j$. Let $\vec\eta^j = (\eta_1^j,\ldots,\eta_{d_{\Delta}}^j)$. Because $|\vec\eta^j|_1 \ge 1$ and $\calP_0$ includes $\vec 0$, the distance between $\frac{\vec x^j}{|\vec\eta^j|_1}$ and $\calP_0$ is at most $\frac 1j$. Notice that $\{\vec \eta\in {\mathbb R}^{d_\Delta}:|\vec \eta|_1=1\}$ is closed and compact,  there exists a subsequence of $\left\{\frac{\vec\eta^j}{|\vec\eta^j|_1}\right\}$ that converges to a vector $\vec\eta^*$ with $|\vec\eta^*|_1 =1$. Let $\left\{j_i: i\in\mathbb N\right\}$ denote the indices of the subsequence. Let  $\vec x^* = \sum\nolimits_{i=1}^{d_{\Delta}} \eta_i^*\cdot\vec p_i^{\,\infty}$. Because $\vec p_i^{\,\infty}$'s are linearly independent, we have $\vec x^*\ne \vec 0$. Then,  $\{\vec x^{j_i}:i\in \mathbb N\}$ and  $\{\vec y^{j_i}:i\in \mathbb N\}$ both converge to $\vec x^*$. Because $\calP_0$ is closed, we have $\vec x^*\in \calP_0$, which means that $\vec x^*\in  \overline{\calP_0}\cap \calP_0= \{\vec 0\}$. This contradicts (\ref{eq:span-diff}).

Because $\csus{n,B}\ne \emptyset$, the infimum of all $B^*\ge 0$  such that there exists $n^*\in\mathbb N$ so that $\csus{B^*,n^*}\ne \emptyset$ is well-defined, formally defined as $B^{\#}$ as follows.
\begin{equation}
\label{eq:B-sharp}
B^{\#} =  \inf \{B^*\ge 0: \exists n^*\in{\mathbb N}\text{ s.t. } \csus{B^*,n^*}\ne \emptyset\}
\end{equation}
We note that there exists $n^{\#}\in\mathbb N$ such that $\csus{B^{\#},n^{\#}}\ne \emptyset$, because  there are finite number of combinations of operations whose total budget is under $B$, and $B^{\#}$ is the minimum budget among the successful combinations (at a non-negative integer vector whose size is $n^{\#}$).  Clearly $B\ge B^{\#}$. Define
\begin{equation}
\label{eq:calR-B-infty}
\calR_B^\infty = \left\{\sum\nolimits_{i=1}^{d_{\Delta}} \frac{2C'}{C^\infty}\eta_i\cdot\vec p_i^{\,\infty}: \forall i\le d_{\Delta},\eta_i\in \{0,1,\ldots, \left\lfloor\frac{(B-B^{\#})C^\infty}{2d_{\Delta}C'}\right\rfloor\}\right\}
\end{equation}
where we recall that $C'$ is the constant guaranteed by Claim~\ref{claim:int-close}. Intuitively, $\calR_B^\infty$ consists of some ``grids'' in $\frac{2C'}{C^\infty}\cdot {\mathbb L^\infty}\cup \{\vec 0\}$, which is a subset of $\sus{\infty}$, because for all $i\le d_\Delta$, $\vec p_i^{\,\infty}\in \sus{\infty}$. 

\vspace{2mm}
\noindent{\bf\boldmath \hypertarget{lb-step1.4}{Step 1.4}: Define $\calR_{B,n}$ and $\calR^{\mathbb Z}_{B,n}$.}   Recall the definition of $B^{\#}$ from (\ref{eq:B-sharp}) that $\csus{B^{\#},n^{\#}}\ne\emptyset$.  Fix $\vec x^{\,\#}\in \csus{B^{\#},n^{\#}}$. Let $\vec x^{\,@}\in \sus{0}$ denote an arbitrary interior point of $\sus{0}$. That is, let $\ba_0^+$ denote the remaining rows of $\ba_0=\left[\begin{array}{c}\ba_{\text{S}}\\\ba_{\text{T}}\end{array}\right]$  after removing the implicit equalities $\ba_0^=$, there exists a constant $\epsilon^@>0$ such that 
$$\ba_0^+\times \invert{\vec x^{\,@}}< -\invert{\epsilon^@\cdot\vec 1}\text{ and }\ba_0^=\times \invert{\vec x^{\,@}}= \invert{\vec 0}$$
Let $C^@>0$ denote an arbitrary constant such that
\begin{equation}
\label{eq:C-at}
C^@ > \frac{2C'}{C^0}\cdot \frac{\sum\nolimits_{i=1}^{d_0-1}  |\ba_0^+\times \invert{\vec p_i^{\,0}}|_\infty}{\epsilon^@}
\end{equation}
The constraint on $C^@$ in (\ref{eq:C-at}) guarantees that for any $\vec x^{\,0}\in \calR_n^0$, we have $C^@\sqrt n \vec x^{\,@} + \vec x^{\,0}\in \sus{0}$, which will be formally proved and used in the proof of Claim~\ref{claim:calR-in-HB} below.

For any $n\in\mathbb N$ and any $\vec x^{\,\infty}\in \calR_B^\infty$,  define
$$\calR_{\vec x^{\,\infty}} = (n- n^{\#} - C^@\sqrt n\cdot(\vec x^@\cdot\vec 1) -  \vec x^{\,\infty}\cdot\vec 1)\cdot  \pi^* + \vec x^{\,\#}+ C^@\sqrt n\vec x^{\,@}+ \calR_n^0 + \vec x^{\,\infty}  $$
It follows that the size of every vector in $\calR_{\vec x^{\,\infty}} $ is $n$. Let $\calR_{B,n} = \bigcup_{\vec x^{\,\infty}\in\calR_B^\infty} \calR_{\vec x^{\,\infty}}$.  The following claim states that vectors in $\calR_{B,n}$ are   $2C'$ away from each other in $L_\infty$, where we recall that $C'$ is the constant guaranteed by Claim~\ref{claim:int-close}.
\begin{claim}[\bf\boldmath Sparsity of  $\calR_{B,n}$]
\label{claim:far-away}
For any pair of  vectors $\vec x^1,\vec x^2\in \calR_{B,n}$ whose $\calR_n^0$ or $\calR_B^\infty$  components are different, we have $|\vec x^1-\vec x^2|_\infty \ge  2C'$.
\end{claim}
\begin{proof}
For $j\in \{1,2\}$, we write
$$\vec x^j = \ell^j \pi^* + \vec x^{\,\#}+ C^@\sqrt n\vec x^{\,@}+ \sum\nolimits_{i=1}^{d_0-1} \frac{2C'}{C^0}\gamma_i^j\cdot\vec p_i^{\,0} +   \sum\nolimits_{i=1}^{d_{\Delta}} \frac{2C'}{C^\infty}\eta_i^j\cdot\vec p_i^{\,\infty},$$
where $\ell^1$ and $\ell^2$ guarantee that $\vec x^1\cdot\vec 1 = \vec x^2\cdot\vec 1 = n$.
Then, 
\begin{align*}
\vec x^1 - \vec x^2 = &\underbrace{ (\ell^1-\ell^2)\pi^*}_{\in \calP_0} + \underbrace{\sum\nolimits_{i=1}^{d_0-1} \frac{2C'}{C^0}(\gamma_i^1-\gamma_i^2)\cdot\vec p_i^{\,0}}_{\in \frac{2C'}{C^0}\cdot \mathbb L^0\cup\{\vec 0\}}  + \underbrace{\sum\nolimits_{i=1}^{d_{\Delta}} \frac{2C'}{C^\infty}(\eta_i^1-\eta_i^2)\cdot\vec p_i^{\,\infty}}_{\in\frac{2C'}{C^\infty}\cdot\mathbb L^\infty\cup\{\vec 0\}}
\end{align*}
For $j\in \{1,2\}$, let $\vec \gamma^j = (\gamma_1^j,\ldots,\gamma_{d_0-1}^j)$ and $\vec \eta^j = (\eta_1^j,\ldots,\eta_{d_{\Delta}}^j)$. Let $\vec \gamma^{\Delta} = \vec \gamma^1 - \vec \gamma^2$ and   $\vec \eta^{\Delta} = \vec \eta^1 - \vec \eta^2$. Claim~\ref{claim:far-away} is proved in the following two cases.
\begin{itemize}
\item $\vec \eta^{\Delta} \ne \vec 0$. Notice that $\sum\nolimits_{i=1}^{d_{\Delta}}  (\eta_i^1-\eta_i^2)\cdot\vec p_i^{\,\infty}  \in {\mathbb L^\infty}$. Recall that $\pi^*\in\sus{0}\subseteq \calP_0$, which means that $(\ell^2-\ell^1)\pi^*- \sum\nolimits_{i=1}^{d_0-1} \frac{2C'}{C^0}(\gamma_i^1-\gamma_i^2)\cdot\vec p_i^{\,0}\in \calP_0$. According to (\ref{eq:dis-P0-Pinfty}), the $L_\infty$ distance between $\sum\nolimits_{i=1}^{d_{\Delta}}  (\eta_i^1-\eta_i^2)\cdot\vec p_i  \in {\mathbb L^\infty}$  and any vector in $\calP_0$ is at least $C^\infty$. Therefore, 
\begin{align*}
|\vec x^1 - \vec x^2|_\infty = &\left|(\ell^1-\ell^2)\pi^*  + \sum\nolimits_{i=1}^{d_0-1} \frac{2C'}{C^0}(\gamma_i^1-\gamma_i^2)\cdot\vec p_i^{\,0} + \sum\nolimits_{i=1}^{d_{\Delta}} \frac{2C'}{C^\infty}(\eta_i^1-\eta_i^2)\cdot\vec p_i^{\,\infty} \right|\\
=& \frac{2C'}{C^\infty}\cdot\left| \sum\nolimits_{i=1}^{d_{\Delta}} (\eta_i^1-\eta_i^2)\cdot\vec p_i^{\,\infty}  -  \frac{C^\infty}{2C'}\left((\ell^2-\ell^1)\pi^* - \sum\nolimits_{i=1}^{d_0-1} \frac{2C'}{C^0}(\gamma_i^1-\gamma_i^2)\cdot\vec p_i^{\,0} \right) \right|\\
\ge & \frac{2C'}{C^\infty} \cdot C^\infty = 2 C'
\end{align*}
\item $\vec \eta^{\Delta} = \vec 0$. In this case we must have $\ell^1=\ell^2$ and  $\vec \gamma^{\Delta}\ne \vec 0$. Notice that $\sum\nolimits_{i=1}^{d_0-1}(\gamma_i^1-\gamma_i^2)\cdot\vec p_i^{\,0} \in {\mathbb L^0}$, and recall that the $L_\infty$ norm of any vector in ${\mathbb L}^0$ is at least $C^0>0$, we have 
$$\left|\vec x^1 - \vec x^2\right|_\infty =  \frac{2C'}{C^0}\cdot \left|\sum\nolimits_{i=1}^{d_0-1} (\gamma_i^1-\gamma_i^2)\cdot\vec p_i^{\,0}\right| \ge\frac{2C'}{C^0}\cdot C^0 = 2C' $$
\end{itemize}
This completes the proof of Claim~\ref{claim:far-away}.
\end{proof}
By Claim~\ref{claim:far-away}, $|\calR_{B,n}| = \left(1+\left\lfloor\frac{(B-B^{\#})C^\infty}{2d_{\Delta}C'}\right\rfloor\right)^{d_{\Delta}}\times \left(\lfloor  \sqrt n \rfloor\right)^{d_0-1}$, which is $\Omega((B+1)^{d_{\Delta}}\cdot (\sqrt n)^{ d_0-1})$. Next, we prove that for all sufficiently large $n$, $\calR_{B,n}\subseteq \cpolyn{B}\cap{\mathbb R_{\ge 0}^q}$. 

\begin{claim}[\bf\boldmath $\calR_{B,n}\subseteq \cpolyn{B}\cap{\mathbb R_{\ge 0}^q}$]
\label{claim:calR-in-HB}
There exists $N\in\mathbb N$ that does not depend on $B$ or $n$, such that when $n\ge N$, $\calR_{B,n}\subseteq \cpolyn{B}\cap{\mathbb R}_{\ge 0}^q$. 
\end{claim}
\begin{proof}  For any  $\vec x\in \calR_{B,n}$,  we can write
$$\vec x = \ell \pi^* + \vec x^{\,\#}+ C^@\sqrt n\vec x^{\,@} + \vec x^{\,0} +  \vec x^{\,\infty}  ,$$
where $\vec x^{\,0}\in \calR_n^0$, $\vec x^{\,\infty}\in \calR_B^\infty$, and $\ell$ guarantees that $\vec x\cdot\vec 1 = n$. Recall that $\pi^*\in  \sus{0}$ is strictly positive. Therefore, $\vec x\ge \vec 0$ for any sufficiently large $n$, which means that $\calR_{B,n}\subseteq {\mathbb R}_{\ge 0}^q$. The rest of the claim proves $\vec x\in \cpolyn{B}$ in the following three steps. 

\vspace{2mm}\noindent{\bf\bf \bf\boldmath (1) $C^@\sqrt n\vec x^{\,@} + \vec x^{\,0}\in \sus{0}$.} Let $\vec x''=C^@\sqrt n\vec x^{\,@} + \vec x^{\,0}$, we prove the following. 
\begin{itemize}
\item[] {\bf\boldmath (1.1) $\ba_0^=\times\invert{\vec x''} = \invert{\vec 0}$.} Recall that $\vec x^{\,@}$ is an interior point of $\sus{0}\subseteq \sourcepolyz$, which means that $\ba_0^=\times\invert{\vec x^{\,@}}= \invert{\vec 0}$. Also recall that $\calR_n^0\subseteq \calP^0$, which means that $\ba_0^=\times\invert{\vec x^{\,0}}= \invert{\vec 0}$. Therefore, 
$$\ba_0^=\times\invert{\vec x''}=C^@\sqrt n\ba_0^=\times\invert{\vec x^{\,@}}+ \ba_0^=\times\invert{\vec x^{\,0}} =  \invert{\vec 0} $$
\item[]{\bf \boldmath (1.2) $\ba_0^+\times\invert{\vec x''} < \invert{\vec 0}$.
} Let $\vec x^{\,0}= \sum\nolimits_{i=1}^{d_0-1} \frac{2C'}{C^0}\gamma_i^j\cdot\vec p_i^{\,0}$. Recall that $\ba_0^+\times\invert{\vec x^{\,@}}<\invert{-\epsilon^@\cdot\vec 1}$, and recall the lower bound of $C^@$ from (\ref{eq:C-at}). We have 
$$\ba_0^+\times\invert{\vec x''}=C^@\sqrt n\ba_0^+\times\invert{\vec x^{\,@}}+\sum\nolimits_{i=1}^{d_0-1} \frac{2C'}{C^0}\gamma_i^j\cdot \ba_0^+\times\invert{\vec p_i^{\,0}}<    \invert{\vec 0} $$
\end{itemize}

\vspace{2mm}\noindent{\bf\bf \boldmath\bf (2) $\vec x\in \sourcepoly$.} Recall that $\pi^*\in\sus{0}$, $\vec x^{\,\#}\in \csus{B^{\#},n^{\#}}\subseteq \sourcepoly$, and $\vec x^{\,\infty}\in \sus{\infty}\subseteq \sourcepolyz$. Therefore, 
\begin{align*}
\ba_{\text{S}}\times \invert{\vec x} =& \ell \ba_{\text{S}}\times {\underbrace{(\pi^*)}_{\text{in }\sus{0}}}^\top + \ba_{\text{S}}\times {\underbrace{(\vec x^{\,\#})}_{\text{in }\sourcepoly}}^\top + \ba_{\text{S}}\times {\underbrace{(C^@\sqrt n\vec x^{\,@} + \vec x^{\,0})}_{\text{in }\sus{0}}}^\top + \ba_{\text{S}}\times {\underbrace{(\vec x^{\,\infty})}_{\text{in }\sourcepolyz}}^\top\\
\le & \invert{\vec 0} + \invert{\vbb_\text{S}} + \invert{\vec 0}+ \invert{\vec 0} = \invert{\vbb_\text{S}}
\end{align*} 
This proves $\vec x\in \sourcepoly$. 

\vspace{2mm}\noindent{\bf\bf \boldmath\bf (3) $\vec x$ is $B$-manipulable.} Because $\vec x^{\,\#}\in \csus{B^{\#}, n^{\#}}$, there exists $\vo^{\,\#}\in{\mathbb Z}_{\ge 0}^{\vosize}$ such that $\vec c\cdot \vo^{\,\#}\le B^{\#}$ and $\vec x^{\,\#}+\vo^{\,\#}\times{\vomatrix{}} \in \targetpoly$.  Let $\vec x^{\,\infty} = \sum\nolimits_{i=1}^{d_{\Delta}} \frac{2C'}{C^\infty}\eta_i\cdot\vec p_i^{\,\infty} $. Also recall that for all $i\le d_{\Delta}$, $\vec p_i^{\,\infty} = \vec y_i^{\,\infty}- \vo_i^{\,\infty}\times{\vomatrix{}}$, where $\vec y_i^{\,\infty}\in \targetpolyz$, $\vo^{\,\infty}_i\ge \vec 0$, and  $\vec c\cdot \vo^{\,\infty}_i = 1$.  Let $\vo^* = \vo^{\,\#} +\sum\nolimits_{i=1}^{d_{\Delta}} \frac{2C'}{C^\infty}\eta_i\cdot \vo_i^{\,\infty} $. We prove that $\vec c\cdot \vo^*  \le B$ as follows.
\begin{align*}
\vec c \cdot \vo^* &= \vec c \cdot  \vo^{\,\#} +\sum\nolimits_{i=1}^{d_{\Delta}} \frac{2C'}{C^\infty}\eta_i\cdot \vec c \cdot\vo_i^{\,\infty} \\
& \le B^{\#} + \sum\nolimits_{i=1}^{d_{\Delta}} \frac{2C'}{C^\infty} \cdot \frac{(B-B^{\#})C^\infty}{2 d_{\Delta}C'} = B^{\#} + B-B^{\#}=B
\end{align*}
Then, 
\begin{align*}
&\ba_{\text{T}}\times \invert{\vec x + \vo^*\times \vomatrix{}} \\
= & \ell \ba_{\text{T}}\times {\underbrace{(\pi^*)}_{\text{in }\sus{0}}}^\top + \ba_{\text{T}}\times {\underbrace{(\vec x^{\,\#}+\vo^{\,\#}\times \vomatrix{})}_{\text{in }\targetpoly}}^\top + \ba_{\text{T}}\times {\underbrace{(C^@\sqrt n\vec x^{\,@} + \vec x^{\,0})}_{\text{in }\sus{0}}}^\top \\
&\hspace{40mm}+ \sum\nolimits_{i=1}^{d_{\Delta}} \frac{2C'}{C^\infty}\eta_i\cdot\ba_{\text{T}}\times {\underbrace{(\vec p_i^{\,\infty} +\vo_i^{\,\infty}\times\vomatrix{})}_{ = \vec y^{\,i}\in \targetpolyz}}^\top\\
\le & \invert{\vec 0} + \invert{\vbb_\text{T}} + \invert{\vec 0}+ \invert{\vec 0} = \invert{\vbb_\text{T}}
\end{align*} 
This proves that $\vec x + \vo^*\times \vomatrix{}\in \targetpoly$ and completes the proof of Claim~\ref{claim:calR-in-HB}.
\end{proof}
By Claim~\ref{claim:int-close}, for every $\vec x\in \calR_{B,n}$, there exists $\vec x'\in \csus{n,B}$ that is no more than $C'$ away from $\vec x$ in $L_\infty$. Because vectors in $\calR_{B,n}$ are at least $2C'$ away from each other, their corresponding vectors in $\csus{n,B}$ are different, denoted by $\calR^{\mathbb Z}_{B,n}$ and is formally defined as follows.
\begin{dfn}[\bf\boldmath $\calR^{\mathbb Z}_{B,n}$]
\label{dfn:R-B-n-Z} Let $\calR^{\mathbb Z}_{B,n}$ denote the integer vectors corresponding to the vectors in $\calR_{B,n}$ guaranteed by Claim~\ref{claim:int-close}.
\end{dfn}
It follows that $|\calR^{\mathbb Z}_{B,n}| = |\calR_{B,n}| =\Omega \left((B+1)^{d_{\Delta}}\cdot \left( \sqrt n\right)^{d_0-1}\right)$. This completes   \myhyperlink{lb-step1}{Step 1} for poly lower bound. 

\vspace{2mm}
\noindent{\bf\boldmath \hypertarget{lb-step2}{Step 2} for poly lower bound:  Define $\vec \pi^\circ$ for any $\pi^*\in\conv(\Pi)$.}  In this step, we define  $\vec \pi^\circ=(\pi^\circ_1,\ldots,\pi^\circ_{ n})\in \Pi^n$ such that $\sum_{j=1}^n\pi^\circ_j$ is $O(\sqrt n)$ away from $n\pi^*$ in a way that is similar to~\citep{Xia2021:How-Likely}. 

More precisely, because $\pi^*\in \conv(\Pi)$, by Carathéodory's theorem for convex/conic hulls (see e.g., \cite[p.~257]{Lovasz2009:Matching}), we can write  $\pi^*$ as the convex combination of $1\le t\le q$ distributions in $\Pi$, i.e., 
$$\pi^* = \sum\nolimits _{i=1}^t\alpha_i\pi_i^*, \text{ where }\sum\nolimits _{i=1}^t\alpha_i =1\text{ and }\pi_i^*\in \Pi$$
We note that $\pi^*\ge \epsilon\cdot \vec 1$, because $\Pi$ is strictly positive (by $\epsilon$).  For each $i\le t-1$, let $\vec \pi_i^\ell$ denote the vector of $\beta_i=\lfloor \ell \alpha_i\rfloor$ copies of $\pi_i^*$. Let $\vec \pi_k^\ell$ denote the vector of $\beta_t=n-\sum_{i=1}^{t-1} \beta_i$ copies of $\pi_t^*$. It follows that for any $i\le t-1$, $|\beta_i-\ell \alpha_i| \le 1$, and $|\beta_t-\ell \alpha_t|\le t + (\vec y^n\cdot \vec 1 - \ell \pi^{*}\cdot \vec 1) = O(\sqrt\ell) = O(\sqrt n)$. 

Let $\vec \pi^\circ = (\vec \pi_1^\ell,\ldots,\vec \pi_t^\ell)$, or equivalently
$$\vec \pi^\circ = (\underbrace{\pi_1^*,\ldots,\pi_1^*}_{\beta_1},\underbrace{\pi_2^*,\ldots,\pi_2^*}_{\beta_2},\ldots,  \underbrace{ \pi_t^*, \ldots, \pi_t^*}_{\beta_t})$$

\vspace{2mm}
\noindent{\bf\boldmath \hypertarget{lb-step3}{Step 3} for poly lower bound:  Lower-bound  $\Pr_{P\sim\vec\pi}(\vXp\in\cpoly{B})$.} 
\begin{claim}
\label{claim:PMV-point-lb}
For any instability settings $\vosetting{}$, any strictly positive set of distributions $\Pi$ (by $\epsilon>0$), and any $\alpha>0$,  there exist $C_{\vosetting}>0$ and $N>0$ such that for any $n\ge N$, any $0\le B\le \sqrt n$, and any $\vec \pi\in\Pi^n$ such that the $L_\infty$ distance between $\sum_{j=1}^n\pi_j$ and $\sus{0}$ is no more than $\alpha\sqrt n$,
$$\Pr(\vXp \in\csus{n,B}) \ge C_{\vosetting}\cdot (B+1)^{d_{\Delta}}\cdot\left( \frac{1}{\sqrt n}\right)^{q-d_0}$$
\end{claim}
\begin{proof}
Let $\pi^*\in\sus{0}$ denote an arbitrary vector such that $|\sum_{j=1}^n\pi_j-n \pi^*|_\infty<2\alpha\sqrt n$. Let $\calR^{\mathbb Z}_{B,n}$ denote the set of integer vectors for $\pi^*$ (Definition~\ref{dfn:R-B-n-Z} in \myhyperlink{lb-step1.4}{Step 1.4} above). Because $B = O(\sqrt n)$,  each vector in $\calR^{\mathbb Z}_{B,n}$ is $O(\sqrt n)$ away from $n \pi^*$, which is    $O(\sqrt n)$ away from $\sum_{j=1}^n\pi_j$. By the point-wise concentration bound (\citep[Lemma~1]{Xia2021:How-Likely}), for every $\vec x\in \calR^{\mathbb Z}_{B,n}$, we have $\Pr_{P\sim \vec\pi}(\hist(P) = \vec x) = \Omega((\sqrt n)^{1-q})$. Therefore,
\begin{align*}
&\Pr(\vXp \in\csus{n,B}) \ge \Pr\nolimits_{P\sim \vec\pi}(\hist(P) \in \calR^{\mathbb Z}_{B,n}) \ge |\calR^{\mathbb Z}_{B,n}|\times \Omega\left((\sqrt n)^{1-q}\right) \\
=& \Omega\left( (B+1)^{d_{\Delta}}\cdot (\sqrt n)^{d_0-1} \right)\times \Omega\left((\sqrt n)^{1-q}\right) = \Omega\left( (B+1)^{d_{\Delta}}\cdot (\frac{1}{\sqrt n})^{q-d_0} \right)
\end{align*} 
This completes proof of Claim~\ref{claim:PMV-point-lb}.
\end{proof}
The lower bound for PT-$\Theta(\sqrt n)$-$\sup$ follows after applying Claim~\ref{claim:PMV-point-lb} to any $\pi^*\in \conv(\Pi)\cap \sus{0}$ and $\vec \pi^\circ$ defined in \myhyperlink{lb-step2}{Step 2 for poly lower bound} above.

\vspace{3mm}
\noindent \hypertarget{PMV-proof-sup-n}{{\bf \boldmath Proof for the phase transition at $\Theta(n)$ case of $\sup$ (PT-$\Theta(n)$-$\sup$ for short).}} 

\vspace{3mm}
\noindent \hypertarget{PMV-proof-sup-n-smallB}{{\bf \boldmath $B\le C_2n$.}} Let $c^* = \frac{1}{2}(C_2+B_{\conv(\Pi)})$. We first show that $n\cdot\conv(\Pi)$ is $\Omega(n)$ away from $\sus{nc^*}$, which is equivalent to $\conv(\Pi)$ being $\Omega(1)$ away from $\sus{c^*}$.
Due to the minimality of $B_{\conv(\Pi)}$, we have $\conv(\Pi)\cap \sus{c^*}=\emptyset$. Notice that  $\conv(\Pi)$ is convex and compact  and $\sus{c^*}$ is convex. By the strict separating hyperplane theorem, the distance between $\conv(\Pi)$  and $\sus{c^*}$ is strictly positive, denoted by $c'$. Let $c_2>0$ denote any fixed constant that is smaller than $c'$. 

Next, we prove that $n\cdot\conv(\Pi)$ is $\Omega(n)$ away from $\cpoly{B}$. To this end, we prove the following claim, which states that for any sufficiently large $B'\ge 0$, $\cpoly{B'}$ is $O(1)$ away from $\sus{B'+O(1)}$. 

\begin{claim}
\label{claim:distance}
Given an instability setting $\vosetting$, there exists a constant $C$ such that for any $B'\ge C$, any $n$, and any $\vec x\in \cpoly{B'}$, there exists $\vec x'\in \sus{B'+C}$ such that $|\vec x-\vec x'|_\infty\le C$.
\end{claim}
\begin{proof}
The proof is done by analyzing feasible solutions to the following two linear programs LP$_{\poly}$ and LP$_{\cone }$, whose variables are $\vec x$ and $\vo$.
\renewcommand{\arraystretch}{1.5}
\begin{center}
\begin{tabular}{|c | c|}
\hline 
LP$_{\poly}$ &   LP$_{\cone }$\\
\hline
\begin{tabular}{rl}
$\max$ &  $0$\\
s.t. & $\ba_{\text{S}}\times \invert{\vec x}\le \invert{\vbb_{\text S}}$ \\
& $\ba_{\text{T}}\times \invert{\vec x +\vo\times \vomatrix{}}\le \invert{\vbb_{\text T}}$ \\
& $-\vo\le \vec 0$
\end{tabular}
&  
\begin{tabular}{rl}
$\max$ &  $0$\\
s.t. & $\ba_{\text{S}}\times \invert{\vec x}\le \invert{\vec 0}$ \\
& $\ba_{\text{T}}\times \invert{\vec x +\vo\times \vomatrix{}}\le \invert{\vec 0}$ \\
& $-\vo\le \vec 0$
\end{tabular}\\
\hline 
\end{tabular}
\end{center}
\renewcommand{\arraystretch}{1}
Because $\vec x\in \cpoly{B'}$, there exists $\vo\ge \vec 0$ such that $\vec x+\vo\times \vomatrix{}\in \targetpoly$ and $\vec c\cdot \vo\le B'$. Therefore,  $(\vec x,\vec w)$ is a feasible solution to LP$_{\poly}$. Notice that LP$_{\cone}$ is feasible (for example $\vec 0$ is a feasible solution) and LP$_{\poly}$ and LP$_{\cone}$ only differ on the right hand side of the inequalities. Therefore, due to~\citep[Theorem 5 (i)]{Cook86:Sensitivity}, there exists a feasible solution $(\vec x',\vec w')$  to LP$_{\cone}$ that is no more than $q\Delta \max\{|\vbb_{\text{S}}-\vec 0|_\infty,|\vbb_{\text{T}}-\vec 0|_\infty\} = O(1)$ away from $(\vec x,\vec w)$ in $L_\infty$, where $\Delta$ is the maximum absolute value of determinants of square sub-matrices of the left hand side of LP$_{\poly}$ and LP$_{\cone}$, i.e.,  $\multimatrix{\begin{array}{cc}\ba_{\text{S}} &  0 \\ \ba_{\text{T}} & \ba_{\text{T}}\times \invert{\vomatrix{}}\\ 0 & -{\mathbb I} \end{array}}$. This implies that 
$$\vec c\cdot\vo'\le \vec c\cdot \vec w + (\vec c\cdot\vec 1)q\Delta \max\{|\vbb_{\text{S}}|_\infty,|\vbb_{\text{T}}|_\infty\}$$
The claim follows after letting $C = (\vec c\cdot\vec 1)q\Delta \max\{|\vbb_{\text{S}}|_\infty,|\vbb_{\text{T}}|_\infty\}$.
\end{proof}
As proved above,  $n\cdot\conv(\Pi)$ is $\Omega(n)$ away from $\sus{nc^*}$ for any sufficiently large $n$.  
Notice that for any sufficiently large $n$, we have $(c^*-C_2)n>C$, where $C$ is the constant guaranteed by Claim~\ref{claim:distance}. Therefore, by Claim~\ref{claim:distance}, every vector $\vec x$ in $\cpoly{B}$ is $O(1)$ away from a vector $\vec x'$ in $\sus{B+O(1)} \subseteq \sus{nc^*}$. Recall from above that $n\cdot\conv(\Pi)$ is $\Omega(n)$ away from $\sus{nc^*}$. Therefore, $n\cdot\conv(\Pi)$ is $\Omega(n)$ away from $\cpoly{B}$.

Finally, for any $\vec\pi\in\Pi^n$, let $\pi' = \frac{1}{n}\sum_{j=1}^n\pi_j$. Because $\pi'\in\conv(\Pi)$, the $L_\infty$ distance between $\pi'$, which is the mean vector of $\vXp$, and $\cpoly{B}$ is $\Omega(n)$.   The (exponential) upper bound for the $B\le C_2n$  case is proved by a straightforward application of Hoeffding's inequality and the union bound to all $q$ dimensions as in the proof of the \myhyperlink{PMV-proof-sup-exp}{exponential  case  of $\sup$}. The exponential lower bound  trivially holds. 

\vspace{3mm}
\noindent \hypertarget{PMV-proof-sup-n-largeB}{{\bf \boldmath $B\ge C_3n$.}} The $O\left((\frac{1}{\sqrt n})^{q-d_{\infty}} \right)$ upper bound of this case follows after Claim~\ref{claim:poly-upper-strong}.  At a high level, the proof of the $\Omega\left((\frac{1}{\sqrt n})^{q-d_{\infty}} \right)$ lower bound  is similar to the proof of the \myhyperlink{PMV-proof-sup-sqrtn-lb}{polynomial lower bound  of PT-$\Theta(\sqrt n)$-$\sup$}. The main difference is that  the condition is weaker now ($\conv(\Pi)\cap\sus{\infty}\ne\emptyset$ compared to $\sus{0}\cap\conv(\Pi)\ne \emptyset$). Therefore, we will identify two different sets  $\calR^+_{B,n}\subseteq \cpolyn{B}$ and $\calR^{\mathbb Z+}_{B,n}\subseteq \csus{n,B}$ for any strictly positive $\pi^+\in \cone_{B^+}$ with $\pi^+\cdot\vec 1 =1$, where $B^+$ is a fixed number such that $C_3> B^+ > B_{\conv(\Pi)}$.  It follows that $\pi^+\in \sourcepolyz$ and we can write $\pi^+ = \vec y^{\,+} - \vo^{\,+}\times\vomatrix{}$, where $\vec y^{\,+}\in \targetpolyz$, $\vo^{\,+}\ge \vec 0$, and  $ \vec c\cdot \vo^{\,+} \le  B^+ $. 

In the following procedure, we define a basis $\bB^+$ of $\calP_\infty$ that is similar to $B^\infty$. Let $\{\vec p_1^{\,+},\ldots,\vec p_{d_\infty}^{\,+}\}$ denote a set of $d_\infty$ linearly independent vectors in  $\sus{\infty}\subseteq \calP_\infty$. W.l.o.g.~suppose for every $j\le d_\infty$,  ${\vec p_j^{\,+}}\cdot\vec 1\in \{-1,0,1\}$---otherwise we divide ${\vec p_j^{\,+}}$ by $|{\vec p_j^{\,+}}\cdot\vec 1|$. For every $j\le d_\infty$, let $\vec p_j^{\,+}  = \vec y_j^{\,+} - \vo_j^{\,+}\times \vomatrix{}$, where $\vec y_j^{\,+}\in \targetpolyz$ and $\vo_j^{\,+}\ge \vec 0$. For convenience, let $\vec p_0^{\,+} \triangleq  \pi^+$.

\vspace{2mm}\noindent{\bf\bf Procedure.} Start with $\bB^+ = \{\vec p_0^{\,+} \}$.  For every $1\le j\le d_\infty$, we add $\vec p_j^{\,+}$ to $\bB^+$ if and only if it is linearly independent of existing vectors in $\bB^+$. At the end of the procedure, we have $|\bB^+| = d_\infty$.  W.l.o.g., let $\bB^+ = \{\vec p_0^{\,+}, \vec p_1^{\,+},\ldots,\vec p_{d_\infty-1}^{\,+}\}$.

Next, we define ${\mathbb L^+}$, $C^+$, and $\calR_n^+$ that are similar to $ {\mathbb L}^0$, $ C^0$, and $\calR_n^0$ in \myhyperlink{lb-step1.2}{Step 1.2} of the proof of the \myhyperlink{PMV-proof-sup-sqrtn-lb}{polynomial lower bound  of PT-$\Theta(\sqrt n)$-$\sup$}, respectively.

Let ${\mathbb L^+}$ denote the lattice generated by $\bB^+ $ excluding $\vec 0$. That is, 
$${\mathbb L^+} = \left\{ \sum\nolimits_{i=0}^{d_{\infty}-1} \lambda_i\cdot\vec p_i^{\,+}: \forall {0\le i\le d_{\infty}-1}, \lambda_i\in {\mathbb Z}\text{ and } \exists 0\le  i\le d_{\infty}-1 \text{ s.t. }  \lambda_i \ne 0\right\}$$
Let $C^+$ denote the minimum $L_\infty$ norm of vectors in  ${\mathbb L^+}$. That is, 
$$C^+ =  \inf \left\{\left|\vec x \right|_\infty: \vec x\in\hat {\mathbb L} \right\}$$
Following a similar argument that proves $C^0>0$, we have $C^+>0$.  Then, define
 \begin{equation*}
\calR_n^+ = \left\{\sum\nolimits_{i=1}^{d_{\infty}-1} \left\lceil\frac{2C'}{C^+}\right\rceil\cdot\lambda_i\cdot\vec p_i^{\,+}: \forall 1\le i\le d_{\infty}-1,\lambda_i\in \{0,1,\ldots, \lfloor \sqrt n\rfloor\}\right\}
\end{equation*}
Notice that vectors in $\calR_n^+$ do  not have $\vec p_0 =\pi^+$ components.
Recall that $\vec x^\#\in \csus{B^{\#},n^{\#}}$, which means that $\vec x^\#\cdot\vec 1 = n^\#$. For any $n\in\mathbb N$, define
$$\calR^+_{B,n}  =\left\{ (n- n^\# -\vec x^{\,+}\cdot\vec 1)\cdot  \pi^+ + \vec x^{\,\#} + \vec x^{\,+}:   \vec x^{\,+} \in \calR_n^+\right\}$$
Like Claim~\ref{claim:far-away}, in the following claim we prove that vectors in $\calR^+_{B,n}$ are at least $2C'$ from each other in $L_\infty$. 

\begin{claim}[\bf\boldmath Sparsity of  $\calR_{B,n}^+$]
\label{claim:far-away+}
For any pair of  vectors $\vec x^1,\vec x^2\in \calR_{B,n}^+$ whose $\calR_n^+$ components are different, we have $|\vec x^1-\vec x^2|_\infty \ge  2C'$.
\end{claim}
\begin{proof}
For $j\in \{1,2\}$, we write
$$\vec x^j = \ell^j \cdot \pi^+ + \vec x^{\,\#} + \sum\nolimits_{i=1}^{d_{\infty}-1} \left\lceil\frac{2C'}{C^+}\right\rceil\cdot\lambda_i^j\cdot\vec p_i^{\,+},$$
where $\ell^1$ and $\ell^2$ guarantee that $\vec x^1\cdot\vec 1 = \vec x^2\cdot\vec 1 = n$.
Then, because $(\vec x^1 - \vec x^2)\cdot \vec 1=0$ and $\pi^+\cdot\vec 1 =1$, we have
$$\ell^1-\ell^2 = \sum\nolimits_{i=1}^{d_\infty-1} \left\lceil\frac{2C'}{C^+}\right\rceil(\lambda_i^2-\lambda_i^1)\cdot(\vec p_i^{\,+}\cdot \vec 1)$$
Therefore,
\begin{align*}
\vec x^1 - \vec x^2 = & (\ell^1-\ell^2)\pi^+ +  \sum\nolimits_{i=1}^{d_\infty-1} \left\lceil\frac{2C'}{C^+}\right\rceil(\lambda_i^1-\lambda_i^2)\cdot\vec p_i^{\,+}\\
=& \sum\nolimits_{i=1}^{d_\infty-1} \left\lceil\frac{2C'}{C^+}\right\rceil(\lambda_i^2-\lambda_i^1)\cdot(\vec p_i^{\,+}\cdot \vec 1)\cdot\pi^+ +  \sum\nolimits_{i=1}^{d_\infty-1} \left\lceil\frac{2C'}{C^+}\right\rceil(\lambda_i^1-\lambda_i^2)\cdot\vec p_i^{\,+}\\
=&\left\lceil\frac{2C'}{C^+}\right\rceil\left( \sum\nolimits_{i=1}^{d_\infty-1}  (\lambda_i^2-\lambda_i^1)\cdot(\vec p_i^{\,+}\cdot \vec 1)\cdot \pi^+ +  \sum\nolimits_{i=1}^{d_\infty-1}  (\lambda_i^1-\lambda_i^2)\cdot\vec p_i^{\,+} \right)
\end{align*}
Recall that for all $i\le d_\infty-1$, $\vec p_i^{\,+}\cdot \vec 1$ is an integer, which means that $\sum\nolimits_{i=1}^{d_\infty-1}  (\lambda_i^2-\lambda_i^1)\cdot(\vec p_i^{\,+}\cdot \vec 1)$ is an integer. Also because the $\calR_n^+$ components of $\vec x^1$ and $\vec x^2$ are different,  there exists $i\le d_\infty-1$ such that $\lambda_i^1-\lambda_2^2\ne 0$. Therefore, 
$$ \sum\nolimits_{i=1}^{d_\infty-1}  (\lambda_i^2-\lambda_i^1)\cdot(\vec p_i^{\,+}\cdot \vec 1)\cdot\pi^+ +  \sum\nolimits_{i=1}^{d_\infty-1}  (\lambda_i^1-\lambda_i^2)\cdot\vec p_i^{\,+} \in {\mathbb L^+}$$
Therefore, 
\begin{align*}
\left|\vec x^1 - \vec x^2\right|_\infty = & (\ell^1-\ell^2)\pi^+ +  \sum\nolimits_{i=1}^{d_\infty-1} \left\lceil\frac{2C'}{C^+}\right\rceil(\lambda_i^1-\lambda_i^2)\cdot\vec p_i^{\,+}\\
=& \sum\nolimits_{i=1}^{d_\infty-1} \left\lceil\frac{2C'}{C^+}\right\rceil(\lambda_i^2-\lambda_i^1)\cdot(\vec p_i^{\,+}\cdot \vec 1)\cdot\pi^+ +  \sum\nolimits_{i=1}^{d_\infty-1} \left\lceil\frac{2C'}{C^+}\right\rceil(\lambda_i^1-\lambda_i^2)\cdot\vec p_i^{\,+}\\
=&\left\lceil\frac{2C'}{C^+}\right\rceil\left| \sum\nolimits_{i=1}^{d_\infty-1}  (\lambda_i^2-\lambda_i^1)\cdot(\vec p_i^{\,+}\cdot \vec 1)\cdot \pi^+ +  \sum\nolimits_{i=1}^{d_\infty-1}  (\lambda_i^1-\lambda_i^2)\cdot\vec p_i^{\,+} \right|_\infty\\
\ge & \left\lceil\frac{2C'}{C^+}\right\rceil\cdot C^+ =2C'
\end{align*}
This completes the proof of Claim~\ref{claim:far-away+}.
\end{proof}
Then, we prove the counterpart of Claim~\ref{claim:calR-in-HB} for $\calR_{B,n}^+$ in the following claim. 
\begin{claim}[\bf\boldmath $\calR_{B,n}^+\subseteq \cpolyn{B}\cap{\mathbb R_{\ge 0}^q}$]
\label{claim:calR-in-HB+}
There exists $N\in\mathbb N$ that does not depend on $B$ or $n$, such that when $n\ge N$, $\calR_{B,n}^+\subseteq \cpolyn{B}\cap{\mathbb R}_{\ge 0}^q$. 
\end{claim}
\begin{proof}
 Let $\vec x  = (n- n^\# -\vec x^{\,+}\cdot\vec 1)\cdot  \pi^+ + \vec x^{\,\#} + \vec x^{\,+}$ denote any vector in $\calR^+_{B,n}$, where $\vec x^{\,+}=\sum\nolimits_{i=1}^{d_{\infty}-1} \left\lceil\frac{2C'}{C^+}\right\rceil\cdot \lambda_i\cdot\vec p_i^{\,+}$.  Clearly $\calR_{B,n}^+\subseteq \ {\mathbb R}_{\ge 0}^q$ for sufficiently large $n$, because $\pi^+$ is strictly positive. It is not hard to verify that $\vec x \in \sourcepoly$, because $ \pi^+\in \sourcepolyz$, $\vec x^{\,\#}\in \sourcepoly$, and $ \vec x^{\,+}\in \sourcepolyz$.   Let 
$$\vo = (n- n^{\#}-\vec x^{\,+}\cdot\vec 1)\cdot \vo^{\,+}+\vo^{\,\#}+\sum\nolimits_{i=1}^{d_{\infty}-1} \left\lceil\frac{2C'}{C^+}\right\rceil\cdot \lambda_i \cdot\vo_i^{\,+}$$ 
It follows that $\vec c\cdot \vec w \le B^+n + B^{\#}+O(\sqrt n)$, which  is smaller than $B\ge C_3n$ for any sufficiently large $n$, because $C_3>B^+$. Then,
$$\vec x+\vec w\times \vomatrix{} = \left(n- n^{\#}-\vec x^{\,+}\cdot\vec 1\right)\cdot  \underbrace{\vec y^{\,+}}_{\text{in }\targetpolyz}+  \underbrace{\vec y^{\,\#}}_{\text{in }\targetpoly} + \sum\nolimits_{i=1}^{d_{\infty}-1} \left\lceil\frac{2C'}{C^+}\right\rceil\cdot\lambda_i\cdot\underbrace{\vec y_i}_{\text{in }\targetpolyz}\in  \targetpoly$$
Therefore, $\calR^+_{B,n}\subseteq \cpolyn{B}$.  This completes the proof of Claim~\ref{claim:calR-in-HB+}.
\end{proof}
Let $\calR^{\mathbb Z+}_{B,n}\subseteq \csus{n,B}$ denote the   integer vectors corresponding to vectors in $\calR^+_{B,n}$  guaranteed by Claim~\ref{claim:int-close}. Then, $|\calR^{\mathbb Z+}_{B,n}| = \Omega((\sqrt n)^{d_\infty-1})$. The $\Omega\left((\frac{1}{\sqrt n})^{q-d_{\infty}} \right)$ lower bound follows after  similar steps as \myhyperlink{lb-step2}{Step 2} and \myhyperlink{lb-step3}{Step 3} for the \myhyperlink{PMV-proof-sup-sqrtn-lb}{polynomial lower bound  of PT-$\Theta(\sqrt n)$-$\sup$}. Specifically, we have the following counterpart to Claim~\ref{claim:PMV-point-lb} with a  similar proof.

\begin{claim}
\label{claim:PMV-PTn-point-lb}
For any instability settings $\vosetting{}$, any strictly positive set of distributions $\Pi$ such that $B^-_{\conv(\Pi)}<\infty$, and any $C^-_3>B^-_{\conv(\Pi)}$,  there exist $C_{\vosetting}>0$ and $N>0$ such that for any $n\ge N$, any $  B\ge C^-_3 n$, and any $\vec \pi\in\Pi^n$,
$$\Pr(\vXp \in\csus{n,B}) \ge C_{\vosetting}\cdot  \left( \frac{1}{\sqrt n}\right)^{q-d_\infty}$$
\end{claim}
 
\vspace{3mm}
\noindent \hypertarget{inf}{{\bf\boldmath Proof for $\inf$.}} Like the $\sup$ part, we call the four cases in $\inf$ part the $0$ case, the exponential case, the phase transition at $\Theta(\sqrt n)$ case, and the phase transition at $\Theta(n)$ case. 

 The $0$ case is straightforward.  To prove the exponential upper bound, we need to prove that there exists $\vec\pi\in\Pi^n$ such that $\Pr(\vXp\in \csus{n,B})=\exp(-\Omega(n))$.
 Let  $\pi'$ denote an arbitrary vector in $\conv(\Pi)\setminus \sus{\infty}$. Because $ \sus{\infty}$ is a closed cone, the distance between $\pi'$ and $\sus{\infty}$ is strictly positive, which means that the distance between  $n\cdot\pi'$ and $\sus{\infty}$ is $\Omega(n)$.      Then, let $\vec\pi\in\Pi^n$ denote an arbitrary vector such that $\frac{1}{n}\sum_{j=1}^n\pi_j$ is $O(1)$ from $\pi'$ in $L_\infty$. Therefore, $\sum_{j=1}^n\pi_j$, which is the mean vector of $\vXp$, is $\Theta(n)$ away from $\sus{\infty}\supseteq \sus{B}$. Then, similar to the \myhyperlink{PMV-proof-sup-exp}{proof for the exponential  case  of $\sup$}, by Claim~\ref{claim:cc-h-inf}, $\sum_{j=1}^n\pi_j$ is $\Theta(n)$ away from $\cpoly{B}\supseteq \csus{n,B}$. The exponential upper bound of $\inf$ then follows after Hoeffding's inequality and the union bound. The exponential lower bound trivially holds.

\vspace{2mm}
\noindent{\bf\boldmath Proof for the phase transition at $\Theta(\sqrt n)$ case of $\inf$.} The (polynomial) upper bound  follows after the upper bound of the phase transition at $\Theta(\sqrt n)$ case of $\sup$.  The polynomial lower bound follows after applying Claim~\ref{claim:PMV-point-lb} to every $\vec \pi=(\pi_1,\ldots,\pi_n)\in\Pi^n$ and letting $\pi^* = \frac{1}{n}\sum_{j=1}^n\pi_j\in\sus{0}$.

\vspace{2mm}
\noindent{\bf\boldmath Proof for the phase transition at $\Theta(n)$ case of $\inf$.} When $B\le C^-_2n$, the (exponential) 
lower bound is straightforward. The proof for the (exponential) upper bound is similar to the proof of the \myhyperlink{PMV-proof-sup-n-smallB}{$B\le C_2n$ case of PT-$\Theta(n)$-$\sup$}. More precisely, let $c^* = 
\frac{1}{2}(C^-_2+B^-_{\conv(\Pi)})$. According to the minimality of $B^-_{\conv(\Pi)}$, we have $\conv(\Pi)\nsubseteq \sus{c^*}$. Let $\pi'\in \conv(\Pi)\setminus \sus{c^*}$, which means that $\pi'$ is $\Omega(1)$ away from $\sus{c^*}$, because $\sus{c^*}$ is a closed set. Let $\vec\pi\in\Pi^n$ denote an arbitrary vector such that $\sum_{j=1}^n\pi_j$ is $O(1)$ from $n\cdot\pi'$ in $L_\infty$. Then, $\sum_{j=1}^n\pi_j$ is $\Omega(n)$ away from $\sus{c^*}$. By Claim~\ref{claim:distance}, $\sum_{j=1}^n\pi_j$, which is the mean vector of $\vXp$, is $\Omega(n)$ away from $\cpoly{B}$ in $L_\infty$. The  exponential upper bound follows after Hoeffding's inequality and the union bound. 


When $B\ge C^-_3n$, the polynomial upper bound follows after the $\sup$ case. To prove the polynomial lower bound, notice that $C^-_3>B^-_{\conv(\Pi)}\ge B_{\conv(\Pi)}$. Let  $B^+$ be any number such that $B^-_{\conv(\Pi)}<B^+<C^-_3$. This means that $\conv(\Pi)\subseteq \sus{B^+}$. Because $B^-_{\conv(\Pi)}\ge B_{\conv(\Pi)}$, we have $B^+>B_{\conv(\Pi)}$.  The proof for the lower bound is similar to the proof for the 
lower bound in the phase transition at $\Theta(n)$ case of $\inf$: notice that in the \noindent \myhyperlink{PMV-proof-sup-n-largeB}{proof of the $B\ge C_3n$ case of PT-$\Theta(n)$-$\sup$}, the proof works for any $\pi^+\in\conv(\Pi)\cap \sus{B^+}=\conv(\Pi)$. 
\end{proof}


\subsection{Full Version of Theorem~\ref{thm:PMV-instability-psi>0} and Its Proof}
\label{sec:psi>0}


\appThm{Semi-Random PMV-Instability, $\psi>0$}
{thm:PMV-instability-psi>0}{ Given any $q\in\mathbb N$, any closed and strictly positive $\Pi$ over $[q]$, and any instability setting $\vosetting{} = \langle\sourcepoly,\targetpoly,\voset{},\vec c\,\rangle$, any $C_1>0$, any $n\in \mathbb N$, and any $0\le B\le C_1 \sqrt n$,  
$$\sup_{\vec\pi\in\Pi^n}\Pr\left(\vXp\in \nb{\csus{n,B}}{\psi} \right)= 
\begin{cases}
0 & \text{if }\csus{n,B}=\emptyset\\
\exp(-\Theta(n)) & \text{otherwise, if }\conv(\Pi)\cap \nbs{\sus{0}}{\psi} =\emptyset\\
\Theta(1)&\text{otherwise}
\end{cases}$$ 
$$\inf_{\vec\pi\in\Pi^n}\Pr\left(\vXp\in \nb{\csus{n,B}}{\psi} \right)= 
\begin{cases}
0 & \text{if }\csus{n,B}=\emptyset\\
\exp(-\Theta(n)) & \text{otherwise, if }\conv(\Pi)\cap \nbs{\sus{0}}{\psi} =\emptyset\\
\Theta(1)&\text{otherwise}
\end{cases}$$
}
\begin{proof}  The $0$ case of $\sup$ is straightforward.  The proof of the exponential case of $\sup$ is similar to the proof of the exponential case of $\sup$ for $\psi=0$. We will show that there is a linear separation between $\conv(\Pi)$ and $\nb{\csus{n,B}}{\psi}$. Because $\csus{n,B}\subseteq \cpoly{B}$, it suffices to show the separation between $\conv(\Pi)$ and $\nb{\cpoly{B}}{\psi}$. To this end, we first recall from Claim~\ref{claim:close-to-cone0} that every vector $\vec x\in \cpoly{B}$ is no more than $\Theta(B+1)$ away from  $\sus{0}$. Let $\psi^*>\psi$ be an arbitrary number such that $\conv(\Pi)\cap \nbs{\sus{0}}{\psi^*} =\emptyset$. Then, as $B= O(\sqrt n)$, $\nb{\csus{n,B}}{\psi}\subseteq \nbs{\sus{0}}{\psi^*}$. 

Next, we prove that $\nbs{\sus{0}}{\psi^*}$ is a polyhedron. We first define  LP$_\psi$, whose variables are $\vec \chi,\vec\kappa$, and $\vec x$. Let $\momatrix{}$ be the $q(q-1) \times q$ matrix that represents the changes that the contamination adversary can make to a single vote, i.e., every row contains a $1$  and a $-1$.  $\vec \kappa$ represents the modification by the contamination adversary.
\begin{equation} 
\label{eq:LPpsi}
\text{LP}_{\psi} \triangleq \left\{
\text{\begin{tabular}{rl}
$\max$ &\ \ \ $ 0$\\
s.t. & $\ba_{\text{S}} \times \invert{\vec x} \le \invert{\vec 0}$\\
 & $\ba_{\text{T}} \times \invert{\vec x} \le \invert{\vec 0}$\\
& $ \invert{\vec \chi} = \invert{\vec x} + \invert{\vec\kappa\times \momatrix{}}$\\
& $\vec\kappa \cdot \vec 1 \le \psi \cdot \vec x \cdot \vec 1$\\
   &  $\vec x  \ge \vec 0$,  $\vec \chi  \ge \vec 0$, $\vec \kappa \ge \vec 0$
\end{tabular}}
\right.
\end{equation}
Clearly,  $\nbs{\sus{0}}{\psi^*}$ is  the projection of the feasible solutions of LP$_\psi$ (which is a polyhedron in the $(\vec \chi,\vec\kappa,\vec x)$ space) to $\vec\chi$'s.  The rest of proof for the exponential case is the same as the exponential case of  $\sup$ for $\psi=0$. 

The proof of the $\Theta(1)$ case is similar to the proof of the lower bound of the polynomial case of $\sup$ for $\psi=0$. The main difference and challenge is to define the counterpart of \myhyperlink{lb-step1.4}{Step 1.4}: we will specify $\Theta(n^{(q-1)/2})$ integer vectors in an $O(\sqrt n)$ neighborhood of the sum of distributions in some $\vec \pi^*\in \Pi^n$, such that each such vector can be modified by a contamination adversary  to a vector $\vec x\in \sus{0}$.

Let $\pi^*\in \conv(\Pi)\cap \nbs{\sus{0}}{\psi}$. Clearly $n\pi^*\in \nbs{\sus{0}}{\psi}$, which means that $\sus{0}\cap \nbs{n\pi^*}{\psi}\ne\emptyset$. Choose $\vec x\in \sus{0}\cap \nbs{n\pi^*}{\psi}$. Recall from the proof of Claim~\ref{claim:close-to-cone0} that $\sus{0}$ are the feasible solutions to LP$_{\cone}^B$, which means that $\vec x$ is a feasible solution to LP$_{\cone}^B$. Then, as proved in the proof of Claim~\ref{claim:close-to-cone0}, solutions to LP$_{\cone}^B$ and solutions to LP$_{\poly}^B$ are close to each other (as an application of~\citep[Theorem 5 (i)]{Cook86:Sensitivity}), there exists a solution $\vec x' $ to LP$_{\poly}^B$ such that $|\vec x - \vec x'|_1 = O(B+1)$. It follows that $\vec x' \in\cpoly{B}$. Define
$$\vec \chi^\circ \triangleq n\pi^* + (\vec x'  - \vec x)$$
It follows that $\vec \chi^\circ \in \nbs{\cpoly{B}}{\psi}$.   Let $\vec\chi^*$ denote the vector obtained from $\vec \chi^\circ $ by moving towards the $\vec x'$ direction so that the $L_1$ distance between $\vec \chi^\circ $ and $\vec\chi^*$ is $2\sqrt n$. That is,
$$\vec\chi^* \triangleq  \vec \chi^\circ  + 2\sqrt n  \cdot \frac{\vec x' - \vec \chi^\circ }{|\vec x' - \vec \chi^\circ |_1}$$
Let $B_{\sqrt n}(\vec\chi^*)$ denote the set of all vectors whose $L_1$ distance to $\vec\chi^*$ is no more than $\sqrt n$. When $n$ is sufficiently large, all vectors in $B_{\sqrt n}(\vec\chi^*)$ are in $\mathbb R_{>0}^q$.  In such cases, 
$$B_{\sqrt n}(\vec\chi^*) \subseteq \nbs{\cpoly{B}}{\psi},$$
because for any $\vec\chi\in B_{\sqrt n}(\vec\chi^*)$, $|\vec\chi|_1 =n$ and  $|\vec\chi - \vec x'|_1\le 2n\psi$. Moreover, the dimension of $B_{\sqrt n}(\vec\chi^*)$ is $q-1$. Therefore, we can choose and fix a (sparse) subset $\calR_{B,n,\psi}$ of $\Theta(n^{(q-1)/2})$ of $B_{\sqrt n}(\vec\chi^*)$ whose vectors are $\Theta(1)$ away from each other. 

Next, like Claim~\ref{claim:int-close}, for every $\vec\chi\in B_{\sqrt n}(\vec\chi^*)$, we will identify a vector $\vec \chi'\in \nb{\csus{n,0}}{\psi}$ that is $O(1)$ away by considering the following LP.
\begin{equation} 
\label{eq:LPpsi}
\text{LP}_{n,B,\psi} \triangleq \left\{
\text{\begin{tabular}{rl}
$\max$ &\ \ \ $ 0$\\
s.t. & $\ba_{\text{S}} \times \invert{\vec x} \le \invert{\vbb_\text{S}}$\\
   & $\ba_{\text{T}} \times \invert{\vec x + \vo\times \vomatrix{}} \le \invert{\vbb_\text{T}}$\\ 
   &  $\vec x\cdot\vec 1 = n$\\
  &  $\vec c\cdot \vo \le B$\\
& $ \invert{\vec \chi} = \invert{\vec x} + \invert{\vec\kappa\times \momatrix{}}$\\
& $\vec\kappa \cdot \vec 1 \le \psi \cdot \vec x \cdot \vec 1$\\
   &  $\vec \chi  \ge \vec 0$, $\vec \kappa \ge \vec 0$, $\vec x  \ge \vec 0$, $\vo  \ge \vec 0$,  
\end{tabular}}
\right.
\end{equation}
Because $\vec\chi\in B_{\sqrt n}(\vec\chi^*)\subseteq \nbs{\cpoly{B}}{\psi}$, there exist $\vec \kappa$, $\vec x$, and $\vo$ such that $(\vec\chi,\vec \kappa,\vec x,\vo)$ is a feasible solution to \text{LP}$_{n,B,\psi}$.  Because $\csus{n,B}\ne\emptyset$, \text{LP}$_{n,B,\psi}$ has a feasible integer solution.  Therefore, like in the proof of Claim~\ref{claim:int-close} and following~\citep[Theorem 1]{Cook86:Sensitivity}, \text{LP}$_{n,B,\psi}$ has an integer solution $\vec\chi'$ that is $\Theta(1)$ away from $\vec\chi$. 

The rest of the proof for the $\Theta(1)$ case of $\sup$ is the same as the proof of the lower bound of the polynomial case of $\sup$ for $\psi=0$. The $\inf$ part is similar to the $\inf$ part for $\psi=0$ with similar modifications with the $\sup$ part above.
\end{proof}

\section{Materials for Section~\ref{sec:applications}}

\subsection{Full Version of Theorem~\ref{thm:PMV-instability-applications}  and Its Proof}
\label{app:applications-proofs}

\appThm{\bf Max-Semi-Random Coalitional Influence}
{thm:PMV-instability-applications}{ 
Let $r$ be an integer  positional scoring rule, STV, ranked pairs, Schulze, maximin,   Copeland, plurality with runoff, or Bucklin with lexicographic tie-breaking. For  any closed and strictly positive $\Pi$ with $\piuni\in\conv(\Pi)$, any $X\in \{\cm,\mov\}\cup\econtrol$ (except $X\in \econtrol$ when $r=\pcopeland{0}$), there exists $ N>0$  such that for all $n>N$ and all $B\ge 0$,
$$\satmax{X}{\Pi,0}(r ,n,B) =  \begin{cases}0&\text{if }B=0\\\Theta\left(\min\left\{\frac{B}{\sqrt n},1\right\}\right)  &\text{if }B\ge 1\end{cases}$$
For any $X\in \control$,  any $n>N$ and any $B\ge 0$,
$\satmax{X}{\Pi,0}(r ,n,B) =  \Theta\left(1\right)$. Moreover, for any $\psi>0$, $\satmax{X}{\Pi,\psi}(r ,n,B) =\Theta(1)$ for the non-zero cases above. 
}

\begin{proof}
We first prove the {\bf \boldmath $B=0$ case}. For any $X\in \{\cm,\mov\}\cup\econtrol$, it is not hard to verify that if no budget is given, then the goal  under $X$ cannot be reached, which requires the winner to be changed. Therefore, $\satmax{X}{\Pi,0}(r ,n,B) =0$. For any $X\in \control$, it suffices to prove that for any alternative $a$, 
\begin{equation}
\label{equ:applications-B0}
\sup\nolimits_{\vec\pi\in\Pi^n} \Pr\nolimits_{P\sim \vec\pi}(r(P)=\{a\}) = \Theta(1)
\end{equation}
It is not hard to see that for all voting rules mentioned in this statement of the theorem, there exists a polyhedron $\ppoly{a}$ such that for all $\vec x\in\ppoly{a}$, $r(\vec x) = \{a\}$, $\piuni\in\ppolyz{a}$, and $\dim(\ppolyz{a}) = m!$. Therefore, \eqref{equ:applications-B0} follows after~\citep[Theorem 1]{Xia2021:How-Likely} (or equivalently, Theorem~\ref{thm:PMV-instability} with $\sourcepoly=\targetpoly=\cpoly{a}$ and $B=0$).

In the rest of the proof, we assume that $X\in \{\cm,\mov\}\cup\econtrol$ and $B\ge 1$. 

\vspace{2mm}\noindent{\bf\bf Overview.} Due to Theorem~\ref{thm:PMV-instability-GSR-upper}, it suffices to prove the  $\Theta\left(\min\left\{\frac{B}{\sqrt n},1\right\}\right)$ matching lower bound by identifying an instability setting $\vosetting$ that represents some unstable histograms $\csus{n,B}$, such that for all $B\ge 1$, 
$$\satmax{X}{\Pi,0}(r ,n,B)  \ge  \sup\nolimits_{\vec\pi\in \Pi^n}\Pr\nolimits_{P\sim\vec \pi}(\hist(P)\in\csus{n,B})\text{, and }$$
\begin{equation}
\label{eq:B-ineq}
 \sup\nolimits_{\vec\pi\in \Pi^n}\Pr\nolimits_{P\sim\vec \pi}(\hist(P)\in\csus{n,B}) = \Theta\left(\min\left\{\frac{B}{\sqrt n},1\right\}\right)
\end{equation}
Notice that for any constant $C_1$ and any $B\ge C_1\sqrt n$, the right hands side of \eqref{eq:B-ineq} is $\Theta(1)$. Therefore, it suffices to prove \eqref{eq:B-ineq} for all $1\le B\le C_1 \sqrt n$. For each $X$   and each  voting rule $r$  in the statement of the theorem, we define $\vosetting= \langle\sourcepoly, \targetpoly, \voset{},\vec 1\,\rangle$, prove that $\csus{n,1}\ne\emptyset$ for any sufficiently large $n$ by construction, and prove that $d_0=m!-1$ (which if often obvious), $d_\infty= m!$ (by applying Claim~\ref{claim:dim-infty-sufficient}).

We  now  prove Theorem~\ref{thm:PMV-instability-applications} for  $X=\cm$, and then comment on how to modify the proof for other $X\in \{\mov\}\cup\econtrol$.

\vspace{2mm}\noindent{\bf\bf \bf\boldmath $\cm$ for integer positional scoring rules.}  
Let $r = r_{\vec s}$ denote the positional scoring rule with scoring vector $\vec s$. Let $\sourcepoly$ denote the set of vectors where $1$'s total score is at least as high as $2$'s total score, which is strictly higher than the total score of any other alternative. Let $\targetpoly$ denote the set of vectors where $2$'s total score is strictly the highest. To formally define $\sourcepoly$ and $\targetpoly$, we first recall the definition of score difference vectors. 
\begin{dfn}[\bf Score difference vector~\citep{Xia2020:The-Smoothed}]\label{dfn:varcons} For any scoring vector $\vec s = (s_1,\ldots,s_m)$ and any pair of different alternatives $a,b$, let $\score_{a,b}^{\vec s}$ denote the $m!$-dimensional vector indexed by rankings in $\ml(\ma)$: for any $R\in\ml(\ma)$, the $R$-element of $\score_{a,b}^{\vec s}$ is $s_{j_1}-s_{j_2}$, where $j_1$ and $j_2$ are the ranks of $a$ and $b$ in $R$, respectively.
\end{dfn}
In words, $\score_{a,b}^{\vec s}$ is the score vector of $a$ (under all linear orders) minus the score vector of $b$. Then, we define
$$\sourcepoly\triangleq\left\{\vec x: \begin{split}\score_{2,1}^{\vec s}\cdot\vec x&\le 0\\ \forall i\ge 3, \  \score_{i,2}^{\vec s}\cdot\vec x&\le -1\\   -\vec x&\le  \vec 0\end{split}\right\}\text{, }\targetpoly\triangleq \left\{\vec x:  \begin{split} \forall i\ne 2, \ \score_{i,2}^{\vec s}\cdot\vec x&\le -1\\   -\vec x&\le  \vec 0\end{split}\right\}\text{, and}$$
$$\vosetting_{\vec s} = \langle\sourcepoly, \targetpoly, \voset{\pm}^{1\ra2},\vec 1\,\rangle$$
It is not hard to verify that for any $\vec y\in \sourcepoly$ and any $\vec x\in\targetpoly$, we have $r_{\vec s}(\vec y) = \{1\}$ ($1$ has the highest score and wins due to tie-breaking if $2$ also has the highest score) and $r_{\vec s}(\vec x) = \{2\}$ ($2$ has the strictly highest score).

Next, we show that $\csus{n,B}\ne\emptyset$ for any sufficiently large $n$ by constructing a successful instance of manipulation by a single voter.  We first define some profiles and rankings  that will be used in the rest of the proof. For any $a\in \ma$, let $\sigma_a$ denote a cyclic permutation among $\ma\setminus\{a\}$. Let $P_a$ denote the following $(m-1)$-profile.
$$P_a \triangleq \left\{\sigma_a^i(a\succ \others): 1\le i\le m-1 \right\}\cup 3\times \ml(\ma),$$
where alternatives in ``$\others$'' are ranked alphabetically. Let 
$$P_* \triangleq 2m\times (P_1\cup P_2) \cup \bigcup_{i=3}^m 2(m-i) \times P_i$$
It follows that the $\vec s(P_*,1)=\vec s(P_*,2)>\vec s(P_*,3)>\cdots>\vec s(P_*,m)$.

Let $R_1$ (respectively, $R_2$) denote the ranking where $1$ (respectively, $2$) is ranked at the top, $2$ (respectively, $1$) is ranked at the bottom, and the remaining alternatives are ranked alphabetically. That is,
$$R_1 \triangleq [1\succ 3\succ\cdots\succ m\succ 2]\text{ and }R_2 \triangleq [2\succ 3\succ\cdots\succ m\succ 1]$$
Let $\ell\le m-1$ denote the index to the minimum value of $s_{\ell}- s_{\ell+1}$. Let $R_2'$ denote the ranking where $2$ and $1$ are ranked at the $\ell$-th and the $(\ell+1)$-th positions respectively, and the remaining alternatives are ranked alphabetically. That is,
$$R_2' \triangleq \underbrace{3\succ \cdots\succ \ell+1}_{\ell-1}\succ 2\succ 1\succ \underbrace{l+2\succ \cdots\succ m}_{m-\ell-1}$$
Next, we define $\sourceprofile$ and $\targetprofile$. We first define $\sourceprofile'$ to be the $n$-profile that consists of as many copies of $P_*$ as possible, and the remaining rankings are $R_1$. That is, let $n' \triangleq \left\lfloor\frac{n}{|P_*|}\right\rfloor\times |P_*|$, and
$$\sourceprofile'\triangleq \frac{n'}{|P_*|}\times P_* \cup (n-n')\times \{R_1\}$$
Let $\sourceprofile$ denote the profile obtained from $\sourceprofile'$ by replacing $\left\lfloor\frac{(n-n')(s_1-s_m)}{s_1+s_{\ell+1}-s_m-s_\ell}\right\rfloor$ copies of $R_2'$ by $R_2$. It follows that  for any sufficiently large $n$ (so that $\sourceprofile'$ contains enough copies of $R_2'$ and the score difference between $1$ and any alternative $i\ge 3$ is sufficiently large), $\sourceprofile$ is well-defined and $r_{\vec s}(\sourceprofile)=\{1\}$. Let $\targetprofile$ be obtained from $\sourceprofile$ by replacing an $R_2'$ vote by an $R_2$ vote. It follows that $r_{\vec s}(\targetprofile)=\{2\}$. This proves that $\csus{n,1}\ne\emptyset$ for any sufficiently large $n$. 

It is not hard to verify that $d_0  = m!-1$ (the only implicit equality represents $1$ and $2$ have the same score). Let $\sourceprofile^*$ be the profile obtained from $\sourceprofile$ by replacing an $R_2$ vote by an $R_2'$ vote. It follows that   $\hist(\sourceprofile^*)$ and $\hist(\targetprofile)$ are interior points of $\sourcepolyz$ and $\targetpolyz$, respectively.  Therefore, By Claim~\ref{claim:dim-infty-sufficient}, $d_\infty  = m!$. The lower bound \eqref{eq:B-ineq} follows after the polynomial case of $\sup$ of Theorem~\ref{thm:PMV-instability} (applied to $\vosetting_{\vec s}$).  This completes the proof of Theorem~\ref{thm:PMV-instability-applications} for $\cm$ under integer positional scoring rules.

\vspace{2mm}\noindent{\bf\bf \boldmath $\cm$ for STV.} Let $\sourcepoly$ consists of the histograms (for which the STV winner is $1$) where the execution of STV satisfies the following conditions:
\begin{itemize}
\item   for every $1\le i\le m-4$, in round $i$, alternative $m+1-i$ has the strictly lowest plurality score among the remaining alternatives; 
\item In round $m-3$, alternative $3$ has the highest score, and the score of  $2$ is no more than the score of $1$;
\item if $1$ loses in round $m-3$, then $2$ would become the winner; and if $2$ loses in round $m-3$, then $1$ would become the winner.
\end{itemize}
Formally, let us recall the score difference vector (for a pair of alternatives $a,b$, after a set of alternatives $B$ is removed) to define $\sourcepoly$ and $\targetpoly$.
\begin{dfn}[\citep{Xia2021:Semi-Random}]
For any pair of alternatives $a,b$ and any subset of alternatives $B\subseteq(\ma\setminus\{a,b\})$, we let $\scorediff{B,a,b}$ denote the vector, where for every $R\in\ml(\ma)$, the $R$-th component of $\scorediff{B,a,b}$ is the plurality score of $a$ minus the plurality score of $b$ in $R$ after alternatives in $B$ are removed. 
\end{dfn}
Then, $\sourcepoly$ consists of vectors $\vec x$ that satisfies the following linear constraints.
\begin{itemize}
\item For every $i\le m-4$ and every   $i'<i$,  
$$\scorediff{\{m+2-i,\ldots,m\},m+1-i,i'}\cdot \vec x\le -1$$
\item Let $B_3 = \{4,\ldots,m\}$. There are two constraints:  $\scorediff{B_3,2,1}\cdot \vec x\le 0$ and $\scorediff{B_3,1,3}\cdot \vec x\le -1$.
\item  $\scorediff{B_3\cup\{1\},3,2}\cdot \vec x\le -1$ and $\scorediff{B_3\cup\{2\},3,1}\cdot \vec x\le -1$.
\item For every $R\in\ml(\ma)$, there is a constraint $-x_R\le 0$.
\end{itemize}

Let $\targetpoly$ denote the polyhedron that differs from $\sourcepoly$ in round $m-3$, where $1$ has the lowest plurality score and drops out, which means that $2$ is the STV winner. It is not hard to verify that for all $\vec y\in \sourcepoly$ and all $\vec x\in \targetpoly$, we have $\stv(\vec y) =\{1\}$ and $\stv(\vec x) =\{2\}$.

Let $\vosetting_\stv \triangleq \langle\sourcepoly, \targetpoly, \voset{\pm}^{1\ra 2},\vec 1\,\rangle$.  Next, we construct profiles $P_\text{S}$ and $P_\text{T}$ to show that $\csus{n,B}\ne\emptyset$ for any sufficiently large $n$ and $B\ge 1$. For any $a\in\ma$, let $P_a^*$ denote the $(m-1)!$-profile that is obtained from $\ml(\ma\setminus\{a\})$ by putting $a$ at the top. Let
$$P^* \triangleq \bigcup_{i=4}^m (m-i)\times P_a^* \text{ and }n^* \triangleq |P^*|$$
Let 
\begin{align*}
P_\text{S} \triangleq & \left\lfloor \frac{n-n^*}{3}-1\right\rfloor\times \{ [1\succ 2\succ\others], [2\succ 1\succ\others]\}\\
& + \left(n-n^*+2-2\left\lfloor \frac{n-n^*}{3} \right\rfloor\right)\times \{3\succ 2\succ 1\succ \others \} +P^*
\end{align*}
It follows that $|P_\text{S}| = n$, and for all $n\ge n^* + 3 m(m-1)!$, $\hist(P_\text{S})\in \sourcepoly$.
Let $P_\text{T}$ be the profile obtained from $P_\text{S}$ by replacing one vote of $[3\succ 2\succ 1\succ \others]$ by $[2\succ 3\succ 1\succ \others]$. It is not hard to verify that $\hist(P_\text{T})\in \targetpoly$. Therefore, $\csus{n,B}\ne\emptyset$ for every $B\ge 1$.

It is not hard to verify that $d_0=\dim(\sourcepolyz\cap \targetpolyz) = m!-1$, where the only implicit equality requires $1$ and $2$ are tied for the last place in round $m-3$. To see $d_\infty = m!$, let $\sourceprofile^*$ denote the profile obtained from $\sourceprofile$ by adding one vote of $[1\succ 2\succ \others]$. Let 
$\vec y = \hist(\sourceprofile^*)$ and let $\vec x$ be the histogram of the profile  obtained from $\sourceprofile^*$ by changing two votes of $[3\succ 2\succ 1\succ \others]$ to $[2\succ 1\succ 3\succ \others]$. It is not hard to verify that $\vec y$ is an interior point of $\sourcepolyz$, $\vec y$ is an interior point of $\targetpolyz$, $\dim(\sourcepolyz) = \dim(\targetpolyz)=m!$.  By Claim~\ref{claim:dim-infty-sufficient}, we have $d_\infty = m!$.  

Then, \eqref{eq:B-ineq} follows after the application of the $\sup$ part of Theorem~\ref{thm:PMV-instability} to the instability settings $\vosetting_\stv$ for $\cm$ under $\stv$.

\vspace{2mm}\noindent{\bf\bf \boldmath $\cm$ for Ranked Pairs,  Schulze, and maximin.} The proof for the three rules share the same construction. 
Let $\sourcepoly$ denote the polyhedron that consists of vectors $\vec x$ whose WMG satisfies the following conditions.
\begin{itemize}
\item The weights on the following edges are strictly positive: $1\ra 2$, $2\ra 3$, $3\ra 1$, $\{1,2,3\}\ra \{4,\ldots,m\}$.
\item For all $i\ge 4$, the weight  on $1\ra i$ is strictly larger than the weight on $1\ra 2$.
\item $w_{\vec x}(2\ra 3)>w_{\vec x}(1\ra 2)\ge w_{\vec x}(3\ra 1)$.
\end{itemize}
See Figure~\ref{fig:rp-schulze-B1-odd} (a) for an example of WMG that satisfies these conditions.  Formally, we first recall the pairwise difference vectors as follows.
\begin{dfn}[\bf Pairwise difference vectors~\citep{Xia2020:The-Smoothed}]\label{dfn:pairdiff} For any pair of different alternatives $a,b$, let $\pair_{a,b}$ denote the $m!$-dimensional vector indexed by rankings in $\ml(\ma)$: for any $R\in\ml(\ma)$, the $R$-component of $\pair_{a,b}$ is $1$ if $a\succ_R b$; otherwise it is $-1$. 
\end{dfn}
Then, let $\sourcepoly$ be characterized by the following linear inequalities/constraints:
\begin{itemize}
\item For each  edge $a\ra b \in \{1\ra 2, 2\ra 3, 3\ra 1\}\cup \left( \{1,2,3\}\ra \{4,\ldots,m\}\right)$, there is a constraint $\pair_{b,a}\cdot\vec x \le -1$. 
\item For all $i\ge 4$, $(\pair_{1,2} - \pair_{1,i}) \cdot\vec x \le -1$.
\item $(\pair_{1,2} - \pair_{2,3}) \cdot\vec x \le -1$ and $(\pair_{3,1} - \pair_{1,2}) \cdot\vec x \le 0$.
\item For all linear order $R\in\ml(\ma)$, there is a constraint $-x_R\le 0$.
\end{itemize}

Let $\targetpoly$ denote the polyhedron that consists of vectors $\vec x$ whose WMG satisfies the same conditions as $\targetpoly$, except that now it is required that $w_{\vec x}(2\ra 3)>w_{\vec x}(3\ra 1)>w_{\vec x}(1\ra 2)$. See Figure~\ref{fig:rp-schulze-B1-odd} (b) for an example of WMG that satisfies these conditions for odd $n$.  We have $\dim(\targetpolyz) = m!-1$ (the implicit equality is $(\pair_{3,1} - \pair_{2,3}) \cdot\vec x =0 $) and $\dim(\targetpolyz) = m!$.
\begin{figure}[htp]
\centering
\begin{tabular}{ccc}
\includegraphics[width = .3\textwidth]{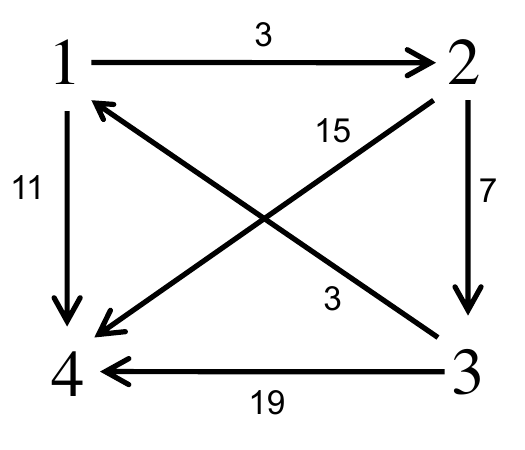} & \includegraphics[width = .3\textwidth]{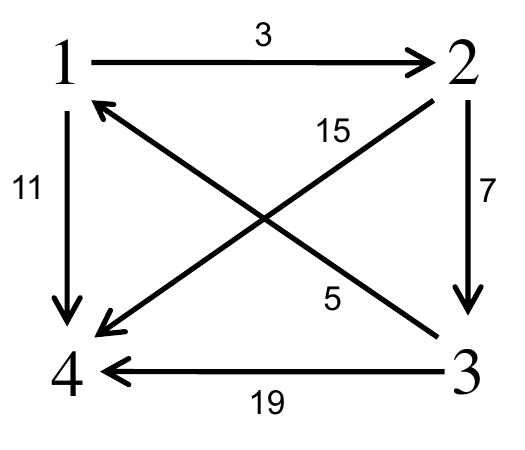}& \includegraphics[width = .3\textwidth]{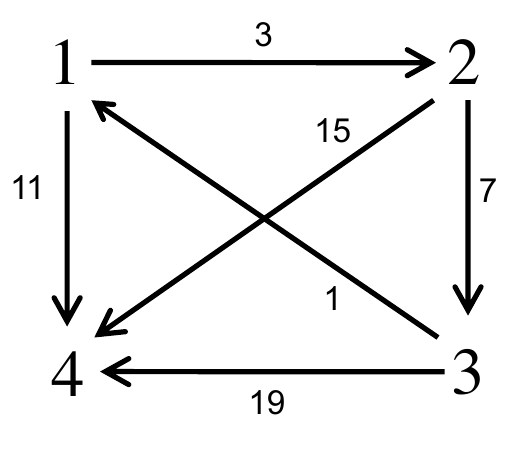} \\
(a) $\wmg(P_\text{S})$ & (b) $\wmg(P_\text{T})$  & $\wmg(\vec y)$
\end{tabular}
\caption{$\cm$ under $\rp$, $\schulze$, and $\maximin$.\label{fig:rp-schulze-B1-odd}}
\end{figure}

Let $\vosetting = \langle\sourcepoly, \targetpoly, \voset{\pm}^{1\ra 2},\vec 1\,\rangle$. It follows that for any $\vec y\in \sourcepoly$ and any $\vec x\in \targetpoly$, $\rp(\vec y) = \{1\}$ ($1\ra 2$ is fixed before $3\ra 1$ due to tie-breaking) and $\rp(\vec x) = \{2\}$ ($3\ra 1$ is fixed before $1\ra 2$); $\schulze(\vec y) = \{1\}$ ($1$ and $2$  are co-winners, so $1$ wins due to tie-breaking) and $\schulze(\vec x) = \{2\}$ ($2$ is the unique winner); and $\maximin(\vec y) = \{1\}$ ($1$ and $2$  are co-winners, so $1$ wins due to tie-breaking) and $\maximin(\vec x) = \{2\}$ ($2$ is the unique winner).

Next, we define $n$-profiles $P_\text{S}$ and $P_\text{T}$ to show that $\csus{n,1}\ne\emptyset$. The construction depends on the parity of $n$. If $n$ is odd, let $P_\text{S}$ denote a profile whose WMG is the figure shown in Figure~\ref{fig:rp-schulze-B1-odd} (a), where the weights on $1\ra 2$ and $3\ra 1$ are the same.  Due to McGarvey's theorem~\citep{McGarvey53:Theorem}, such $P_\text{S}$ exists for all sufficiently large odd number $n$, and we can assume that $P_\text{S}$ contains two copies of $\ml(\ma)$.  Let $P_\text{T}$ denote the profile obtained from  $P_\text{S}$ by replacing a $[2\succ 1\succ 3\succ\others]$ vote by $[2\succ 3\succ 1\succ\others]$, which means that the WMG of $P_\text{T}$ is Figure~\ref{fig:rp-schulze-B1-odd} (b).  If $n$ is even, then let the positive weights on edges in the WMGs of $P_\text{S}$ and $P_\text{T}$ be one more than those for odd $n$, so that all weights become even numbers.  In either case, it is not hard to verify that  $\csus{n,B}\ne\emptyset$ for every sufficiently large $n$ and every $B$.

It is not hard to verify that $d_0=m!-1$. To see $d_\infty = m!$, let  $\vec y$ be any vector such that $\wmg(\vec x)$ 
is the same as Figure~\ref{fig:rp-schulze-B1-odd} (c). Let $\vec x$ denote the vector obtained from $\vec y$ by replacing two votes of $[2\succ 1\succ 3\succ\others]$  by $[2\succ 3\succ 1\succ\others]$. It follows that the $\wmg(\vec x)$ is Figure~\ref{fig:rp-schulze-B1-odd} (b). Notice that $\vec y$ is an interior point of $ \sourcepolyz$; $\vec x$ is an interior point of $\targetpolyz$; and $\dim(\sourcepolyz)=\dim(\targetpolyz)=m!$. By Claim~\ref{claim:dim-infty-sufficient}, we have $d_{\infty}=m!$.  Then, \eqref{eq:B-ineq} follows after the application of the $\sup$ part of Theorem~\ref{thm:PMV-instability} to the instability settings  for $\cm$ under ranked pairs,  Schulze, and maximin.

\vspace{2mm}\noindent{\bf\bf \boldmath $\cm$ for Copeland.} 

\vspace{2mm}\noindent{\bf\bf \bf\boldmath $\cm$,  odd $n$.} 
Let $\sourcepoly$ denote the polyhedron that consists of vectors $\vec x$ whose UMG has the following edges: $1\ra 3$, $3\ra 2$, $2\ra 1$, $\{1,2,3\}\ra \{4,\ldots,m\}$.  
See Figure~\ref{fig:cd-odd} (a) for an example for $m=4$.   Formally, $\sourcepoly$ is characterized by the following linear inequalities/constraints:
\begin{itemize}
\item For each  edge $a\ra b \in \{1\ra 3, 3\ra 2, 2\ra 1\}\cup (\{1,2,3\}\ra \{4,\ldots,m\})$, there is a constraint $\pair_{b,a}\cdot\vec x \le -1$.  
\item For all linear order $R\in\ml(\ma)$, there is a constraint $-x_R\le 0$.
\end{itemize}

Let $\targetpoly$ denote the polyhedron that consists of vectors $\vec x$ whose UMG has the same edges $\sourcepoly$, except that the direction between $2$ and $ 3$ is flipped, i.e.,   $w_{\vec x}(3\ra 2) \le -1$. See Figure~\ref{fig:cd-odd} (b) for an example of the UMG for $m=4$.   
\begin{figure}[htp]
\centering
\begin{tabular}{ccc}
\includegraphics[width = .3\textwidth]{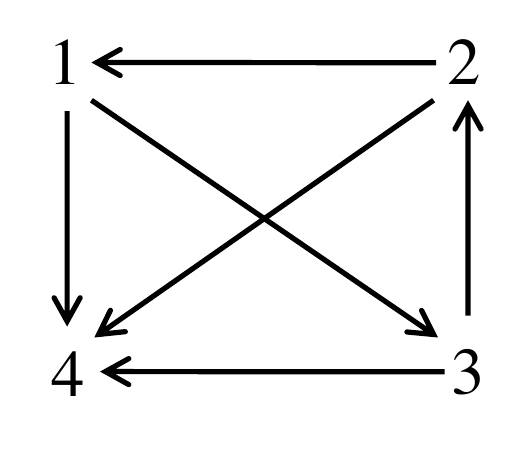} & \includegraphics[width = .3\textwidth]{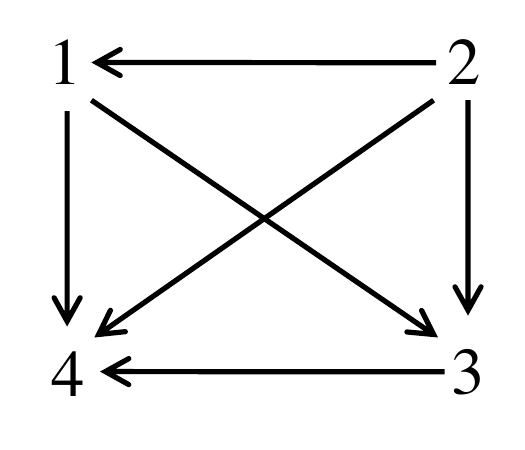}& \includegraphics[width = .3\textwidth]{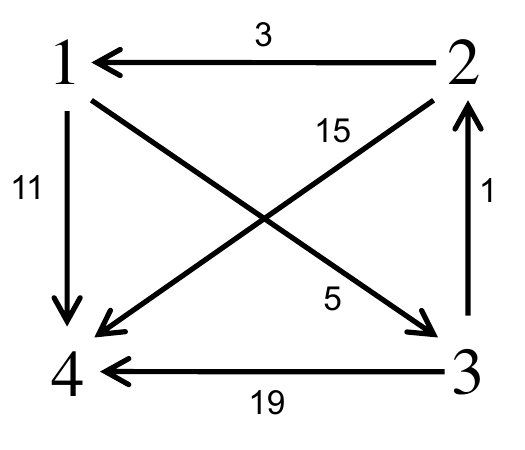}\\
(a) UMG for $\sourcepoly$& (a) UMG for $\targetpoly$ & (c) $\wmg(P_\text{S})$ 
\end{tabular}
\caption{ $\cm$ under $\copeland$, odd $n$.\label{fig:cd-odd}}
\end{figure}

Let $\vosetting_{\copeland} = \langle\sourcepoly, \targetpoly, \voset{\pm}^{1\ra 2},\vec 1\,\rangle$. It follows that for any $\vec y\in \sourcepoly$ and any $\vec x\in \targetpoly$, $\copeland(\vec y) = \{1\}$ ($1,2,3$ have the same highest Copeland score, and then $1$ wins due to tie-breaking) and $\copeland(\vec x) = \{2\}$ ($2$ is the Condorcet winner).

Next, we define $n$-profiles $P_\text{S}$ and $P_\text{T}$ to show that $\csus{n,1}\ne\emptyset$.  Let $P_\text{S}$ denote a profile whose WMG is the figure shown in Figure~\ref{fig:cd-odd} (c), where the weight  on $3\ra 2$ is $1$.  Due to McGarvey's theorem~\citep{McGarvey53:Theorem}, such $P_\text{S}$ exists for all sufficiently large odd number $n$, and we can assume that $P_\text{S}$ contains   $\ml(\ma)$.  Let $P_\text{T}$ denote the profile obtained from  $P_\text{S}$ by replacing a $[3\succ 2\succ 1 \succ\others]$ vote by $[2\succ 3\succ 1\succ\others]$, which means that the UMG of $P_\text{T}$ is Figure~\ref{fig:cd-odd} (b).

Recall that $P_\text{T}$ is obtained from $P_\text{S}$ by replacing a vote. Therefore, $\hist(P_\text{S})\in\csus{n,1}$, which means that $\csus{n,1}\ne\emptyset$. It is not hard to verify that $d_0=m!-1$ (the implicit equality corresponds to the tie between $2$ and $3$). Notice that $\hist(P_\text{S})$ and $\hist(P_\text{T})$ are interior points of $\sourcepolyz$ and $ \targetpolyz$, respectively, and $\dim(\sourcepolyz) = \dim(\targetpolyz) = m!$.   By Claim~\ref{claim:dim-infty-sufficient}, we have $d_\infty = m!$. Also notice that $\piuni\in \sourcepolyz\cap \targetpolyz$. Therefore, by Theorem~\ref{thm:PMV-instability}, for any  $B\ge 1$ and any sufficiently large odd $n$ , we have 
$$\sup\nolimits_{\vec\pi\in \Pi^n}\Pr\nolimits_{P\sim\vec \pi}(\hist(P)\in\csus{n,B} ) =\Theta\left(\min\left\{\frac{B}{\sqrt n},1\right\}\right)$$

\vspace{2mm}\noindent{\bf\bf \bf\boldmath $\cm$,  even $n$, $\alpha >0$.} Let $\sourcepoly$ be the same as defined for the odd $n$ case above (the UMG of all vectors in $\sourcepoly$ is illustrated in Figure~\ref{fig:cd-even-B1-alpha} (a)). 
Let $\targetpoly$ be the polyhedron  that consists of vectors $\vec x$ whose UMG satisfies the following conditions: the weights on the following edges are strictly positive: $1\ra 3$,  $2\ra 1$, $\{1,2,3\}\ra \{4,\ldots,m\}$. In addition, we require that $\wmg(\vec x)$ does not contain the edge $3\ra 2$, i.e., we require $w_{\vec x} (2\ra 3)\ge 0$.
See Figure~\ref{fig:cd-even-B1-alpha} (b) for the UMG for $m=4$ (where the dashed edge from $2$ to $3$ means that either there is no edge between $2$ and $3$, or there is an edge $2\ra 3$).   Formally, $\targetpoly$ is characterized by the following linear inequalities/constraints:
\begin{itemize}
\item For each  edge $a\ra b \in \{1\ra 3,  2\ra 1\}\cup ( \{1,2,3\}\ra \{4,\ldots,m\})$, there is a constraint $\pair_{b,a}\cdot\vec x \le -1$.  
\item $ \pair_{3,2}   \cdot\vec x \le  0$.
\item For all linear order $R\in\ml(\ma)$, there is a constraint $-x_R\le 0$.
\end{itemize}
 
\begin{figure}[htp]
\centering
\begin{tabular}{ccc}
\includegraphics[width = .3\textwidth]{fig/CDUMG-1.pdf} & \includegraphics[width = .3\textwidth]{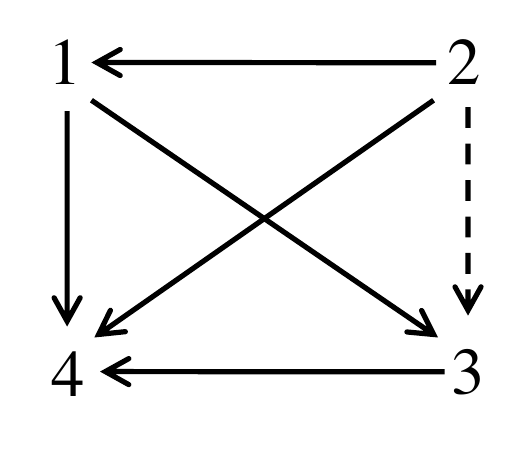}& \includegraphics[width = .3\textwidth]{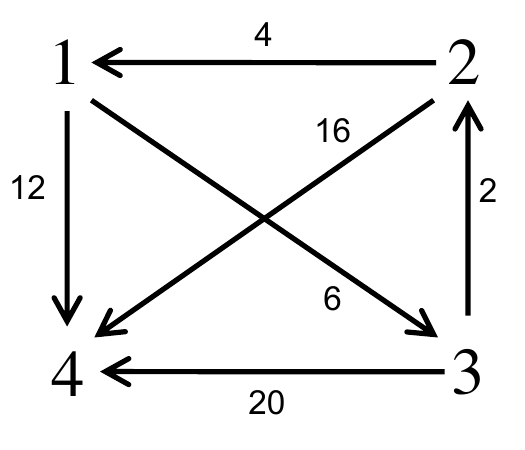}\\
(a) UMG for $\sourcepoly$& (b) UMG for $\targetpoly$ & (c) $\wmg(P_\text{S})$ 
\end{tabular}
\caption{$\eccav$ under  $\copeland$,   even $n$, $\alpha>0$.\label{fig:cd-even-B1-alpha}}
\end{figure}
Let $\vosetting_{\copeland} = \langle\sourcepoly, \targetpoly, \voset{\pm}^{1\ra 2},\vec 1\,\rangle$. It follows that for any $\vec y\in \sourcepoly$ and any $\vec x\in \targetpoly$, $\copeland(\vec y) = \{1\}$ ($1,2,3$ have the same highest Copeland score, and then $1$ wins due to tie-breaking) and $\copeland(\vec x) = \{2\}$ ($2$ has the highest Copeland score, which is at least $m-2+\alpha$).

Next, we define $n$-profiles $\sourceprofile$ and $\targetprofile$ to show that $\csus{n,1}\ne\emptyset$.  Let $\sourceprofile$ denote any $n$-profile that contains $\ml(\ma)$, whose UMG is as in Figure~\ref{fig:cd-even-B1-alpha} (a), and its existence is guaranteed by McGarvey's theorem~\citep{McGarvey53:Theorem}.   Let $\targetprofile$ denote the profile obtained from  $\sourceprofile$ by replacing a $[3\succ 2\succ 1 \succ\others]$ vote by $[2\succ 3\succ 1\succ\others]$, which means that the UMG of $\targetprofile$ is consistent Figure~\ref{fig:cd-even-B1-alpha} (b) (and there is no edge between $2$ and $3$ in $\umg(\targetprofile)$). 

To characterize  $d_\infty$, we prove the following convenient claim for general PMV-instability problems (for general $q$) that will be frequently used in the proofs of this paper.
\begin{claim}
\label{claim:dim-infty-sufficient}
Suppose  $\dim(\targetpolyz)=q$ and $\sourcepolyz$ has an interior point that can be manipulated to an interior point of $\sourcepolyz$, then $d_\infty =  \dim(\sourcepolyz)$. 
\end{claim}
\begin{proof} Because $\sus{\infty} \subseteq \sourcepolyz$, we have $d_\infty \le \dim(\sourcepolyz)$. Next, we prove that $d_\infty \ge \dim(\sourcepolyz)$. Let $\vec y = \vec x -\vo\times\vomatrix{}$ denote an interior point $ \sourcepolyz$, where $\vec x$ is an interior point of $ \targetpolyz$ and $\vo\ge \vec 0$ represents  successful operations (without budget constraints). Let $C>0$ denote an arbitrary number so that the $C$ neighborhood of $\vec x$ in $L_\infty$ is contained in $\dim(\targetpolyz)$. Therefore, the $C$ neighborhood of $\vec y$ in $L_\infty$ is contained in $\targetpolyz+\calQ_{\vec c\cdot\vo}$. It follows that any vector $\vec y' \in \sourcepolyz$ that is at most $C$ away from $\vec y$ in $L_\infty$ is in $\sus{\infty}$. Because $\vec y$ is an interior point of $\sourcepolyz$, we have $d_\infty \ge \dim(\sourcepolyz)$, which proves Claim~\ref{claim:dim-infty-sufficient}.
\end{proof} 

It is not hard to verify that $d_0=m!-1$ and   $ d_\infty = m!$ (by Claim~\ref{claim:dim-infty-sufficient}). Notice that $\piuni\in \sourcepolyz\cap \targetpolyz$.  The even $n$ and $\alpha >0$ case follows after Theorem~\ref{thm:PMV-instability}.

\vspace{2mm}\noindent{\bf\bf \bf\boldmath $\cm$, even $n$, $\alpha = 0$.} 
Let $\sourcepoly$ be the polyhedron  that consists of vectors $\vec x$ whose UMG contains $1\ra 3$,  $2\ra 1$, $\{1,2,3\}\ra \{4,\ldots,m\}$. In addition, we require that $\wmg(\vec x)$ does not contain $2\ra 3$, that is, $w_{\vec x} (3\ra 2)\ge 0$.
See Figure~\ref{fig:cd-even-B1-alpha0} (a) 
for the UMG for $m=4$ (where the dashed edge from $3$ to $2$ means that either there is no edge between $2$ and $3$, or there is an edge $3\ra 2$).   Formally, $\sourcepoly$ is characterized by the following linear inequalities/constraints:
\begin{itemize}
\item For each  edge $a\ra b \in \{1\ra 3,  2\ra 1\}\cup (\{1,2,3\}\ra \{4,\ldots,m\})$, there is a constraint $\pair_{b,a}\cdot\vec x \le -1$.  
\item $ \pair_{2,3}   \cdot\vec x \le  0$.
\item For all linear order $R\in\ml(\ma)$, there is a constraint $-x_R\le 0$.
\end{itemize}

Then, we let $\targetpoly$ be the same as $\targetpoly$  in the odd $n$ case above (as illustrated in Figure~\ref{fig:cd-even-B1-alpha0} (b) for $m=4$).  Let $\vosetting_{\copeland} = \langle\sourcepoly, \targetpoly, \voset{\pm}^{1\ra 2},\vec 1\,\rangle$. It follows that for any $\vec y\in \sourcepoly$ and any $\vec x\in \targetpoly$, $\copeland(\vec y) = \{1\}$ ($1$ and $2$ have the same highest Copeland score, so $1$ wins due to tie-breaking) and $\copeland(\vec x) = \{2\}$ ($2$ has the strictly highest Copeland score $m-1$).

Let $\sourceprofile$ denote an arbitrary $n$-profile whose WMG is as shown in Figure~\ref{fig:cd-even-B1-alpha0} (c) and it contains two copies of $\ml(\ma)$. Let $\targetprofile$ denote the $n$-profile obtained from $\sourceprofile$ by  replacing a $[3\succ 2\succ 1 \succ\others]$ vote by $[2\succ 3\succ 1\succ\others]$, which means that the UMG of $P_\text{T}$ is like Figure~\ref{fig:cd-even-B1-alpha0} (b).  It follows that $\hist(P_\text{S})\in\sourcepoly$ and $\hist(P_\text{T})\in \targetpoly$, which proves that $\csus{n,1}\ne\emptyset$.  
\begin{figure}[htp]
\centering
\begin{tabular}{@{}cccc@{}}
\includegraphics[width = .22\textwidth]{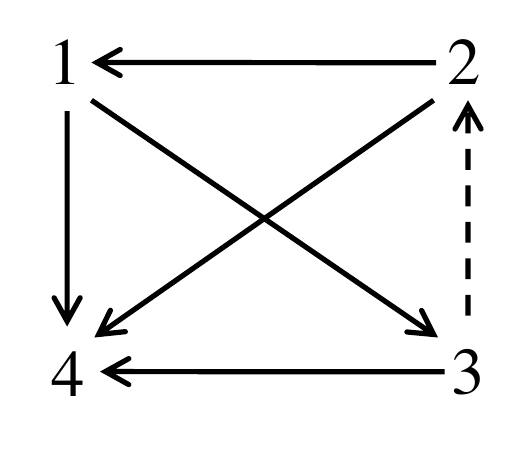} & \includegraphics[width = .22\textwidth]{fig/CDUMG-2.pdf} &
\includegraphics[width = .22\textwidth]{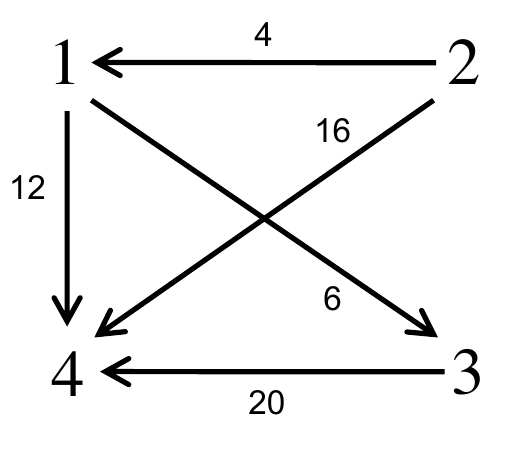} & \includegraphics[width = .22\textwidth]{fig/CDS-even-alpha.pdf}\\
(a) UMG for $\sourcepoly$& (b) UMG for $\targetpoly$ & (c) $\wmg(P_\text{S})$  & (d) $\wmg(\vec y)$ \end{tabular}
\caption{$\eccav$ under  $\copeland$,   even $n$, $\alpha=0$.\label{fig:cd-even-B1-alpha0}}
\end{figure}

It is not hard to verify that $d_0=m!-1$. To see $d_\infty = m!$, let  $\vec y$ be any vector such that $\wmg(\vec x)$ 
is the same as Figure~\ref{fig:cd-even-B1-alpha0} (d) and it contains two copies of $[3\succ 2\succ 1 \succ\others]$. Let $\vec x$ denote the vector obtained from $\vec y$ by replacing two votes of $[3\succ 2\succ 1 \succ\others]$ vote by $[2\succ 3\succ 1\succ\others]$. It follows that the $\umg(\vec x)$ is Figure~\ref{fig:cd-even-B1-alpha0} (b). Notice that $\vec y$ is an interior point of $ \sourcepolyz$; $\vec x$ is an interior point of $\targetpolyz$; and $\dim(\sourcepolyz)=\dim(\targetpolyz)=m!$. By Claim~\ref{claim:dim-infty-sufficient}, we have $d_{\infty}=m!$.  The even $n$ and $\alpha =0$ case follows after Theorem~\ref{thm:PMV-instability}.

 \vspace{2mm}\noindent{\bf \boldmath Other coalitional influence problems.} \ 

\vspace{2mm}
\noindent {\bf The  proof for $\mov$} is based on the same constructions of $P_\text{S}$, $P_\text{T}$, $\sourcepoly$,  and $\targetpoly$. The only difference is that the set of vote operations in $\vosetting{}$ is $\voset{\pm}$.

\vspace{2mm}
\noindent {\bf The  proofs for $\eccav$ and $\eccdv$} are is based on similarly constructions of $P_\text{S}$, $P_\text{T}$, $\sourcepoly$,  and $\targetpoly$. The main differences are, first, the set of vote operations in $\vosetting{}$ is $\voset{+}$ and $\voset{-}$ for $\eccav$ and $\eccdv$, respectively. Second, the added (respectively, deleted) votes correspond to the new (respectively, old) votes in $\cm$. We add $3\times \ml(\ma)$ to $P_\text{S}$ so that there is enough votes to be deleted for $\eccdv$. Below we take constructive control $\{d\}=r(P_\text{T})$ for example (where $P_\text{T}$ depends on the problem and will be specified soon). Other cases can be proved similarly.
\begin{itemize}
\item {\bf Integer positional scoring rules.} If $d\ne 1$, then $\sourceprofile$,   $\sourcepoly$,  and $\targetpoly$ are similar to their counterparts in the proof of $\cm$ under integer positional scoring rules. Take $d=2$ for example, for $\eccav$, the added votes are $R_2$; and for $\eccdv$, the deleted votes are $R_1$. If $d = 1$, then we switch the roles of $\sourcepoly$,  and $\targetpoly$, and switch the roles of $\sourceprofile$ and $\targetprofile$ in the proof of $\cm$ under integer positional scoring rules. Then, for  $\eccav$, the added votes are $R_1$; and for $\eccdv$, the deleted votes are $R_2$.
\item  {\bf STV.} If $d\ne 1$, then $\sourceprofile$,   $\sourcepoly$,  and $\targetpoly$ are similar to their counterparts in the proof of    $\cm$ under $\stv$. Take $d=2$ for example, for $\eccav$, for $\eccav$, the added votes are $[2\succ 3\succ 1\succ\others]$; and for $\eccdv$, the deleted votes are $[4\succ 1\succ 2\succ 3\succ \others]$. If $d = 1$, then we switch the roles of $\sourcepoly$,  and $\targetpoly$, and switch the roles of $\sourceprofile$ and $\targetprofile$ in the proof of $\cm$ under integer $\stv$.
\item {\bf Ranked pairs, Schulze, and maximin.}  If $d\ne 1$, then $\sourceprofile$,   $\sourcepoly$,  and $\targetpoly$ are similar to their counterparts for $\cm$. Take $d=2$ for example, for $\eccav$, the added votes are $[2\succ 3\succ 1\succ\others]$; for $\eccdv$, the deleted votes are $[ 1\succ 3\succ 2\succ \others]$. When $n$ is even, the weights in $\wmg(\sourceprofile)$ are all even, for example the positive weights can be one plus the weights in Figure~\ref{fig:rp-schulze-B1-odd} (a). If $d = 1$, then we switch the roles of $\sourcepoly$,  and $\targetpoly$, and switch the roles of $\sourceprofile$ and $\targetprofile$ in the proof of their counterparts for $\cm$.

\item {\bf Copeland.}  The proof for Copeland ($\alpha\ne 0$) is slightly more complicated than the proof for other rules, as  $\sourcepoly$, $\targetpoly$, $\sourceprofile$, and $\targetprofile$ depend on the parity of $n$.  We prove $\eccav$ with $d=2$ (changed from a profile $\sourcepoly$ where $1$ is the winner) for odd and even $n$ respectively, then comment on how to modify it for other cases.

{\bf\boldmath When $n$ is odd}, let $\sourcepoly$ and $\targetpoly$ be the same as those for $\cm$, even $n$, $\alpha >0$ (Figure~\ref{fig:cd-even-B1-alpha}  (a) and (b)). Let $\sourceprofile$ be any $n$-profile whose WMG is as shown in Figure~\ref{fig:cd-control-odd} (a), and let $\targetprofile$ be the $(n+1)$-profile obtained from $\sourceprofile$ by adding one vote of $[2\succ 1\succ3\succ \cdots\succ m]$. The WMG of $\targetprofile$ is shown in Figure~\ref{fig:cd-control-odd} (b).

\begin{figure}[htp]
\centering
\begin{tabular}{ccc}
\includegraphics[width = .3\textwidth]{fig/CDS-odd.pdf} & \includegraphics[width = .3\textwidth]{fig/CDT-even-alpha.pdf}&
\includegraphics[width = .3\textwidth]{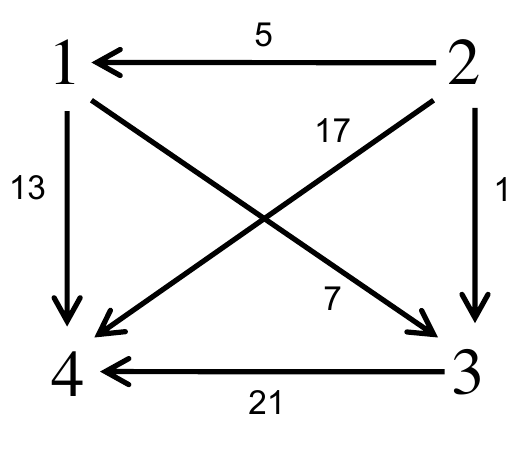}
\\
(a) $\wmg(\sourceprofile)$ & (b) $\wmg(\targetprofile)$ & (c) $\wmg(\vec x)$ 
\end{tabular}
\caption{$\eccav$ under $\copeland$, odd $n$.\label{fig:cd-control-odd}}
\end{figure}
Let $\vosetting = \langle\sourcepoly,\targetpoly,\vosetting_{+},\vec 1\,\rangle$. It follows that for any $\vec y\in \sourcepoly$ and any $\vec x\in \targetpoly$, $\copeland(\vec y) = \{1\}$ ($1$, $2$, and $3$ have the same highest Copeland score, so $1$ wins due to tie-breaking) and $\copeland(\vec x) = \{2\}$ ($2$ has the strictly highest Copeland score $m-2+\alpha$).  It is not hard to verify that $d_0= m!-1$ (the only implicit equality is the tie between $2$ and $3$). Moreover, let $\vec y = \hist(\sourcepoly)$ and let $\vec x = \hist(\sourcepoly \cup 2\times [2\succ 1\succ 3\succ \cdots\succ m])$, whose WMG is illustrated in Figure~\ref{fig:cd-control-odd} (c). Then we have that $\vec y$ and $\vec x$ are the interior points of $\sourcepolyz$ and $\targetpolyz$, respectively, and it follows from Claim~\ref{claim:dim-infty-sufficient} that $d_\infty  = m!$. The case of $d=2$, odd $n$  follows after Theorem~\ref{thm:PMV-instability}.

{\bf\boldmath When $n$ is even}, let $\sourcepoly$ and $\targetpoly$ be the same as those for  $\cm$, even $n$, $\alpha = 0$ (Figure~\ref{fig:cd-even-B1-alpha0}  (a) and (b)). Let $\sourceprofile$ be any $n$-profile whose WMG is as shown in Figure~\ref{fig:cd-control-even} (a) (which is the same as \ref{fig:cd-even-B1-alpha0} (c)), and let $\targetprofile$ be the $(n+1)$-profile obtained from $\sourceprofile$ by adding one vote of $[2\succ 1\succ3\succ \cdots\succ m]$. The WMG of $\targetprofile$ is shown in Figure~\ref{fig:cd-control-even} (b) (which is the same as Figure~\ref{fig:cd-control-even} (c)).

\begin{figure}[htp]
\centering
\begin{tabular}{ccc}
\includegraphics[width = .3\textwidth]{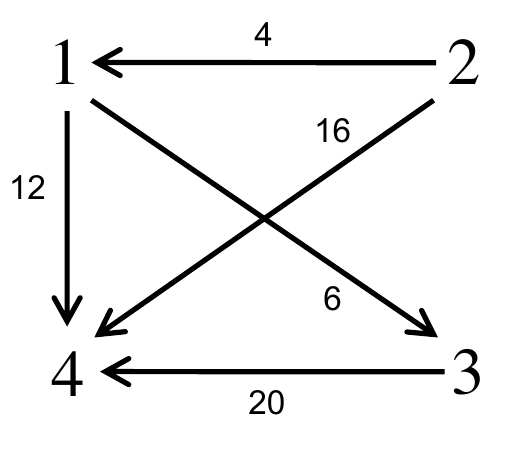} & 
\includegraphics[width = .3\textwidth]{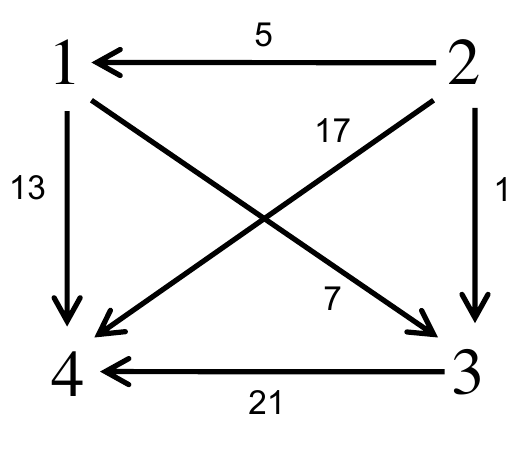}&
\includegraphics[width = .3\textwidth]{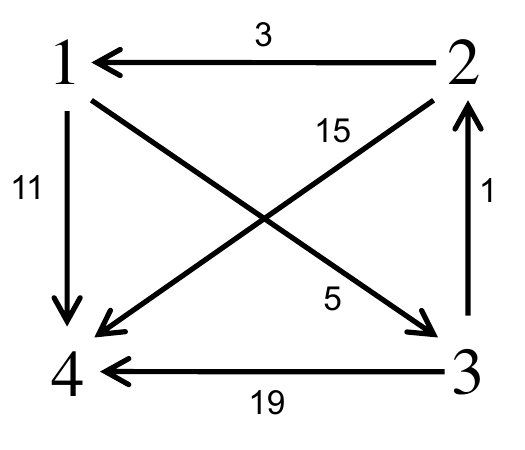}
\\
(a) $\wmg(\sourceprofile)$ & (b) $\wmg(\targetprofile)$ & (c) $\wmg(\vec y)$ 
\end{tabular}
\caption{$\eccav$ under $\copeland$, even $n$.\label{fig:cd-control-even}}
\end{figure}
Let $\vosetting = \langle\sourcepoly,\targetpoly,\vosetting_{+},\vec 1\,\rangle$. It follows that for any $\vec y\in \sourcepoly$ and any $\vec x\in \targetpoly$, $\copeland(\vec y) = \{1\}$ ($1$, $2$, and $3$ have the same highest Copeland score, so $1$ wins due to tie-breaking) and $\copeland(\vec x) = \{2\}$ ($2$ has the strictly highest Copeland score $m-2+\alpha$).  It is not hard to verify that $d_0= m!-1$ (the only implicit equality is the tie between $2$ and $3$). Moreover, let $\vec y $ denote the histogram of $\sourceprofile$ subtracting one vote of $[2\succ 1\succ 3\succ\cdots\succ m]$ (whose WMG is illustrated in Figure~\ref{fig:cd-control-even} (c), which is the same as Figure~\ref{fig:cd-control-odd} (a)), and let $\vec x = \hist(\targetpoly)$. Then we have that $\vec y$ and $\vec x$ are the interior points of $\sourcepolyz$ and $\targetpolyz$, respectively, and it follows from Claim~\ref{claim:dim-infty-sufficient} that $d_\infty  = m!$. The case of $d=2$, odd $n$  follows after Theorem~\ref{thm:PMV-instability}.

The proof for {\bf \boldmath $\eccav$ and any  $d\ge 3$} is similar, which is done by simply switching the role $2$ and $d$  in the proof for $d=2$. The proof for {\bf \boldmath $\eccav$ and $d=1$} is done by (1) switching the role of $\sourcepoly$ and $\targetpoly$, (2) switching the role of $\sourceprofile$ and $\targetprofile$, and (3) the added vote is the inverse of $[2\succ 1\succ 3\succ\cdots\succ m]$.  The proof for {\bf \boldmath $\eccdv$} is similar, by noticing that adding $[2\succ 1\succ 3\succ\cdots\succ m]$ is equivalent to subtracting its inverse. 
\end{itemize}

\vspace{2mm}
\noindent {\bf The  proofs for $\edcav$ and $\edcdv$} are similar to the proofs for $\eccav$ and $\eccdv$, as the proof essentially works for control of changing any source winner to any target winner.  
\end{proof}

\subsection{Full Version of Theorem~\ref{thm:PMV-instability-GSR-upper} and Its Proof}
\label{app:proof-Lemma-PMV-upper}
\appThm{\bf Upper bound on Coalitional Influence Problems under GSRs}{thm:PMV-instability-GSR-upper}{
Let $r$ denote any GSR with fixed $m\ge 3$. For any closed and strictly positive $\Pi$, any $X\in\{ \cm, \mov\}\cup\control\cup \econtrol$, any $n$, and any $B\ge 0$,
$$\satmax{X}{\Pi,0}(r,n,B) =  O\left(\min\left\{\frac{B+1}{\sqrt n},1\right\}\right) $$
}
\begin{proof} We prove the theorem as a    corollary of Theorem~\ref{thm:GSR} in Appendix~\ref{sec:semi-random-CI}.  Let $X$ be any coalitional influence problem described in the lemma and let its PMV multi-instability representation be $\calM = \{\vosetting^i = \langle\sourcepoly^i, \targetpoly^i, \voset{}^i,\vec c^{\,i}\,\rangle:i\le I\}$ due to Lemma~\ref{lem:CI-GSR}. For every $i\le I$, recall that there exist a pair of feasible signatures $\vec t_1$ and $\vec t_2$ so that $\cor(\vec t_1)\ne \cor(\vec t_2)$ and $\sourcepoly^i =\cpoly{\vec t_1}$ and $\targetpoly^i =\cpoly{\vec t_2}$. Therefore, at least one component of $\vec t_1\oplus\vec t_2$ is zero. This means that $d_0^{i} = \dim(\cpoly{\vec t_1\oplus\vec t_2})\le m!-1$.  Also notice that $d_\infty^i\le m!$, which means that $d_\Delta\le m! - d_0^i$. Therefore, according to Theorem~\ref{thm:PMV-instability},  we have 
\begin{align*}
&\sup\nolimits_{\vec\pi\in \Pi^n}Pr_{P\sim\vec \pi}(\hist(P)\in \csus{n,B}^i) = O\left(\min\{B+1,\sqrt n\}^{d^i_{\Delta}}\cdot (\frac{1}{\sqrt n})^{m!-d^i_0}  \right)\\ 
= & O\left(\left(\min\left\{\frac{B+1}{\sqrt n},1\right\}\right)^{m!-d_0^i}\right)= O\left(\min\left\{\frac{B+1}{\sqrt n},1\right\}\right)
\end{align*}
This proves  Theorem~\ref{thm:PMV-instability-GSR-upper} because $I$ is finite.
\end{proof}

\subsection{Proof of Theorem~\ref{thm:CML}}
\label{app:proof-thm-CML}
\appThm{Max-Semi-Random Coalitional Manipulation for The Loser}{thm:CML}{
Let $r_{\vec s}$  be an integer positional scoring rule with lexicographic tie-breaking for fixed $m\ge 3$ that is different from veto. For  any closed and strictly positive $\Pi$ with $\piuni\in\conv(\Pi)$, there exist $ N>0$ and $B^*>0$ such that for any $n>N$ and any $B\ge B^*$,
$$\satmax{\cml}{\Pi,0}(r_{\vec s} ,n,B) =  \Theta\left(\min\left\{\frac{B}{\sqrt n},1\right\}^{m-1}\right)  $$
Moreover, for any $\psi>0$, $\satmax{\cml}{\Pi,\psi}(r ,n,B) =\Theta(1)$
}  

\begin{proof} The proof proceeds in the following two steps. 

\myparagraph{Define the multi-instability setting $\calM_{\cml}$.} For every pair of different alternatives $a,b$, we define an instability setting $\vosetting_{\cml}^{a\ra b} = (\sourcepoly, \targetpoly, \voset{\pm}^{a\ra b}, \vec 1)$, where 
\begin{itemize}
\item $\sourcepoly$ denote the set of vectors where $a$ is the winner and $b$ is the loser under $r_{\vec s}$.
\item $\targetpoly$ denote the set of vectors where $b$ is the winner  under $r_{\vec s}$.
\end{itemize}
Take $a = 1$ and $b =2 $ for example. Recall that $\score_{a,b}^{\vec s}$ is the score vector of $b$ (under all linear orders) minus the score vector of $a$ (Definition~\ref{dfn:varcons}). Then, we have:
$$\sourcepoly\triangleq\left\{\vec x: \begin{split}\forall i\ge 2, \   \score_{i,1}^{\vec s}\cdot\vec x&\le 0\\ \forall i\ge 3, \  \score_{2,i}^{\vec s}\cdot\vec x&\le -1\\   -\vec x&\le  \vec 0\end{split}\right\}\text{, }\targetpoly\triangleq \left\{\vec x:  \begin{split}  \score_{1,2}^{\vec s}\cdot\vec x&\le -1\\  \forall i\ge 3, \ \score_{i,2}^{\vec s}\cdot\vec x&\le 0\\   -\vec x&\le  \vec 0\end{split}\right\}\text{, and}$$ 
$$\calM_{\cml} = \left\{\vosetting_{\cml}^{a\ra b}: a,b\in\ma, a\ne b\right\}$$
\noindent{\bf \boldmath Apply Theorem~\ref{thm:PMV-instability}.} In this step, we prove that for every $\vosetting_{\cml}^{a\ra b}$ (with corresponding $\csus{n,B}^{a\ra b}$), 
\begin{equation}
\label{equ:CML}
\sup_{\vec\pi\in\Pi^n}\Pr\left(\vXp \in \csus{n,B}\right)= \Theta\left(\min\left\{\frac{B}{\sqrt n},1\right\}^{m-1}\right)
\end{equation}
We first prove $\condition{1} = \text{false}$ (i.e.,  $\csus{n,B}^{a\ra b}\ne\emptyset$) for any sufficiently large $n$ and $B$, by constructing a successful manipulation by $B$ voters.  Recall from the proof of Theorem~\ref{thm:PMV-instability-applications} that for any $a\in \ma$,  $\sigma_a$ denotes a cyclic permutation among $\ma\setminus\{a\}$. Then, we define two profiles, each of which consists of $m-1$ votes as follows.
$$P_{a}^\text{top}  \triangleq \left\{\sigma_a^i(a\succ \others): 1\le i\le m-1 \right\}\text{ and } P_{a}^\text{bot} \triangleq \left\{\sigma_a^i( \others\succ a): 1\le i\le m-1 \right\}$$ 
We further define the following ``cyclic'' profile of $m$ votes: let $\sigma$ denote any cyclic permutation among $\ma$, e.g., $1\ra 2\ra\cdots\ra m\ra 1$. 
$$P_\text{cyc}  \triangleq \left\{\sigma^i(1\succ \cdots\succ m): 1\le i\le m \right\}$$
Let $\sourceprofile$ denote the $n$-profile that consists of 
$P_a^\text{top}$, $P_a^\text{bot}$,  as many copies of $P_\text{cyc}$ as possible, and the remaining rankings are $[a\succ \others\succ b]$. Let $\targetprofile$ denote the profile obtained from $\sourceprofile$ by replacing $\left\lceil\frac{(m-1)(s_1-s_m)}{(m-2)(s_1-s_{m-1})}\right\rceil +1$ copies of $[\others\succ b\succ a]$  to $[b\succ \others\succ a]$. $\targetprofile$  is well-defined for any sufficiently large $n$. It is not hard to verify that $a$ has the strictly highest score in $\sourceprofile$, $b$ has the strictly lowest score in $\sourceprofile$, and $b$ has the strictly highest score in $\targetprofile$.  This means that for any sufficient large $B$, $\csus{n,B}^{a\ra b}\ne\emptyset$.

Next, let $\sus{0}^{a\ra b}= \sourcepolyz\cap \targetpolyz$. We have $\piuni\in \conv(\Pi)\cap \sus{0}$, which means that $\condition{3} = \text{false}$. Therefore, the polynomial case of Theorem~\ref{thm:PMV-instability} holds for $\vosetting_{\cml}^{a\ra b}$. Notice that
$$\sus{0}^{a\ra b}= \sourcepolyz\cap \targetpolyz = \left\{\vec x: \begin{split}\forall i\ne b, \   \score_{i,b}^{\vec s}\cdot\vec x&\le 0\\ 
\forall i\ne b, \   \score_{b,i}^{\vec s}\cdot\vec x&\le 0\\   -\vec x&\le  \vec 0\end{split}\right\}$$
Therefore, we have $d_0 = \dim(\sus{0}^{a\ra b}) = m!-(m-1)$ (because the implicit equalities represent the scores of all alternatives are the same, which are characterized by $m-1$ equations). Notice that $\hist(\sourceprofile)$ is an interior point of $\sourcepolyz$, $\hist(\targetprofile)$ is an interior point of $\targetpoly$, and $\dim(\sourcepolyz) = \dim(\targetpolyz) = m!$. Therefore, by Claim~\ref{claim:dim-infty-sufficient}, we have $d_{\infty} = \dim(\sourcepolyz) = m!$, which means that $d_{\Delta} =  m-1 $. It follows from Theorem~\ref{thm:PMV-instability} that
$$\sup_{\vec\pi\in\Pi^n}\Pr\left(\vXp \in \csus{n,B}\right)= \Theta\left(\dfrac{\min\{B+1,\sqrt n\}^{m-1}}{(\sqrt n)^{m!-(m!-(m-1))}} \right) = \Theta\left(\min\left\{\frac{B}{\sqrt n},1\right\}^{m-1}\right),$$
which proves Equation~\eqref{equ:CML}. Then, Theorem~\ref{thm:CML} follows after applying Equation~\eqref{equ:CML} to all $a\ne b$.
\end{proof}


\section{A General Theorem on Semi-Random Coalitional Influence}
\label{sec:semi-random-CI}
\subsection{Constructive/Destructive Generalized Bribery with Anonymous Prices}
In this section, we  define two large classes of bribery problems that include some commonly-studied control problems as special cases.
\begin{dfn}
A {\em constructive generalized bribery with anonymous prices} problem is denoted by $\cb_{d,\vec c}(r,P,B)$, where  $r$ is a voting rule, $P$ is  a profile, $a$ is a distinguished alternative, $\vec c>\vec 0$ is a strictly positive cost vector, where each component is indexed by a pair $(R,R')\in (\ml(\ma)\cup \{\emptyset\})\times (\ml(\ma)\cup \{\emptyset\})$ that represents the price for the briber to convert an $R$ vote to an $R'$ vote, and $B\ge 0$ is the total budget.  We are asked whether the briber can make $a$ win by changing the votes in the profile under the budget constraint $B$---if so then we let $\cb_{d,\vec c}(r,P,B)=1$, otherwise we let $\cb_{d,\vec c}(r,P,B)=0$.

 {\em Destructive bribery with anonymous price} problem, denoted by $\db_{d,\vec c}(r,P,B)$, is defined similarly, and the only difference is that the goal of the briber is to make $a$ {\em not} the winner. 
\end{dfn}
Specifically, when $R=\emptyset$, performing an $(R, R')$ bribery is effectively the same as adding an $R'$ vote to $P$; and if $R'=\emptyset$,  performing an $(R, R')$ bribery is effectively the same as removing an $R$ vote from $P$. Moreover, we allow the price of an $(R, R')$ operation to be $\infty$, which means that this operation is not allowed in the problem.

For each constructive/destructive control/bribery problem, we also study its ``effective'' variant, which requires that the influencers' goal is not achieved in the original profile. We will add ``{\sc e}-'' to the name to denote this variant. For example, $\ecb_{d,\vec c}(r,P,B) = 1$ if $r(P)\ne \{a\}$  and $a$ can be made the winner under budget $B$.

\begin{prop}$\ccav_{d}$ and $\ccdv_{d}$ are special cases of $\cb_{d,\vec c}$. $\eccav_{d}$ and $\eccdv_{d}$ are special cases of $\ecb_{d,\vec c}$.  $\dcav_{d}$ and $\dcdv_{d}$ are special cases of $\db_{d,\vec c}$. $\edcav_{d}$ and $\edcdv_{d}$ are special cases of $\edb_{d,\vec c}$. 
\end{prop}
\begin{proof} It is not hard to verify that $\ccav_{d}$ (respectively, $\dcav_{d}$) is equivalent to $\cb_{d,\vec c}$ (respectively, $\db_{d,\vec c}$), where for any $(R,R')\in (\ml(\ma)\cup \{\emptyset\})\times (\ml(\ma)\cup \{\emptyset\})$, the $(R,R')$ component of $\vec c$, denoted by $[\vec c\,]_{(R,R')} $, is $\begin{cases}1 &\text{if }R=\emptyset\\ 0&\text{otherwise}\end{cases}$. 

$\ccdv_{d}$ (respectively, $\dcdv_{d}$) is equivalent to $\cb_{d,\vec c}$ (respectively, $\db_{d,\vec c}$), where for any $(R,R')\in (\ml(\ma)\cup \{\emptyset\})\times (\ml(\ma)\cup \{\emptyset\})$,  
$$[\vec c\,]_{(R,R')}=\begin{cases}1 &\text{if }R'=\emptyset\\ 0&\text{otherwise}\end{cases}$$

The proofs for  {\sc e}-variants are similar.
\end{proof}

\subsection{Theorem~\ref{thm:GSR}: Semi-Random Coalitional Influence for GSRs}

\begin{thm}
\label{thm:GSR}
Let $r$ denote any int-GSR with fixed $m\ge 3$.  For any closed and strictly positive $\Pi$ and any  $X\in\{\cm,\mov\}$, there exists a constant $C_1>0$, such that for any  $n\in \mathbb N$ and any $B\ge 0$ with $B\le C_1n$,   there exist $\{d_0^{\max},d_\Delta^{\max},d_0^{\min},d_\Delta^{\min}\}\subseteq [m!]$ such that
\begin{center}
$\satmax{X}{\Pi,0}(r,n,B)$ is $0$, $\exp(-\Theta(n))$, or $\Theta\left(\dfrac{\min\{B+1,\sqrt n\}^{d_\Delta^{\max}}} {( \sqrt n )^{m!-d_0^{\max}}}\right)$, and \\
$\satmin{X}{\Pi,0}(r,n,B)$ is $0$, $\exp(-\Theta(n))$, or $\Theta\left(\dfrac{\min\{B+1,\sqrt n\}^{d_\Delta^{\min}}} {( \sqrt n )^{m!-d_0^{\min}}}\right)$ 
\end{center}
\end{thm}
As explained in the Introduction, the main merit of Theorem~\ref{thm:GSR} is conceptual, as it illustrates a separation between $0$, exponential, and polynomial cases of different degrees. The most interesting part is the  asymptotically tight  polynomial lower bounds, because little was known about them even under $\cm$, IC, and $B=1$, as discussed in Section~\ref{sec:related-work}.  The theorem  works for many other commonly-studied coalitional influence problem such as those defined in Appendix~\ref{app:more-CI}.

\vspace{2mm}\noindent{\bf Constants in the theorem.} $C_1$ and the constants in the asymptotic notation $\Theta(\cdot)$ depend on  $r$,   $X$, and $\Pi$. In other words, they do not depend on $n$ or $B$, which means that  that there exists $\beta>1$ that does not depend on $n$ or $B$, such that for every  $n$ and every $B\ge 0$,
$$\hfill
\satmax{X}{\Pi,0}(r,n,B)\in \{0\}   \cup  \left[\exp(-\beta \cdot n), \exp\left(-\frac1\beta \cdot n\right)\right] \cup  \left ( \left [  \frac 1\beta,  \beta \right]\cdot  \dfrac{\min\{B+1,\sqrt n\}^{d_\Delta^{\max}}} {( \sqrt n )^{m!-d_0^{\max}}}\right)
\hfill$$

A similar argument holds for $\satmin{X}{\Pi,0}(r,n,B)$. We note that $\{d_0^{\max},d_\Delta^{\max},d_0^{\min},d_\Delta^{\min}\}$ may depend on $n$ and $B$, yet they are integers that are no more than $m!$ as stated in the theorem.

\subsection{Proof of  Theorem~\ref{thm:GSR}}


\vspace{2mm}\noindent{\bf Proof overview.} 
Recall that the full version of Lemma~\ref{lem:CI-GSR} in Appendix~\ref{app:proof-prop-CI-GSR} states that the  coalitional influence problems can be represented as   multi-instability problems (Definition~\ref{dfn:multi-instability-problem}). Next, we extend Theorem~\ref{thm:PMV-instability} to solve the multi-instability problem in Theorem~\ref{thm:multi-instability}, and then apply it to the multi-instability problems defined in the proof of Lemma~\ref{lem:CI-GSR} to prove  Theorem~\ref{thm:GSR}. 

For every $i\le I$, we use superscript $i$ to denote the notation defined for $\vosetting{}^i$. For example,  $d_0^i$, and $d_\Delta^i$ denote  $d_0$  and $d_\Delta$ for $\vosetting^i$. 
To present the result, it is convenient to define the following graph.
\begin{dfn}[\bf Activation graph for multi-instability]
\label{dfn:activation-graph} 
Given a PMV multi-instability setting  $\calM=\{\vosetting{}^i:i\le I\}$,  $n$, and $B$, we define a weighted undirected bipartite graph, called {\em activation graph} and is denoted by $\calA_{n,B}$, as follows.
\begin{itemize}
\item {\bf Vertices.} There are two sides of  vertices: $\conv(\Pi)$ and $\{\vosetting^1,\ldots,\vosetting^I\}$.
\item {\bf Edges and weights.} For any $\pi\in\conv(\Pi)$ and any $\vosetting^i\in \calM$, the weight on the edge between $\pi$ and $\vosetting^i$, denoted by $w_{n,B}(\pi,\vosetting^i)$, is defined as follows: for every instability settings $\vosetting$, define
\begin{equation}
\label{eq:weight-B-n}
w_{n,B}(\pi,\vosetting)  \triangleq \left\{\begin{array}{ll}
-\infty & \text{if }\csus{n,B}=\emptyset\\
- \frac{2n}{\log n } & \text{if }\csus{n,B}\ne\emptyset\text{ and } \pi\notin\sus{0}\\
d_0+d_\Delta\cdot\min\left\{\frac{2\log (B+1)}{\log n},1\right\}& \text{otherwise}
\end{array}\right.
\end{equation}
\end{itemize}
\end{dfn}

Notice that while the conditions in (\ref{eq:weight-B-n}) depend on $\pi$, the values of $w_{n,B}(\pi,\vosetting )$   do not depend on $\pi$, and they  are chosen so that $(\sqrt n)^{w_{n,B}}$ corresponds to the values in the exponential cases and polynomial cases of Theorem~\ref{thm:PMV-instability}. Specifically, the $- \frac{2n}{\log n }$ value in the second case is chosen so that $(\sqrt n)^{- \frac{2n}{\log n }} = \exp(-n)$, which corresponds to the exponential cases. 

Given the PMV multi-instability setting  $\calM=\{\vosetting{}^i:i\le I\}$, $B$, and $n$, let  $w_{\max}$ denote the   maximum weigh on edges in $\calA_{n,B}$ and let $(\pi_{\max},i_{\max})$ denote an arbitrary edge with the max weight. That is, 
$$w_{\max}\triangleq \max\nolimits_{\pi\in\conv(\Pi), i\le I}\left\{w_{n,B}(\pi,\vosetting^i)\right\} \text{ and }(\pi_{\max},i_{\max}) \triangleq \arg\max\nolimits_{\pi\in\conv(\Pi), i\le I} w_{n,B}(\pi,\vosetting^i)$$
Let $w_{\min}$ denote the weight of the minimax weighted edge in $\calA_{n,B}$, denoted by $(\pi_{\min},i_{\min})$, where  $\min$ is taken over all $\pi\in\conv(\Pi)$ and $\max$ is taken over all edges connected to  $\pi$.  That is,

$\hfill w_{\min}\triangleq \min\nolimits_{\pi\in\conv(\Pi)}\max\nolimits_{i\le I}\left\{w_{n,B}(\pi,\vosetting^i)\right\},  \pi_{\min} \triangleq \arg\min\nolimits_{\pi\in\conv(\Pi)} \max\nolimits_{  i\le I} w_{n,B}(\pi,\vosetting^i),\hfill $

$\hfill \text{ and }i_{\min}\triangleq \arg\max\nolimits_{  i\le I} w_{n,B}(\pi_{\min},\vosetting^i)\hfill $

Notice that $i_{\max}$ and $i_{\min}$ are both in $[q]$ and $w_{\max}$, $w_{\min}$, $i_{\max}$, and $i_{\min}$  depend on $B$ and $n$, which are clear from the context. 

\begin{thm}[\bf \boldmath Semi-Random multi-instability, $B=O(n)$]
\label{thm:multi-instability} Given any $q\in\mathbb N$, any closed and strictly positive $\Pi$ over $[q]$, and any PMV multi-instability setting $\calM=\{\vosetting^i: i\le I\}$, there exists a constant $C_1>0$ so that for any  $n\in \mathbb N$ and any $B\ge 0$ with $B\le C_1n$, 
\begin{align*}
&\sup_{\vec\pi\in\Pi^n}\Pr\left(\vXp \in \csus{n,B}^\calM\right)=\left\{\begin{array}{ll}0 &\text{if } w_{\max} = -\infty\\
\exp(-\Theta(n)) &\text{if }   w_{\max} = -\frac{2n}{\log n}\\
\Theta\left( \left(\frac{1}{\sqrt n}\right)^{q-w_{\max}}\right) &\text{otherwise} 
\end{array}\right. \\
&\inf_{\vec\pi\in\Pi^n}\Pr\left(\vXp \in \csus{n,B}^\calM\right)=\left\{\begin{array}{ll}0 &\text{if } w_{\min} = -\infty\\
\exp(-\Theta(n)) &\text{if }  w_{\min} = -\frac{2n}{\log n}\\
\Theta\left( \left(\frac{1}{\sqrt n}\right)^{q-w_{\min}}\right) &\text{otherwise} 
\end{array}\right. 
\end{align*}
\end{thm}
\begin{proof} The $0$ cases of $\sup$ and $\inf$ are straightforward. Like the proof of Theorem~\ref{thm:PMV-instability}, it suffices to prove the other cases of $\sup$ and $\inf$ for every sufficiently large $n$.

\vspace{2mm}\noindent{\bf\bf \bf\boldmath Proof for the $\sup$ part.} We first prove a convenient corollary of Theorem~\ref{thm:PMV-instability}, which uses the weight in the activation graph to represent conditions of the $\sup$ part of Theorem~\ref{thm:PMV-instability}.

\begin{coro}[\bf Activation graph representation of the $\sup$ part of Theorem~\ref{thm:PMV-instability}]
\label{coro:PMV-instability}
Given any $q\in\mathbb N$, any closed and strictly positive $\Pi$ over $[q]$,  and any instability settings $\vosetting = \langle\sourcepoly,\targetpoly,\voset{},\vec c\,\rangle$,  any $C_2$ with $C_2<B_{\conv(\Pi)}$, any $n\in\mathbb N$, and any $0\le B\le C_2 n$, let $w^* = \sup_{\pi\in\conv(\Pi)}w_{n,B}(\pi,\vosetting)$,
\begin{align*}
&\sup_{\vec\pi\in\Pi^n}\Pr\left(\vXp \in \csus{n,B}\right)=\left\{\begin{array}{ll}0 &\text{if } w^* = -\infty\\
\exp(-\Theta(n)) &\text{if } -\infty< w^* <0\\
\Theta\left( \left(\frac{1}{\sqrt n}\right)^{q-w^*}\right) &\text{otherwise}
\end{array}\right. 
\end{align*}
\end{coro}
\begin{proof}
The $0$ case is straightforward. Like in the proof of Theorem~\ref{thm:PMV-instability}, for the non-zero cases it is without loss of generality to assume that   $n$ is larger than a constant.

\vspace{2mm}
\noindent{\em Exponential case.} In this case we have $w^* = -\frac{2n}{\log n}$, which means that   $\conv(\Pi)\cap \sus{0}=\emptyset$. Therefore, the exponential case of Corollary~\ref{coro:PMV-instability} follows after the exponential case and the $B\le C_2n$ case of in Theorem~\ref{thm:PMV-instability}.   

\vspace{2mm}
\noindent{\em Polynomial case.} Because $\conv(\Pi)$ is bounded and closed, it is compact. Therefore, there exists  $\pi^*\in \conv(\Pi)$ such that $w_{n,B}(\pi^*,\vosetting) = w^*=d_0+d_\Delta\cdot\min\left\{\frac{2\log (B+1)}{\log n},1\right\}$. It follows from Theorem~\ref{thm:PMV-instability} that
 
$$\sup_{\vec\pi\in\Pi^n}\Pr\left(\vXp \in \csus{n,B}\right)=\Theta\left(\dfrac{\min\{B+1,\sqrt n\}^{d_{\Delta}}}{(\sqrt n)^{q-d_0}}\right) $$
Notice that 
\begin{align*}
&\log\left(\dfrac{\min\{B+1,\sqrt n\}^{d_{\Delta}}}{(\sqrt n)^{q-d_0}}\right) = d_\Delta\min\left\{\log(B+1),\frac{\log n}{2}\right\} + (d_0-q)\frac{\log n}{2}\\
=& \frac{\log n}{2}\cdot \left( d_0+d_\Delta\cdot\min\left\{\frac{2\log (B+1)}{\log n},1\right\}-q\right) = \log\left(\left(\frac{1}{\sqrt n}\right)^{q-w^*}\right)
\end{align*}
This completes the proof of Corollary~\ref{coro:PMV-instability}.
 \end{proof}

\vspace{2mm}
\noindent{\bf \boldmath Define $C_1$ for $\sup$.} Let $C_1>0$ denote  any positive number that is smaller than any strictly positive $B_{\conv(\Pi)}^i$. That is, 
$$0<C_1<\min\{B_{\conv(\Pi)}^i: B_{\conv(\Pi)}^i>0, i\le I\}$$
The rest of the proof for the $\sup$ part of Theorem~\ref{thm:multi-instability} is done by combining the results of the applications of Corollary~\ref{coro:PMV-instability} to all instability settings and the following inequality.
\begin{equation}
\label{eq:pmv-union-ineq}
\max_{i\in I}\sup_{\vec\pi\in\Pi^n}\Pr\left(\vXp \in \csus{n,B}^i\right)\le  \sup_{\vec\pi\in\Pi^n}\Pr\left(\vXp \in \csus{n,B}^\calM\right)\le  I\cdot \max_{i\in I}\sup_{\vec\pi\in\Pi^n}\Pr\left(\vXp \in \csus{n,B}^i\right)
\end{equation}

\vspace{2mm}
\noindent{\bf \boldmath Exponential case of $\sup$.} In this case  $w_{\max} = -\frac{2n}{\log n}$. Notice that for all $\pi\in\conv(\Pi)$ and all $i\le I$, we have $w_{n,B} (\pi,\vosetting^i)\le -\frac{2n}{\log n}$, and there exists $\pi^*\in\conv(\Pi)$ and $i^*\le I$ such that  $w_{n,B} (\pi^*,\vosetting^{i^*})= -\frac{2n}{\log n}$. Therefore, by Corollary~\ref{coro:PMV-instability}, $\sup_{\vec\pi\in\Pi^n}\Pr\left(\vXp \in \csus{n,B}^{i^*}\right)= -\frac{2n}{\log n}$, which means that 
$$\max_{i\in I}\sup_{\vec\pi\in\Pi^n}\Pr\left(\vXp \in \csus{n,B}^i\right) = \exp(-\Theta(n))$$
The exponential case follows after (\ref{eq:pmv-union-ineq}).

\vspace{2mm}
\noindent{\bf \boldmath  Polynomial case of $\sup$.} The proof for the polynomial case is similar. To prove the {\bf polynomial upper bound}, notice that for all $\pi\in\conv(\Pi)$ and all $i\le I$, we have $w_{n,B}(\pi,\vosetting^i)\le w_{\max}$. By Corollary~\ref{coro:PMV-instability}, when $n$ is sufficiently large, for all $i\le I$, we have  
$$\max_{i\in I}\sup_{\vec\pi\in\Pi^n}\Pr\left(\vXp \in \csus{n,B}^i\right) = \max_{i\in I} \Theta\left(\left(\frac{1}{\sqrt n}\right)^{q-w^{i*}}\right)= O\left(\left(\frac{1}{\sqrt n}\right)^{q-w_{\max}}\right),$$
where $w^{i*}$ is the weight $w^*$ in Corollary~\ref{coro:PMV-instability} applied to $\vosetting = \vosetting^i$.

To prove the {\bf polynomial lower bound}, let $\pi^*\in\conv(\Pi)$ and $i^*\le I$ be such that  $w_{n,B}^{i^*}(\pi^*,\vosetting^{i^*})= w_{\max}$. According to the polynomial case of 
Corollary~\ref{coro:PMV-instability}, we have 
$$\max_{i\in I}\sup_{\vec\pi\in\Pi^n}\Pr\left(\vXp \in \csus{n,B}^i\right) = \Omega \left(\left(\frac{1}{\sqrt n}\right)^{q-w_{\max}}\right)$$
The polynomial lower bound of $\sup$ follows after (\ref{eq:pmv-union-ineq}). This proves the $\sup$ part of Theorem~\ref{thm:multi-instability}.

\vspace{2mm}\noindent{\bf\bf \bf\boldmath Proof for the $\inf$ part.} The hardness in proving the $\inf$ part is that any $\vec\pi\in\Pi^n$ that achieves $\inf$ for one instability settings $\vosetting^i$ may not achieve $\inf$ for another instability settings $\vosetting^{i'}$, and therefore may not  achieve $\inf$ of the multi-instability problem. This is different from the $\sup$ part, where the $\vec \pi\in\Pi^n$ with the highest value of $\sup$ under some $\vosetting^{i'}$ achieves $\sup$ of the union-manipulation problem. Consequently, even though the $\inf$ counterpart of Corollary~\ref{coro:PMV-instability} can be proved for $\inf$, it is cannot be leveraged to prove the $\inf$ part of Theorem~\ref{thm:multi-instability}.

To prove the $\inf$ part, we first define some distributions that will be used to prove the upper bounds for $\inf$. 
\begin{dfn}[\bf \boldmath $\pi_{\calI}$'s]
Given a multi-instability setting $\calM$, for every non-empty set $\calI\subseteq [I]$ such that $\conv(\Pi)\nsubseteq \bigcup_{i\in\calI}\sus{0}^i$, we choose $\pi_{\calI}\in \conv(\Pi)\setminus\left( \bigcup_{i\in\calI}\sus{0}^i\right) $. 
\end{dfn}

Because every $\sus{0}^i$ is a closed set, the distance between $\pi_{\calI}$ and $\bigcup_{i\in\calI}\sus{0}^i$ is strictly positive and is denoted by $\delta_{\calI}>0$. Due to Claim~\ref{claim:close-to-cone0}, for every $i\in \calI$ there exists a constant  $c_i$ such that each vector in $\cpoly{B}^i$ is no more than $c_i(B+1)$ away from a vector in $\sus{0}^i$. It follows that for any $n>\frac{4c_i}{\delta_{\calI}}$ and any $B\le \frac{\delta_{\calI}}{4c_i}n$, the distance between $\pi_{\calI}$ and $\cpoly{B}^i$ is at least $\frac{\delta_{\calI}}{2}n$. 

\vspace{2mm}\noindent{\bf\bf \bf\boldmath  Define $C_1$ for $\inf$.} Let $C_1$ denote the minimum $\frac{\delta_{\calI}}{4c_i}$ for all well-defined $\delta_{\calI}$ and all $i\in \calI$. Let $\delta$ denote the minimum $\frac{\delta_{\calI}}{2}$ for all well-defined $\delta_{\calI}$.  It follows that for any sufficiently large $n$ (that is larger than all $\frac{4c_i}{\delta_{\calI}}$) and any $B\le C_1n$, the distance between any well-defined $\pi_\calI$ and any $i\in\calI$ is at least $\delta n$.

\vspace{2mm}\noindent{\bf\bf \bf\boldmath  Exponential case of $\inf$.} In this case $w_{\min} =  -\frac{2n}{\log n}$. The {\bf exponential lower bound} trivially holds, because there exists an active $\vosetting^i$. To prove the {\bf exponential upper bound}, let $\pi_{\text{MM}}\in \conv(\Pi)$ be an arbitrary distribution such that for all $i\le I$, $w_{n,B}(\pi_{\text{MM}},\vosetting^i)\le -\frac{2n}{\log n}$. Let $\calI_{\text{MM}}$ denote the indices to the active instability settings (whose $\sus{0}$'s do not contain $\pi_{\text{MM}}$), that is,
$$\calI_{\text{MM}} \triangleq \left\{i\le I: w_{n,B}(\pi_{\text{MM}},\vosetting^i)= -\frac{2n}{\log n} \right\}$$
Because $\pi_{\text{MM}}\in \conv(\Pi)\setminus\left( \bigcup_{i\in\calI}\sus{0}^i\right)$,  we have that $\pi_{\calI_\text{MM}}$ is well-defined. Therefore, for every $i\in\calI_{\text{MM}}$, the distance between $\pi_{\calI_{\text{MM}}}$ and $\sus{0}^i$ is at least $\delta n$.  The exponential upper bound follows after applying Claim~\ref{claim:distance}, Hoeffding's inequality, and the union bound to any $\vec \pi\in\Pi^n$ such that $|\sum_{j=1}^n\pi_j - n\cdot\pi_{\calI_\text{MM}}|_\infty = O(1)$, as done in the $B\le C_2n$ case of the \myhyperlink{PMV-proof-sup-n}{proof for PT-$\Theta(n)$-$\sup$ of Theorem~\ref{thm:PMV-instability}}.  

\vspace{2mm}\noindent{\bf\bf \bf\boldmath  Polynomial case of $\inf$.} To prove the {\bf polynomial lower bound}, notice that for every $\vec\pi\in\Pi^n$, there exists $i\le I$ such that $w_{n,B}(\avg{\vec\pi}, \vosetting^i) = d_{n,B}^i \ge w_{\min}$. It follows from Claim~\ref{claim:PMV-point-lb} (applied to $\vosetting^i$,  and $\vec \pi$) that, for every $B\le \sqrt n$,
$$\Pr(\vXp\in \csus{n,B}^\calM)\ge \Pr(\vXp\in \csus{n,B}^i)  = \Theta\left(\left(\frac{1}{\sqrt n}\right)^{q-d_{n,B}^i}\right) = \Omega\left(\left(\frac{1}{\sqrt n}\right)^{q-w_{\min}}\right)$$
Notice that for every $B> \sqrt n$, the inequality still holds, because we have $\Pr(\vXp\in \csus{n,B}^i)\ge \Pr(\vXp\in \csus{n,\sqrt n}^i)$.

To prove the {\bf polynomial upper bound}, let
$$\pi_{\text{MM}} \triangleq \arg\min\nolimits_{\pi\in\conv(\Pi)}\max\nolimits_{i\le I}\left\{w_{n,B}(\pi,\vosetting^i)\right\}$$
Like in the proof of the exponential upper bound of $\inf$ above,  define
$$\calI_{\text{MM}} \triangleq \left\{i\le I: w_{n,B}(\pi_{\text{MM}},\vosetting^i)<0 \right\}$$
Because $\pi_{\text{MM}}\in \conv(\Pi)\setminus\left( \bigcup_{i\in\calI}\sus{0}^i\right)$,  we have that $\pi_{\calI_\text{MM}}$ is well-defined. Therefore, for every $i\in\calI_{\text{MM}}$, the distance between $\pi_{\calI_{\text{MM}}}$ and $\sus{0}^i$ is at least $\delta n$. Choose any $\vec \pi\in\Pi^n$ such that $|\sum_{j=1}^n\pi_j - n\cdot\pi_{\calI_\text{MM}}|_\infty = O(1)$. Like the exponential case above, for every $i\in\calI_{\text{MM}}$, we have
$\Pr(\vXp\in \csus{n,B}^i) \le  \exp(-\Omega(n))$. Moreover, for every $i\notin \calI_{\text{MM}}$ such that $\vosetting^i$ is active, we have $\pi_{\text{MM}}\in \sus{0}^i$, which means that $\sus{0}^i\cap\conv(\Pi)\ne\emptyset$. Recall that $d_{n,B}^i\le w_{\min}$. By Corollary~\ref{coro:PMV-instability}, we have
$$\Pr(\vXp\in \csus{n,B}^i) = O\left( \left(\frac{1}{\sqrt n}\right)^{q-d_{n,B}^i}\right) \le  O\left( \left(\frac{1}{\sqrt n}\right)^{q-w_{\min}}\right)$$ 
Therefore,
$$\Pr(\vXp\in \csus{n,B}^\calM)\le \sum\nolimits_{i\le I}\Pr(\vXp\in \csus{n,B}^i)  =  O\left(\left(\frac{1}{\sqrt n}\right)^{q-w_{\min}}\right),$$
which proves the polynomial upper bound of $\inf$. 
\end{proof}

Notice that in Theorem~\ref{thm:multi-instability}, when $w_{\max}>0$ and $w_{\min}>0$, we have 
\begin{equation}
\label{eq:i-max}
\left(\frac{1}{\sqrt n}\right)^{q-w_{\max}} = \dfrac{\min\{B+1,\sqrt n\}^{d_{\Delta}^{i_{\max}}}}{ ( \sqrt n)^{q-d_0^{i_{\max}}}} \text{ and }
\left(\frac{1}{\sqrt n}\right)^{q-w_{\min}} = \dfrac{\min\{B+1,\sqrt n\}^{d_{\Delta}^{i_{\min}}}}{ ( \sqrt n)^{q-d_0^{i_{\min}}}}
\end{equation}
Recall from Lemma~\ref{lem:CI-GSR} that  any coalitional influence problem $X\in \{\cm,\mov\}$ under any GSR $r$ can be represented by a multi-instability problem. Therefore, Theorem~\ref{thm:GSR} follows immediately after Theorem~\ref{thm:multi-instability} and \eqref{eq:i-max}. This completes the proof of Theorem~\ref{thm:GSR}.

\section{ Ties $\not\Leftrightarrow \left[\Theta(1)\text{ instability}\right]$}
\label{app:tie-instability}
\begin{ex}[\bf \boldmath Ties$\not\Rightarrow  \Theta(1)\text{ instability} $]
\label{ex:ties-stable}
Consider  $\copeland$ with four alternatives. Let $P'$ denote an arbitrary profile whose UMG is the same as Figure~\ref{fig:exstable} (a). For any $n'\in \mathbb N$, we let $P= n' P'$. It is not hard to verify that $\copeland(P) = \{1,2\}$. The winner under $P$ is stable with $\Theta(n)$ changes in votes,  because the UMG of any profile whose histogram is $\Theta(1)$ away from $\hist(P)$ is the same as Figure~\ref{fig:exstable} (a). 
 \end{ex}

\begin{figure}[htp]
\centering
\begin{tabular}{cc}
\includegraphics[width = .3\textwidth]{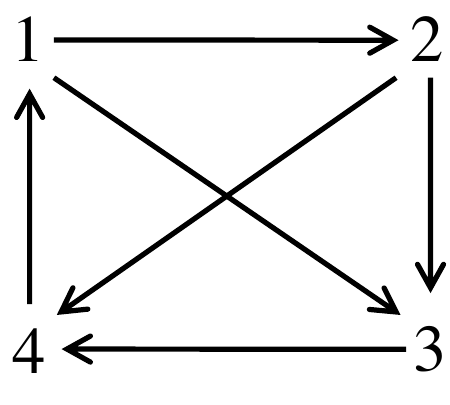}&\includegraphics[width = .3\textwidth]{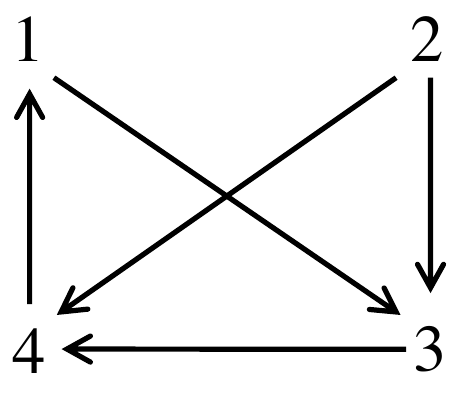}\\
(a) UMG for a tied stable profile.  & (b) UMG for a un-tied unstable profile.
\end{tabular}
\caption{Graphs used in Example~\ref{ex:ties-stable} and~\ref{ex:noties-unstable}. \label{fig:exstable}}
\end{figure}

\begin{ex}[\bf \boldmath  Ties$\not\Leftarrow \Theta(1)\text{ instability} $] 
\label{ex:noties-unstable}
Let $\cor$ denote a biased Copeland$_{0}$ rule for four alternatives,  which differs from Copeland$_{0}$ in that if $1\ra 2$, then alternative $1$ gets $2$ points (instead of $1$). Let $P'$ denote an arbitrary profile whose UMG is the same as Figure~\ref{fig:exstable} (b). For any $n'\in \mathbb N$, we let $P= n' P'$. Notice that $\cor(P) = \{2\}$. To see that $P$ is $\Theta(1)$ unstable, let $R$ denote any vote in $P$ where $2\succ 1$. Replace $R$ by $[1\succ 2\succ \others]$,   the winner becomes $1$.  

Let $P^*$ denote any  profile that is $\Theta(1)$ away from $P$. It is not hard to see that $\umg(P^*)$ contains the same edges as the graph in Figure~\ref{fig:exstable} (b) except the edge between $1$ and $2$. This means that either $\cor(P^*) = \{2\}$ (if there is no edge between $1$ and $2$ or there is an edge $2\ra 1$ in $\umg(P^*)$), or $\cor(P^*) = \{1\}$ (if $1\ra 2$ in $\umg(P^*)$). This means that $P$ is not close to any tied profile under $\cor$. 
\end{ex}

%
%
%

 
\end{document}

%% file: macro.tex
\usepackage[utf8]{inputenc}

\usepackage{caption}
\usepackage{algpseudocode}
\usepackage{booktabs} 
\usepackage[ruled]{algorithm2e} 
\usepackage{pifont}
\usepackage{nicefrac}
\usepackage{color}
\usepackage{xcolor,colortbl}
\usepackage{wrapfig}
\usepackage{graphicx}
\usepackage{amsmath,amsthm}
\usepackage{mathtools}
 
\usepackage{bm}
\usepackage{rotating}
\usepackage{multicol}
\usepackage{multirow}
\usepackage{titlesec}
\usepackage{latexsym}
\usepackage{enumitem}
\usepackage{hyperref}[backref=page]
\usepackage{float}
\hypersetup{
     colorlinks   = true,
     linkcolor    = red, 
     urlcolor     = blue, 
	 citecolor    = blue 
}

\newtheorem{thm}{Theorem}
\newtheorem{dfn}{Definition}

\newtheorem{lem}{Lemma}
\newtheorem{ex}{Example}

\newtheorem{prop}{Proposition}
\newtheorem{coro}{Corrollary}

\newtheorem{claim}{Claim}

\newcounter{newct}

\newtheorem{asmp}{Assumption}

\newcommand{\myhyperlink}[2]{\hyperlink{#1}{\color{red} #2}}
\newcommand{\myparagraph}[1]{\vspace{2mm}\noindent {\bf\boldmath #1}}

\newcommand{\bv}{\begin{array}}

\newcommand{\appLem}[3]{\vspace{1mm}\noindent{\bf Lemma~\ref{#2}.} {\bf (#1). } {\em #3}}
\newcommand{\appThm}[3]{\vspace{1mm}\noindent{\bf Theorem~\ref{#2}} {({\bf #1}). }{\em #3 \vspace{1mm}}}
\newcommand{\appDfn}[3]{\vspace{1mm}\noindent{\bf Definition~\ref{#2}} {({\bf #1}). }{\em #3 \vspace{1mm}}}

\newcommand{\cal}{\mathcal }
\newcommand{\ma}{\mathcal A}

\newcommand{\calP}{\mathcal P}
\newcommand{\vH}{{\vec H}}

\newcommand{\ml}{\mathcal L}
\newcommand{\mV}{\mathcal V}

\newcommand{\calQ}{\mathcal Q}

\newcommand{\calN}{\mathcal N}
\newcommand{\calR}{\mathcal R}
\newcommand{\calV}{\mathcal V}
\newcommand{\calB}{\mathcal B}
\newcommand{\calI}{\mathcal I}
\newcommand{\ra}{\rightarrow}

\newcommand{\cor}{{\overline r}} 

\newcommand{\rank}{\text{Rank}}
\newcommand{\sign}{\text{Sign}}
\newcommand{\signH}{\text{Sign}_{\vH}}

\newcommand{\others}{\text{others}}

\newcommand{\hist}{\text{Hist}}

\newcommand{\ind}{\text{Id}}

\newcommand{\fs}{{\cal S}_{\vec H}}

\newcommand{\sk}{{\cal S}_{K}}

\newcommand{\Omit}[1]{}

\newcommand{\calM}{\mathcal M}
\newcommand{\calA}{\mathcal A}

\newcommand{\ba}{{\mathbf A}}

\newcommand{\pba}[1]{{\mathbf A}^{#1}}

\newcommand{\vbb}{{{\mathbf b}}}
\newcommand{\pvbb}[1]{{{\mathbf b}}^{#1}}

\newcommand{\vXp}{{\vec X}_{\vec\pi}}

\newcommand{\piuni}{{\pi}_{\text{uni}}}
\newcommand{\avg}[1]{ \text{avg}(#1)}

\newcommand{\bd}{{\mathbf D}}

\newcommand{\conv}{\text{\rm CH}}
\newcommand{\cone}{\text{\rm Cone}}
\newcommand{\Span}{\text{\rm Span}}
\newcommand{\range}[2]{\text{width}_{#1}(#2)}

\newcommand{\csus}[1]{{\mathcal U}_{#1}}
\newcommand{\sus}[1]{{\mathcal C}_{#1}}

\newcommand{\wmg}{\text{WMG}}

\newcommand{\umg}{\text{UMG}}

\newcommand{\pair}{\text{Pair}}
\newcommand{\scorediff}[1]{\text{Score}^{\Delta}_{#1}}

\newcommand{\score}{\text{Score}}

\newcommand{\copeland}{\text{Cd}_\alpha}
\newcommand{\pcopeland}[1]{\text{Cd}_{#1}}

\newcommand{\borda}{\text{Borda}}
\newcommand{\veto}{\text{veto}}

\newcommand{\maximin}{\text{MM}}

\newcommand{\rp}{\text{RP}}

\newcommand{\schulze}{\text{Sch}}

\newcommand{\stv}{\text{STV}}

\newcommand{\plurunoff}{\text{Pro}}

\newcommand{\condition}[1]{{\text{\rm K}_{#1}}}

\newcommand{\poly}{{\mathcal H}}

\newcommand{\polyz}{{\mathcal H}_{\leq 0}}

\newcommand{\ppoly}[1]{{\mathcal H}^{#1}}

\newcommand{\ppolyz}[1]{{\mathcal H}_{\leq 0}^{#1}}

\newcommand{\cpoly}[1]{{\mathcal H}_{#1}}
\newcommand{\cpolyn}[1]{{\mathcal H}_{#1,n}}
\newcommand{\cpolynint}[1]{{\mathcal H}_{#1,n}^{\mathbb Z}}
\newcommand{\cpolyz}[1]{{\mathcal H}_{#1,\leq 0}}

\newcommand{\sourcepoly}{\cpoly{\text{S}}}
\newcommand{\targetpoly}{\cpoly{\text{T}}}
\newcommand{\sourcepolyz}{\cpolyz{\text{S}}}
\newcommand{\targetpolyz}{\cpolyz{\text{T}}}
\newcommand{\sourceprofile}{P_{\text{S}}}
\newcommand{\targetprofile}{P_{\text{T}}}

\newcommand{\invert}[1]{\left(#1\right)^\top}

\newcommand{\satmax}[2]{\widetilde{#1}_{#2}^{\max}}
\newcommand{\satmin}[2]{\widetilde{#1}_{#2}^{\min}}
\newcommand{\sat}[1]{{#1}}

\newcommand{\vo}{\vec o}
\newcommand{\vosize}{|\voset{}|}
\newcommand{\voset}[1]{{\mathbb O}_{#1}}
\newcommand{\vomatrix}[1]{{\mathbf O}_{#1}}

\newcommand{\vosc}{{\mathcal O}_{\pm}}

\newcommand{\vosetting}{{\cal S}}

\newcommand{\bB}{{\mathbf B}}

\newcommand{\multimatrix}[1]{\left[\begin{array}{c}#1\end{array}\right]}

\newcommand{\ccav}{\text{\sc CCAV}}
\newcommand{\ccdv}{\text{\sc CCDV}}
\newcommand{\dcav}{\text{\sc DCAV}}
\newcommand{\dcdv}{\text{\sc DCDV}}
\newcommand{\eccav}{\text{\sc e-CCAV}}
\newcommand{\eccdv}{\text{\sc e-CCDV}}
\newcommand{\edcav}{\text{\sc e-DCAV}}
\newcommand{\edcdv}{\text{\sc e-DCDV}}
\newcommand{\mov}{\text{\sc MoV}}
\newcommand{\cm}{\text{\sc CM}} 
\newcommand{\cml}{\text{\sc CML}} 
\newcommand{\cb}{\text{\sc CB}} 
\newcommand{\ecb}{\text{\sc e-CB}}
\newcommand{\db}{\text{\sc DB}} 
\newcommand{\edb}{\text{\sc e-DB}}

\newcommand{\control}{\text{\sc Control}} 
\newcommand{\econtrol}{\text{\sc e-Control}}

\newcommand{\nb}[2]{\text{\rm Nb}(#1,#2)}
\newcommand{\nbs}[2]{\text{\rm Nb}(#1,#2)}

\newcommand{\momatrix}[1]{{\mathbf K}_{#1}}